\documentclass[12pt,a4paper,reqno,oneside]{amsart}

\usepackage{graphicx}

\usepackage{mathptmx}      
\usepackage{times}
\usepackage{t1enc}
\usepackage{amsmath, amssymb}
\usepackage[mathscr]{euscript}
\usepackage{graphicx}
\usepackage{hyperref}   
\usepackage{microtype}

\DeclareMathAlphabet\scr{U}{scr}{m}{n}
\SetMathAlphabet\scr{bold}{U}{scr}{b}{n}
  \DeclareFontFamily{U}{scr}{\skewchar\font'177}%
  \DeclareFontShape{U}{scr}{m}{n}{<-6>rsfs5<6-8>rsfs7<8->rsfs10}{}%
  \DeclareFontShape{U}{scr}{b}{n}{<-6>rsfs5<6-8>rsfs7<8->rsfs10}{}%

\newcommand{\mal}{\stackrel{\mbox{\tiny$\bullet$}}{}}

\renewcommand{\Re}{\mathrm{Re}}
\renewcommand{\Im}{\mathrm{Im}}

\newcommand{\rr}{\mathbb R}
\newcommand{\rp}{\mathbb R _+}

\newcommand{\F}{\scr F}
\newcommand{\G}{\scr G}
\newcommand{\D}{\scr D}
\newcommand{\E}{\scr E}

\newcommand{\B}{\scr B}

\renewcommand{\epsilon}{\varepsilon}
\renewcommand{\theta}{\vartheta}
\renewcommand{\rho}{\varrho}

\newcommand{\essinf}{\mathrm{ess\ inf\ }}

\newcommand{\fraum}{(\Omega,\F,(\F_t)_{t\in\mathbb R_+},P)}
\newcommand{\osm}{(X,\Psi_0,\alpha,\beta,M)}
\newcommand{\osmzwei}{(\widetilde X,\widetilde \Psi_0,\widetilde \alpha,\widetilde \beta,\widetilde M)}
\newcommand{\raum}[1]{L^1(\mathbb R_+,#1)}

\newcommand{\ito}{It\^o}
\newcommand{\Xp}{X^\parallel}

\newcommand{\gamm}{\psi^{(\Xp,M)}}

\numberwithin{equation}{section} 
\numberwithin{figure}{section} 
  \theoremstyle{plain}
  \newtheorem{theorem}{Theorem}[section]
   \theoremstyle{plain}
  \newtheorem{proposition}[theorem]{Proposition}
 \theoremstyle{plain}
  
  \theoremstyle{definition}
  \newtheorem{definition}[theorem]{Definition}
  \theoremstyle{remark}
  \newtheorem{remark}[theorem]{Remark}
  \theoremstyle{remark}
  \newtheorem*{proofofmaintheorem}{Proof of Theorem \ref{Satz: Bedingungen impl. Vollstaendig}}

  \theoremstyle{plain}
  \newtheorem{lemma}[theorem]{Lemma}
  \theoremstyle{definition}
  \newtheorem{example*}{Example}
  \theoremstyle{definition}
  \newtheorem{subexample*}{Example}[example*]
  \theoremstyle{definition}
  \newtheorem{example}[theorem]{Example}
  \theoremstyle{plain}
  \newtheorem{corollary}[theorem]{Corollary}
  \theoremstyle{plain}

\begin{document}

\title{On a Heath-Jarrow-Morton approach for stock options}

\author{Jan Kallsen  \and Paul Kr\"uhner}
\address{	
	      Department of Mathematics,
	      Kiel University,
	      Westring 383,
	      24118 Kiel, 
	      Germany,
              (e-mail: kallsen@math.uni-kiel.de)}        
  \address{  
	     Department of Mathematics,
             University of Oslo,
             P.O. Box 1053, Blindern,
            0316 Oslo,
            Norway,
            (e-mail: paulkru@math.uio.no) 
            }

\begin{abstract}
This paper aims at transferring the philosophy behind Heath-Jarrow-Mor\-ton to the modelling of call 
options with all strikes and maturities. Contrary to the approach by 
Carmona and Nadtochiy \cite{carmona.nadtochiy.09} and related to the recent contribution 
\cite{carmona.nadtochiy.12} by the same authors, the key parametrisation of our approach 
involves time-inhomoge\-neous L\'evy processes instead of local volatility models. We provide 
necessary and sufficient conditions for absence of arbitrage. Moreover we discuss the construction 
 of arbitrage-free models. Specifically, we prove their existence and uniqueness given basic building blocks.\\



\end{abstract}

\maketitle

\section{Introduction}
The traditional approach to modelling stock options takes the underlying as a starting point. 
If the dynamics of the stock are specified under a risk neutral measure for the whole market 
(i.e.\ all discounted asset price processes are martingales), then options prices are obtained as 
conditional expectations of their payoff. In reality, standard options as calls and puts are liquidly traded. 
If one wants to obtain vanilla option prices which are consistent with observed market values, 
special care has to be taken. A common and also theoretically reasonable way is calibration,  
i.e.\ to choose the parameters for the stock dynamics such that the model approximates market 
values sufficiently well. After a while, models typically have to be recalibrated, i.e.\ different parameters
 have to be chosen in order for model prices to be still consistent with market quotes. However, 
frequent recalibration is unsatisfactory from a theoretical point of view because model parameters 
are meant to be deterministic and 
constant. Its necessity indicates that the chosen class fails to describe the market consistently.

In Markovian factor models with additional unobservable state variables,
the situation is slightly more involved. Since these state variables are randomly changing
within the model, they may be recalibrated, which means that their current values are inferred from
option prices. In practice, however, the model parameters are often recalibrated as well
because the few state variables do not provide enough flexibility to match observed option data.
In this case, we are facing the same theoretically unsatisfactory situation as above.

A possible way out is to model the whole surface of call options as a state variable, i.e.\ 
as a family of primary assets in their own right. This alternative perspective is motivated 
from the Heath-Jarrow-Morton (HJM, see \cite{heath.al.92}) approach in interest rate theory. 
Rather than considering bonds as derivatives on the short rate, HJM treat the whole family 
of zero bonds or equivalently the forward rate curve as state variable in the first place. 
In the context of HJM-type approaches for stock options, Wissel \cite{schweizer.wissel.08} and 
Sch\"onbucher \cite{schoenbucher.05} consider the case of a single strike, whereas 
Cont et al.\ \cite{cont.al.02} and Carmona and Nadtochiy \cite{carmona.nadtochiy.09,carmona.nadtochiy.12} 
allow for all strikes and maturities. Further important references in this context include Jacod and 
Protter \cite{jacod.protter.06}, Schweizer and Wissel \cite{schweizer.wissel.08b} and 
Wissel \cite{wissel.08}. The HJM approach has been adapted to other asset classes, e.g.\ 
credit models in Benanni \cite{bennani.05}, Sch\"onbucher \cite{schoenbucher.05}, or 
Sidenius et al.\ \cite{sidenius.al.08} and variance swaps in B\"uhler \cite{buehler.06}, cf.\ 
Carmona \cite{carmona.09} for an overview and further references.

Similar to Carmona and Nadtochiy \cite{carmona.nadtochiy.09} we aim at modelling the 
whole call option price surface using the HJM methodology. However, our approach differs 
in the choice of the parametrisation or {\em codebook}, which constitutes a crucial step in 
HJM-type setups. By relying on time-in\-homo\-geneous L\'evy processes rather than 
Dupire's local volatility models, we can avoid some intrinsic difficulties of the framework in 
\cite{carmona.nadtochiy.09}. 
E.g., a simpler drift condition makes the approach more amenable to existence and 
uniqueness results. Moreover, the L\'evy-based setup allows for jumps and may hence 
be more suitable to account particularly for short-term option prices,
cf.\ \cite[Section 1.2]{cont.tankov.04}.

More recently and independently of the present study, Carmona and Nadtochiy \cite{carmona.nadtochiy.12} 
have also put forward a HJM-type approach for the option price surface which is based on 
time-inhomogeneous L\'evy processes. The similarities and differences of their and our 
approach are discussed in Section \ref{Abschnitt: CNs Modell}.

The paper is arranged as follows. We start in Section \ref{Abschnitt: HJM philosophie} 
with an informal discussion of the HJM philosophy, as a motivation to its application to stock options. 
Section \ref{Abschnitt: Notationen und Modelvorbereitung} provides necessary and sufficient 
conditions for an option surface model to be arbitrage-free or, more precisely, risk-neutral. 
Subsequently, we turn to existence and uniqueness of option surface models given basic building blocks. 
In particular, we provide a concrete example which turns out to be related to the 
stochastic volatility model proposed by Barndorff-Nielsen and Shephard in \cite{barndorff.shephard.01}. 
Mathematical tools and some technical proofs are relegated to the appendix.
Facts on seminartingale characteristics are summarised in Section \ref{s:localchar}. The subsequent section
concerns mainly option pricing by Fourier transform.
In Section \ref{s:frechetsde} we consider stochastic differential equations in Fr\'echet spaces,
driven by subordinators. This framework is needed for existence and uniqueness results in 
Section \ref{Abschnitt: Beispiele und Existenzresultate}.

\subsection*{Notation}
$\Re$ and $\Im$ denote the real resp.\ imaginary part of a complex vector in $\mathbb C^d$. 
We write $[a,b]$ for the closed interval $\{x\in\mathbb R:a\leq x\leq b\}$, which is empty if $a>b$. 
We use the notations $\partial_u$ and $D$ for partial and total derivatives, respectively. 
We often write $\beta\mal X_t=\int_0^t\beta_sdX_s$ for stochastic integrals. 
$L(X)$ denotes the set of $X$-integrable predictable processes for a semimartingale $X$. 
If we talk about an $m+n$-dimensional semimartingale $(X,Y)$, we mean that 
$X$ is an $\mathbb R^m$-valued semimartingale and $Y$ is an $\mathbb R^n$-valued 
semimartingale. For $u,v\in\mathbb C^d$ we denote the bilinear form of $u$ and $v$ by 
$uv:=\sum_{k=1}^du_k v_k$. The abbreviation {\em PII} stands for processes 
with independent increments in the sense of \cite{js.87}. Further unexplained 
notation is used as in \cite{js.87}.

\section{Heath-Jarrow-Morton and L\'evy models}\label{Abschnitt: HJM philosophie}
This section provides an informal discussion of the HJM philosophy and its application to stock options.

\subsection{The Heath-Jarrow-Morton philosophy}\label{Abschnitt2: HJM philosophie}
 According to the fundamental theorem of asset pricing, there exists at least one equivalent 
probability measure that turns discounted prices of all traded securities into martingales or, 
more precisely, into $\sigma$-martingales. For simplicity we take the point of view of 
risk-neutral modelling in this paper, i.e.\ we specify the dynamics of all assets in the market 
directly under such an equivalent martingale measure (EMM). Moreover, we assume the 
existence of a deterministic bank account
unless we refer to the original HJM setup in interest rate theory. 
This allows us to express all prices easily in discounted terms.

Before we turn to our concrete setup, we want to highlight key features of the HJM approach in general. 
For more background and examples we refer the reader to the brilliant exposition from 
Carmona \cite{carmona.09}, which greatly inspired our work. We proceed by stating seven 
informal axioms or steps. 

\begin{enumerate}
 \item In HJM-type setups there typically exists a canonical underlying asset or reference process, 
namely the money market account in interest rate theory or the stock in the present paper. 
The object of interest, namely bonds in interest rate theory or vanilla options in stock markets, 
can be interpreted as derivatives on the canonical process. HJM-type approaches typically focus 
on a whole manifold of such --- at least in theory --- liquidly traded derivatives, e.g.\ the 
one-parameter manifold of bonds with all maturities or the two-parameter manifold of 
call options with all strikes and maturities. As first and probably most important HJM axiom 
we claim that this manifold of liquid derivatives is to be treated as the set of primary assets. 
It --- rather than the canonical reference asset --- constitutes the object whose dynamics should 
be modelled in the first place.
\begin{example*}
Zero bonds are securities with terminal value $1$ and appear to be somewhat degenerate
derivatives. But as noted above, we consider discounted prices in this paper.
If the money market account $S^0$ is chosen as numeraire, the discounted payoff of a
bond maturing at $T$ is of the form $1/S^0_T$ and hence a function of $S^0$.
In this broader sense, we view bonds here as derivatives on the money market account.
\end{example*}
\begin{example*}
European call options on a stock $S$ with maturity $T$ and strike $K$
have a payoff of the form $(S_T-K)^+$.
The same is true for their discounted payoff  relative to a  deterministic numeraire,
provided that $S$ and $K$ are replaced by their discounted analogues as well.
\end{example*}
\item The first axiom immediately leads to the second one: do {\em not} model the 
canonical reference asset in detail under the market's risk-neutral measure. Indeed, 
otherwise all derivative prices would be entirely determined by their martingale property, 
leaving no room for a specification of their dynamics.
\item Direct modelling of the above manifold typically leads to awkward constraints. 
Zero bond price processes must terminate in 1, vanilla options in their respective payoff. 
Rather than prices themselves one should therefore consider a convenient parametrisation 
(or {\em codebook} in the language of Carmona \cite{carmona.09}), e.g.\ 
instantaneous forward rates in interest rate theory. Specifying the dynamics 
of this codebook leads immediately to a model for the manifold of primary assets.
If the codebook is properly chosen, then static arbitrage constraints are satisfied
automatically, cf.\ Steps \ref{i:4} and  \ref{i:5}.
\item\label{i:4} It is generally understood that choosing a convenient parametrisation 
constitutes a crucial step for a successful HJM-type approach. This is particularly obvious 
in the context of call options. Their prices are linked by a number of non-trivial static arbitrage 
constraints, which must hold independently of any particular model, cf.\ 
Davis and Hobson \cite{davis.hobson.07}. These static constraints have to be respected 
by any codebook dynamics. Specifying the latter properly may therefore be a difficult 
task unless the codebook is chosen such that the constraints naturally hold. We now 
suggest a way how to come up with a reasonable parametrisation.

The starting point is a family of simple risk-neutral models for the canonical underlying 
whose parameter space has --- loosely speaking ---  the same ``dimension'' or ``size'' 
as the space of liquid derivative manifolds. Provided sufficient regularity holds, the 
presently observed manifold of derivative prices is explained by one and only one of these models.
 \addtocounter{example*}{-2}
\begin{example*}\label{Bsp: HJM}
 In interest rate theory consider bank accounts of the form 
\begin{equation}\label{e:bankkonto}
 S^0_t=\exp\bigg(\int_0^tr(s)ds\bigg)
\end{equation}
with deterministic short rate $r(T),T\in\mathbb R_+$. Fix $t\in\mathbb R_+$ and a 
differentiable curve of bond prices $B(t,T)$, $T\in\mathbb R_+$. Except for the past 
up to time $t$, the observed bond prices are consistent with one and only one of these models, namely for
\begin{eqnarray}
 r(T):=-\partial_T\log(B(t,T)).\label{Gleichung: HJM Ausgangspunkt}
\end{eqnarray}
\end{example*}
\begin{example*}
 Consider Dupire's local volatility models $$dS_t=S_t\sigma(S_t,t)dW_t$$ for a discounted stock, 
where $W$ denotes standard Brownian motion and $\sigma:\mathbb R_+^2\rightarrow\mathbb R$ 
a deterministic function. Up to regularity, any surface of discounted call option prices $C_t(T,K)$ 
with varying maturity $T$ and strike $K$ and fixed current date $t\in\mathbb R_+$ is obtained 
by one and only one local volatility function $\sigma$, namely by
\begin{eqnarray}
 \sigma^2(K,T):=\frac{2\partial_TC_t(T,K)}{K^2\partial_{KK}C_t(T,K)}.\label{Gleichung: Carmona Ausgangspunkt}
\end{eqnarray}
\end{example*}
Note that the above starting point should not be taken as a precise mathematical requirement.
As illustrated in the examples, we relate ``size'' liberally to the number of arguments in the 
respective functions: parameter curves correspond to bond price curves,
parameter surfaces to option price surfaces. The actual regularity needed for one-to-one 
correspondence crucially depends on the chosen family of simple models.

If market data follows the simple model, the parameter manifold, e.g.\ $(r(T))_{T\geq t}$ in 
Example \ref{Bsp: HJM}, is deterministic and does not depend on the time $t$ when derivative 
prices are observed. Generally, however, market data does not follow such a simple model as in 
the two examples. Hence, evaluation of the right-hand side of \eqref{Gleichung: HJM Ausgangspunkt} 
and \eqref{Gleichung: Carmona Ausgangspunkt} leads to a parameter manifold which changes randomly over time.
 \addtocounter{example*}{-2}
\begin{example*}
The {\em instantaneous forward rate} curve
 $$f(t,T):=-\partial_T\log(B(t,T)), \quad T\geq t$$
  for fixed $t\in\mathbb R_+$
can be interpreted as the family of deterministic short rates that is consistent with the presently 
observed bond price curve $B(t,T)$, ${T\geq t}$.
\end{example*}
\begin{example*}
 The {\em implied local volatility}
\begin{eqnarray*}
 \sigma^2(K,T):=\frac{2\partial_TC_t(T,K)}{K^2\partial_{KK}C_t(T,K)},\quad K>0,\quad T\geq t
\end{eqnarray*}
for fixed $t\in\mathbb R_+$ can be interpreted as the unique local volatility function that is consistent with the 
presently observed discounted call prices $C_t(T,K)$,
$T\ge t$, $K>0$.
\end{example*}
The idea now is to take this present parameter manifold as a parametrisation or codebook for the 
manifold of derivatives.
\item\label{i:5} In a next step, we set this parameter manifold ``in motion.'' We consider the codebook, 
e.g.\ the instantaneous forward rate curve $f(t,T)$ or the implied local volatility $\sigma_t(T,K)$, 
as an infinite-dimensional stochastic process. It is typically modelled by a stochastic differential equation, e.g.
$$df(t,T)=\alpha(t,T)dt+\beta(t,T)dW_t,$$
where $W$ denotes standard Brownian motion. As long as the solution to this equation moves 
within the parameter space for the family of simple models, one automatically obtains derivative 
prices that satisfy any static arbitrage constraints. Indeed, since the current bond prices resp.\ call 
prices coincide with the prices from an arbitrage-free model, they cannot violate any such constraints, 
however complicated they might be. This automatic absence of static arbitrage motivates the codebook choice in Step~4.
 \item Absence of static arbitrage does not imply absence of arbitrage altogether. Under the risk-neutral 
modelling paradigm, all discounted assets must be martingales. In interest rate theory this leads 
to the well known HJM drift condition. More generally it means that the drift part of the codebook 
dynamics of Step 5 is determined by its diffusive component.
\item Finally we come back to Step 2. The dynamics of the canonical reference asset process is 
typically implied by the current state of the codebook. E.g. in interest rate theory the short rate is determined 
by the so-called consistency condition
$$r(t)=f(t,t).$$
Similar conditions determine the current stock volatility in 
\cite{schweizer.wissel.08,wissel.08,carmona.nadtochiy.09,carmona.nadtochiy.12}.
\end{enumerate}

\subsection{Time-inhomogeneous L\'evy models}\label{Abschnitt: ZeitinhomogeneLevys}
According to the above interpretation, the approach of \cite{carmona.nadtochiy.09} to 
option surface modelling relies on the family of Dupire's local volatility models. 
Similarly as the independent study \cite{carmona.nadtochiy.12}, we suggest 
another family of simple models for the stock, also relying on a two-parameter manifold. 
To this end, suppose that the discounted stock is a martingale of the form $S=e^X$, 
where the {\em return process} $X$ denotes a process with independent increments 
(or time-inhomogeneous L\'evy process, henceforth PII) on some filtered probability space 
$(\Omega,\F,(\F_{t})_{t\in\mathbb R_+},P)$. Recall that we work with risk-neutral probabilities, 
i.e.\ discounted asset prices are supposed to be $P$-martingales. More specifically, the 
characteristic function of $X$ is assumed to be absolutely continuous in time, i.e. 
\begin{equation}\label{Gleichung: 2.3}
E(e^{iuX_t})=\exp\bigg(iuX_0+\int_0^t\Psi(s,u)ds\bigg)
\end{equation}
with some function $\Psi:\mathbb R_+\times\mathbb R\rightarrow \mathbb C$. 

We assume call options of all strikes and maturities to be liquidly traded.
Specifically, we write $C_t(T,K)$ for the discounted price at time $t$ of a call which 
expires at $T$ with discounted strike $K$.
A slight extension of \cite[Proposition 1]{belomestny.reiss.05} shows that option 
prices can be expressed in terms of $\Psi$. To this end, we define {\em modified option prices}
$$\scr O_t(T,x):=e^{-(x+X_t)}C_t(T,e^{x+X_t})-(e^{-x}-1)^+.$$
Since call option prices are obtained from $C_t(T,K)=E((S_T-K)^+\vert\F_t)$, 
by call-put parity, and by $E(S_T\vert\F_t)=S_t$, we have
$$\scr O_t(T,x)=\bigg\{\begin{array}{cc}
E((e^{(X_T-X_t)-x}-1)^+\vert\F_t) & \mathrm{if}\ x\geq0,\\
E((1-e^{(X_T-X_t)-x})^+\vert\F_t) & \mathrm{if}\ x<0.\end{array}$$
Proposition \ref{Proposition: Fourier transformierte der Call-Preise} yields
\begin{eqnarray}
 \scr O_t(T,x)&=&\F^{-1}\bigg\{u\mapsto\frac{1-E(e^{iu(X_T-X_t)}\vert\F_t)}{u^2+iu}\bigg\}(x),
\label{Gleichung: Levy-mod-optionen} \\
 \F\{x\mapsto\scr O_t(T,x)\}(u)&=&\frac{1-E(e^{iu(X_T-X_t)}\vert\F_t)}{u^2+iu}
\end{eqnarray}
where $\F^{-1}$ and $\F$ denote the improper inverse Fourier transform and the improper Fourier transform, 
respectively, in the sense of (\ref{Gleichung: Fourier Transformation}, \ref{Gleichung: Inverse Fourier Transformation}) 
in Section \ref{Abschnitt: Fourier transformation und Optionspreise} of the appendix.
Since 
\begin{equation}
 C_t(T,K) = (S_t-K)^++K\scr O_t\bigg(T,\log\frac{K}{S_t}\bigg)\label{Gleichung: Levy-Optionen}
\end{equation}
and
\begin{equation}
 E(e^{iu(X_T-X_t)}\vert \F_t) = \exp\bigg(\int_t^T\Psi(s,u)ds\bigg),\label{Gleichung Levy char. exponent}
\end{equation}
we can compute option prices according to the following diagram:
$$\Psi \rightarrow \exp\bigg(\int_t^T\Psi(s,\cdot)ds\bigg) \rightarrow \scr O_t(T,\cdot) \rightarrow C_t(T,\cdot).$$
For the last step we also need the present stock price $S_t$. Under sufficient smoothness 
we can invert all transformations. Indeed, we have
\begin{eqnarray}
 \Psi(T,u)=\partial_T\log\bigg(1-(u^2+iu)\F\{x\mapsto\scr O_t(T,x)\}(u)\bigg)\label{Gleichung: Idee fuer Psi}.
\end{eqnarray}
Hence we obtain option prices from $\Psi$ and vice versa as long as we know the present stock price.

\subsection{Setting L\'evy in motion}\label{Abschnitt: Levy in Bewegung}
Generally we do not assume that the {\em return process}
\begin{eqnarray}
X:=\log(S)\label{Gleichung: X} 
\end{eqnarray}
follows a time-inhomoge\-neous L\'evy process. Hence the right-hand side of 
Equation \eqref{Gleichung: Idee fuer Psi} will typically change randomly over time. 
In line with Step \ref{i:4} above, we define modified option prices
\begin{eqnarray}
 \scr O_t(T,x) & := & e^{-(x+X_t)}C_t(T,e^{x+X_t})-(e^{-x}-1)^+\label{Gleichung: O ist gleich}
\end{eqnarray} as before and
\begin{eqnarray}
 \Psi_t(T,u) & := & \partial_T\log\bigg(1-(u^2+iu)\F\{x\mapsto\scr O_t(T,x)\}(u)\bigg).\label{Gleichung: Psi ist gleich}
\end{eqnarray}
This constitutes our {\em codebook process} for the surface of discounted option prices. 
As in Section \ref{Abschnitt: ZeitinhomogeneLevys} the asset price processes $S$ and $C(T,K)$ 
can be recovered from $X$ and $\Psi(T,u)$ via
\begin{eqnarray*}
 S & = & \exp(X),\\
 \scr O_t(T,x) & = & \F^{-1}\bigg\{u\mapsto\frac{1-\exp(\int_t^T\Psi_t(s,u)ds)}{u^2+iu}\bigg\}(x),\\
C_t(T,K) & = & (S_t-K)^++K\scr O_t\bigg(T,\log\frac{K}{S_t}\bigg).
\end{eqnarray*}

In the remainder of this paper we assume that the infinite-dimensional codebook process 
satisfies an equation of the form
\begin{eqnarray}\label{Gleichung:neu}
d\Psi_t(T,u)=\alpha_t(T,u)dt+\beta_t(T,u)dM_t,
\end{eqnarray}
driven by some finite-dimensional semimartingale $M$.

\section{Model setup and risk neutrality}\label{Abschnitt: Notationen und Modelvorbereitung}
As before we fix a filtered probability space $(\Omega,\F,(\F_t)_{t\in\mathbb R_+},P)$ 
with trivial initial $\sigma$-field $\F_0$. In this section we single out conditions such that a given pair 
$(X,\Psi)$ corresponds via (\ref{Gleichung: X} -- \ref{Gleichung: Psi ist gleich}) to a 
risk-neutral model for the stock and its call options.

\subsection{Option surface models}\label{Abschnitt: Optionsfl\"achenmodelle}
We denote by $\Pi$ the set of characteristic exponents of L\'evy processes $L$ such that 
$E(e^{L_1})=1$ . More precisely, $\Pi$ contains all functions $\psi:\mathbb R\rightarrow\mathbb C$ of the form
$$\psi(u)=-\frac{u^2+iu}{2}c+\int(e^{iux}-1-iu(e^x-1))K(dx),$$
where $c\in\mathbb R_+$ and $K$ denotes a L\'evy measure on $\mathbb R$ satisfying 
$\int_{\{x>1\}}e^xK(dx)<\infty$. 
\begin{definition}\label{Def: osm} 
A quintuple $(X,\Psi_0,\alpha,\beta,M)$ is an {\em option surface model} if
\begin{itemize}
\item $(X,M)$ is a $1+d$-dimensional semimartingale that allows for local characteristics 
in the sense of Section \ref{su:lc},
\item $\Psi_0:\mathbb R_+\times\mathbb R\rightarrow\mathbb C$ is measurable with
$\int_0^T\vert\Psi_0(r,u)\vert dr<\infty$ for any $T\in\mathbb R_+$, $u\in\mathbb R,$
\item $\alpha(T,u),\beta(T,u)$ are $\mathbb C$- resp.\ $\mathbb C^d$-valued 
processes for any $T\in\mathbb R_+$, $u\in\mathbb R$,
\item $(\omega,t,T,u)\mapsto\alpha_t(T,u)(\omega),\beta_t(T,u)(\omega)$ are 
$\scr P\otimes\B(\mathbb R_+)\otimes\B$-measurable, where $\scr P$ denotes the 
predictable $\sigma$-field on $\Omega\times\mathbb R_+$,
\item $\int_0^t\int_0^T\vert\alpha_s(r,u)\vert drds<\infty$ for any $t,T\in\mathbb R_+$, $u\in\mathbb R$,
\item $\int_0^T|\beta_t(r,u)|^2dr<\infty$ for any $t,T\in\mathbb R_+$, $u\in\mathbb R$,
\item $((\int_0^T|\beta^j_{t}(r,u)|^2dr)^{1/2})_{t\in\mathbb R_+}\in L(M^j)$ 
for any fixed $T\in\mathbb R_+$, $u\in\mathbb R$, $j\in\{1,\dots,d\}$, where the set $L(M^j)$ of $M^j$-integrable 
processes is defined as in \cite[Definition III.6.17]{js.87},
\item a version of the corresponding {\em codebook process}
\begin{eqnarray}\label{Gleichung:neu2}
\Psi_t(T,u):=\Psi_0(T,u)+\int_0^{t\wedge T}\alpha_s(T,u)ds+\int_0^{t\wedge T}\beta_s(T,u)dM_s
\end{eqnarray}
has the following properties:
\begin{enumerate}
\item $(\omega,t,T,u)\mapsto\Psi_t(T,u)(\omega)$ is $\scr A\otimes\B(\mathbb R_+)\otimes\B$-measurable, 
where $\scr A$ denotes the optional $\sigma$-field on $\Omega\times\mathbb R_+$,
\item $u\mapsto\int_t^T\Psi_s(r,u)(\omega)dr$ is in $\Pi$ for any $T\in\mathbb R_+$, $t\in[0,T]$, 
$s\in[0,t]$, $\omega\in\Omega$.
\end{enumerate}
\end{itemize}
\end{definition}

\begin{remark}
The square-integrability conditions on $\beta$ are imposed only to warrant
\begin{eqnarray}
 \int_0^T\left\vert\int_0^t\beta_s(r,u)dM_s\right\vert dr &<&\infty\quad\mbox{and}\label{e:intbeta1}\\
 \int_0^t\int_0^T\beta_s(r,u)drdM_s &=&\int_0^T\int_0^t\beta_s(r,u)dM_sdr. \label{e:intbeta2}
\end{eqnarray}
If $M$ has increasing components, we can and do replace the integrability conditions on $\beta$ by
the weaker requirement
\begin{itemize}
\item $\sum_{j=1}^d\int_0^t\int_0^T\vert\beta^j_s(r,u)\vert drdM^j_s<\infty$ 
for any $t,T\in\mathbb R_+$, $u\in\mathbb R$,
\end{itemize}
which implies (\ref{e:intbeta1}, \ref{e:intbeta2}) by Fubini's theorem.
\end{remark}

We denote the local exponents of $(X,M), X$ by $\psi^{(X,M)}, \psi^X$ and their 
domains by $\scr U^{(X,M)}, \scr U^X$, cf.\ Definitions \ref{Definition: lokaler Exponent 1. Version} 
and \ref{Definition: lokaler Exponent}.
In line with Section \ref{Abschnitt: Levy in Bewegung}, the discounted stock and call price processes 
associated with an option surface model are defined by
\begin{eqnarray}
 S_t & := & \exp(X_t),\\
 \scr O_t(T,x) & := & \F^{-1}\left\{u\mapsto\frac{1-\exp\big(\int_t^T\Psi_t(r,u)dr\big)}{u^2+iu}\right\}(x),
\label{Gleichung: O als Fourier Transformierte}\\
C_t(T,K) & := & (S_t-K)^++K\scr O_t\bigg(T,\log\frac{K}{S_t}\bigg)\label{Gleichung: Optionen sind}
\end{eqnarray}
for any $T\in\mathbb R_+$, $t\in[0,T]$, $x\in\mathbb R$, $K\in\mathbb R_+$, 
where $\scr F^{-1}$ denotes the improper inverse Fourier transform in the sense of 
Section \ref{Abschnitt: Fourier transformation und Optionspreise}.
From (\ref{Gleichung:neu2}) it follows that 
$\Psi_t(T,u)=\Psi_T(T,u)$ for $T<t$.
By (\ref{Gleichung: O als Fourier Transformierte}, \ref{Gleichung: Optionen sind}) 
this part $\Psi_t(T,u), T<t$ of the codebook does not affect option prices
and is hence irrelevant. 

\begin{remark}\label{Bemerkung: Vergleich zu PII}
The existence of these processes is implied by the assumptions above. Indeed, by 
Fubini's theorem for ordinary and stochastic integrals \cite[Theorem IV.65]{protter.04}, we have
$$\int_0^T\vert\Psi_t(r,u)\vert dr<\infty.$$
Fix $\omega\in\Omega$. Since $u\mapsto\int_t^T\Psi_t(r,u)(\omega)dr\in\Pi$, there is a 
random variable $Y$ on some probability space $(\widetilde\Omega,\widetilde{\scr F},\widetilde P)$ 
which has infinitely divisible distribution and characteristic function
$$\widetilde E(e^{iuY})=\exp\left(\int_t^T\Psi_t(r,u)(\omega)dr\right),\quad u\in\mathbb R.$$
Since this function is in $\Pi$, we have $\widetilde E(e^Y)=1$. Thus 
Proposition \ref{Proposition: Fourier transformierte der Call-Preise} yields the existence of the inverse 
Fourier transform in Equation \eqref{Gleichung: O als Fourier Transformierte}. 
Moreover, it implies $C_t(T,K)(\omega)=\widetilde E((S_t(\omega)e^Y-K)^+)$ and 
thus we have $0\leq C_t(T,K)(\omega)\leq S_t(\omega)$ and $0\leq P_t(T,K)(\omega)\leq K$ 
for any $K\in\mathbb R_+$, $T\in\mathbb R_+$, $t\in[0,T]$, where 
$P_t(T,K):=C_t(T,K)+K-S_t$ for any $K\in\mathbb R_+$, $T\in\mathbb R_+$, $t\in[0,T]$.
\end{remark}

As noted above, we model asset prices under a risk-neutral measure for the whole market. 
Put differently, we are interested in risk-neutral option surface models in the following sense.
\begin{definition}
 An option surface model $\osm$ is called {\em risk neutral} if the corresponding stock $S$ 
and all European call options $C(T,K)$, $T\in\mathbb R_+$, $K>0$ are $\sigma$-martingales or, 
equivalently, local martingales (cf.\ \cite[Proposition 3.1 and Corollary 3.1]{kallsen.03}). 
It is called {\em strongly risk neutral} if $S$ and all $C(T,K)$ are martingales.
\end{definition}

Below, risk-neutral option surface models are characterized in terms of the following properties.
\begin{definition} An option surface model $\osm$ satisfies the {\em consistency condition} if 
$$\psi^X_t(u)=\Psi_{t-}(t,u),\quad u\in \mathbb R$$ outside some $dP\otimes dt$-null set.
Moreover, it satisfies the {\em drift condition} if 
$$\bigg(u,-i\int_t^T\beta_t(r,u)dr\bigg)_{t\in\mathbb R_+}\in \scr U^{(X,M)}$$ 
and
\begin{eqnarray}\label{Gleichung: Driftbedingung}
\int_t^T\alpha_t(r,u)dr=\psi^X_t(u)-\psi_t^{(X,M)}\bigg(u,-i\int_t^T\beta_t(r,u)dr\bigg)
\end{eqnarray}
outside some $dP\otimes dt$-null set for any $T\in\mathbb R_+$, $u\in\mathbb R$.
Finally, the option surface model satisfies the {\em conditional expectation condition} if
$$\exp\left(\int_t^T\Psi_t(r,u)dr\right)=E(e^{iu(X_T-X_t)}\vert\F_t)$$ for any $T\in\mathbb R_+$, 
$t\in[0,T]$, $u\in\mathbb R$.
\end{definition}
\begin{remark}\label{Bemerkung: abgeleitete Driftbedingung}
 The drift condition can be rewritten as
 \begin{eqnarray}\label{G:Driftbedingung, einf.}
\alpha_t(T,u)=-\partial_T\left(\psi_t^{(X,M)}\bigg(u,-i\int_t^T\beta_t(r,u)dr\bigg)\right)  
 \end{eqnarray}
for almost all $T\in\mathbb R_+$. It gets even simpler if $X$ and $M$ are assumed to be 
locally independent in the sense of Definition \ref{Definition: lokale unabh}:
 \begin{eqnarray}\label{G:Driftbedingung, ganz einf.}
\alpha_t(T,u)=-\partial_T\left(\psi_t^M\bigg(-i\int_t^T\beta_t(r,u)dr\bigg)\right).  
 \end{eqnarray}
 If the derivative $\psi_t^\prime(u):=\partial_u\psi_t^M(u)$ exists as well, the drift condition 
simplifies once more and turns into
 $$\alpha_t(T,u)=i\psi_t^\prime\bigg(-i\int_t^T\beta_t(r,u)dr\bigg)\beta_t(T,u).$$
 Now consider the situation that $M$ is a one-dimensional Brownian motion which is 
locally independent of the return process $X$. Then $\psi^M(u)=-u^2/2$ and the drift condition reads as
 $$\alpha_t(T,u)=-\beta_t(T,u)\int_t^T\beta_t(r,u)dr.$$
Thus the drift condition for option surface models is similar to the HJM drift condition (cf.\ \cite{heath.al.92}).
\end{remark}

Drift condition \eqref{Gleichung: Driftbedingung} seems to rely on the joint exponent of $X$ and $M$. 
However, \eqref{G:Driftbedingung, ganz einf.} suggests that only partial knowledge about $X$ is needed. 
It is in fact sufficient to specify the joint exponent $\psi^{(\Xp,M)}$ of $M$ and the dependent part $\Xp$ of 
$X$ relative to $M$, which is defined in Section \ref{Abschnitt: Semimartingalprojektion}. Using this notion, 
Equation \eqref{G:Driftbedingung, einf.} can be rewritten as
\begin{equation}\label{e:letztedrift}
\alpha_t(T,u)=-\partial_T\left(\psi_t^{(\Xp,M)}\bigg(u,-i\int_t^T\beta_t(r,u)dr\bigg)\right)
\end{equation}
because $\psi^{(X,M)} = \psi^{(X^\perp,0)}+\psi^{(\Xp,M)}$ and the first summand on the 
right-hand side does not depend on its second argument.

\subsection{Necessary and sufficient conditions}\label{Unterabschnitt: Notwendige Bedingungen}
The goal of this section is to prove the following characterisation of risk-neutral option surface models.
\begin{theorem}\label{Satz: Bedingungen impl. Vollstaendig}
For any option surface model $\osm$ the following statements are equivalent.
\begin{enumerate}
\item It is strongly risk neutral.
\item It is risk neutral.
\item It satisfies the conditional expectation condition.
\item It satisfies the consistency and drift conditions.
\end{enumerate}
\end{theorem}
The remainder of this section is devoted to the proof of Theorem \ref{Satz: Bedingungen impl. Vollstaendig}. 
We proceed according to the following scheme
$$(1)\Rightarrow(2)\Rightarrow(3)\Rightarrow(4)\Rightarrow(3)\Rightarrow(1).$$
We use the notation
\begin{eqnarray*}
\delta_t(T,u) & := & \int_t^T\alpha_t(r,u)dr-\psi_t^X(u),\\
\sigma_t(T,u) & := & \int_t^T\beta_t(r,u)dr,\\
\Gamma_t(T,u) & := & \int_0^T\Psi_0(r,u)dr+\int_0^t\delta_s(T,u)ds+\int_0^t\sigma_s(T,u)dM_s.
\end{eqnarray*}
The existence of the integrals above is implied by the condition for option surface models. 
Observe that $\Gamma(T,u)$ is a semimartingale.

\begin{lemma}\label{Lemma: Gammazerlegung}
For any $u\in\mathbb R,T\in\mathbb R_+,t\in[0,T]$ we have
 $$\Gamma_t(T,u)-\Gamma_t(t,u)=\int_t^T\Psi_t(r,u)dr.$$
\end{lemma}
\begin{proof}
 Using the definition of $\Gamma,\delta,\sigma$ and applying Fubini's theorem as in \cite[Theorem IV.65]{protter.04} 
yields
\begin{eqnarray*}
\Gamma_t(T,u)-\Gamma_t(t,u) 
& = & \int_t^T\Psi_0(r,u)dr+\int_t^T\int_0^t\alpha_s(r,u)dsdr\\
&&{}+\int_t^T\int_0^t\beta_s(r,u)dM_sdr\\
& = & \int_t^T\Psi_t(r,u)dr.
\end{eqnarray*}
\end{proof}

\begin{lemma}\label{Lemma: Gammas querdrift}
For any $u\in\mathbb R$, $t\in\mathbb R_+$ we have
$$\Gamma_t(t,u)=\int_0^t\left(\Psi_{s-}(s,u)-\psi^X_s(u)\right)ds.$$
\end{lemma}
\begin{proof}
By Fubini's theorem for ordinary and stochastic integrals we have
\begin{eqnarray*}
 \Gamma_t(t,u) & = & \int_0^t\Psi_0(r,u)dr+\int_0^t\delta_s(t,u)ds+\int_0^t\sigma_s(t,u)dM_s\\
 &=& \int_0^t\bigg(\Psi_0(r,u)+\int_0^r\alpha_s(r,u)ds-\psi^X_r(u)\bigg)dr+\int_0^t\int_0^r\beta_s(r,u)dM_sdr.
\end{eqnarray*}
This yields the claim.
\end{proof}

\begin{lemma}\label{Lemma: Bedingte charackteristische Funktion} If $\osm$ is risk neutral, then
it satisfies the conditional expectation condition.
\end{lemma}
\begin{proof} Let $T\in \mathbb R_+$. 
We define
$$O_t(T,x):=\bigg\{\begin{array}{cc}e^{-x}C_t(T,e^x) 
& \mathrm{if}\ x\geq0,\\e^{-x}P_t(T,e^x) & \mathrm{if}\ x<0,\end{array}$$
where $P_t(T,K):=C_t(T,K)+K-S_t$ for any $K\in\mathbb R_+$, $t\in[0,T],x\in\mathbb R$. Then we have
$$O_t(T,x)=\bigg\{\begin{array}{cc}(e^{X_t-x}-1)^++\scr O_t(T,x-X_t) 
& \mathrm{if}\ x\geq0,\\(1-e^{X_t-x})^++\scr O_t(T,x-X_t) & \mathrm{if}\ x<0.\end{array}$$
We calculate the Fourier transform of $O_t(T,x)$ in two steps by considering the summands separately. 
The improper Fourier transform of the second summand $\scr O_t(T,x-X_t)$ exists and satisfies
\begin{eqnarray*}
 \F\{x\mapsto\scr O_t(T,x-X_t)\}(u)&=&\F\{x\mapsto\scr O_t(T,x)\}(u)e^{iuX_t}\\ 
&=& \frac{1-\exp\left(\int_t^T\Psi_t(r,u)dr\right)}{u^2+iu}e^{iuX_t}
\end{eqnarray*}
for any $u\in\mathbb R\setminus\{0\}$ by Remark \ref{Bemerkung: Vergleich zu PII}, 
Proposition \ref{Proposition: Fourier transformierte der Call-Preise} and the translation property for the 
Fourier transform, which holds for the improper Fourier transform as well. 
The Fourier transform of the first summand $A_t(T,x):=O_t(T,x)-\scr O_t(T,x-X_t)$ exists and equals
$$\F\{x\mapsto A_t(T,x)\}(u)=\frac{1}{iu}-\frac{e^{X_t}}{iu-1}-\frac{e^{iuX_t}}{u^2+iu}$$
for any $u\in\mathbb R\setminus\{0\}$.
Therefore the improper Fourier transform of $x\mapsto O_t(T,x)$ exists and is given by
\begin{equation}
 \F\{x\mapsto O_t(T,x)\}(u)=\frac{1}{iu}-\frac{e^{X_t}}{iu-1}-
\frac{\exp(iuX_t+\int_t^T\Psi_t(r,u)dr)}{u^2+iu},\label{Gleichung: Optionen Transformiert}
\end{equation}
for any $u\in\mathbb R\setminus\{0\}$. By Lemmas \ref{Lemma: Gammazerlegung} and 
\ref{Lemma: Gammas querdrift} we have that the right-hand side of 
\eqref{Gleichung: Optionen Transformiert} is a semimartingale, in particular it has c\`adl\`ag paths. 
Remark \ref{Bemerkung: Vergleich zu PII} yields that $0\leq P_t(T,K)\leq K$. 
Hence $\left(P_t(T,K)\right)_{t\in[0,T]}$ is a martingale because it is a bounded local martingale. 
Let $(\tau_n)_{n\in\mathbb N}$ denote a common localising sequence for $\left(C_t(T,1)\right)_{t\in[0,T]}$ 
and $S$, i.e.\ $S^{\tau_n}$, $C^{\tau_n}(T,1)$ are uniformly integrable martingales 
for any $n\in\mathbb N$. Since $C^{\tau_n}_t(T,K)\leq C^{\tau_n}_t(T,1)$ for $K\in[1,\infty)$, 
we have that $(\tau_n)_{n\in\mathbb N}$ is a common localising sequence for all European calls 
with maturity $T$ and strike $K\geq 1$.  The definition of $O_t(T,x)$ yields that it is a local martingale 
for any $x\in\mathbb R$ and $(\tau_n)_{n\in\mathbb N}$ is a common localising sequence for 
$(O_t(T,x))_{t\in[0,T]},x\in\mathbb R$. 

Fix $\omega\in\Omega$. Since $u\mapsto\int_t^T\Psi_t(r,u)(\omega)dr$ is in $\Pi$ for any $t\in[0,T]$, 
there is a random variable $Y$ on some space $(\widetilde\Omega,\widetilde{\scr F},\widetilde P)$ 
with characteristic function
$$\widetilde E(e^{iuY})=\exp\left(\int_t^T\Psi_t(r,u)(\omega)dr\right).$$
Then $\widetilde E (e^Y)=1$ and
$$O_t(T,x)(\omega)=\bigg\{\begin{array}{cc}\widetilde E( (S_t(\omega)e^{Y-x}-1)^+) 
& \mathrm{if}\ x\geq0,\\\widetilde E(1-S_t(\omega)e^{Y-x})^+) & \mathrm{if}\ x<0,\end{array}$$
cf. Remark \ref{Bemerkung: Vergleich zu PII}. By Corollary \ref{Korollar: Abschaetzung der Fourier transformierten} 
we have $\vert\int_{-C}^\infty e^{iux}O_t(T,x)dx\vert\leq S_t(\omega)+\frac{1+2\vert u\vert}{u^2}$. 
Proposition \ref{Proposition: Fourier Transformation von Martingalen} yields that 
$$\left(\F\{x\mapsto O_t(T,x)\}(u)\right)_{t\in[0,T]}$$ 
and hence $\left(\Phi_t(u)\right)_{t\in[0,T]}$ given by 
$$\Phi:\Omega\times[0,T]\times\mathbb R\rightarrow\mathbb C,\quad
(\omega,t,u)\mapsto\exp\left(iuX_t(\omega)+\int_t^T\Psi_t(r,u)(\omega)dr\right)$$
are local martingales for any $u\in\mathbb R\setminus\{0\}$. Since 
$u\mapsto \int_t^T\Psi_t(r,u)(\omega)dr$ is in $\Pi$ for any $t\in[0,T]$, $\omega\in\Omega$, 
its real part is bounded by $0$ from above, cf.\ Lemma \ref{Lemma: Negativitaet von LKM}. 
Hence $\vert\Phi_t(u)\vert\leq 1$ and thus $(\omega,t)\mapsto\Phi_t(u)(\omega)$ is a 
true martingale for any $u\in\mathbb R\setminus\{0\}$. By $\Phi_t(0)=1$ it is a martingale for $u=0$ 
as well. Since $\Phi_T(u)=\exp(iuX_T)$, the two martingales $\left(\Phi_t(u)\right)_{t\in[0,T]}$ and 
$\left(E(\exp(iuX_T)\vert\F_t)\right)_{t\in[0,T]}$ coincide for any $u\in\mathbb R$. Thus we have
$$\exp\bigg(\int_t^T\Psi_t(r,u)dr\bigg)=\exp(-iuX_t)\Phi_t(u)=E(e^{iu(X_T-X_t)}\vert\F_t)$$
for any $u\in\mathbb R$, $t\in[0,T]$.
\end{proof}

\begin{lemma}[Drift condition in terms of $\delta$ and $\sigma$]\label{Lemma: Driftbedingung in delta und sigma} 
If $\osm$ satisfies the conditional expectation condition, we have the drift condition
$$\delta_t(T,u)=\Psi_{t-}(t,u)-\psi_t^X(u)-\psi_t^{(X,M)}(u,-i\sigma_t(T,u))$$
outside some $dP\otimes dt$-null set for $T\in\mathbb R_+$, $u\in\mathbb R$. 
In particular, $(u,-i\sigma(T,u))\in\scr U^{(X,M)}$.
\end{lemma}
\begin{proof}
 For $u\in\mathbb R$ and $T\in\mathbb R_+$ define the process $Z_t:=iuX_t+\int_t^T\Psi_t(r,u)dr$. 
The conditional expectation condition yields that $\exp(Z_t)=E(e^{iuX_T}\vert\F_t)$ is a martingale. 
Hence $-i\in\scr U^Z$ and $\psi^Z_t(-i)=0$ by Lemma \ref{Lemma: exp(X) Martingal gdw psi(-i)=0}. 
With $Y_t:=\Gamma_t(t,u)$ we obtain
\begin{eqnarray*}
0 & = & \psi^Z_t(-i)\\
 & = & \psi^{iuX+\Gamma(T,u)-Y}_t(-i) \\
 & = & \psi^{iuX+\Gamma(T,u)}_t(-i)-\big(\Psi_{t-}(t,u)-\psi^X_t(u)\big) \\
 & = & \psi^{(iuX,\Gamma(T,u))}_t(-i,-i)-\Psi_{t-}(t,u)+\psi^X_t(u) \\
 & = & \psi^{(X,M)}_t(u,-i\sigma_t(T,u))+\delta_t(T,u)-\Psi_{t-}(t,u)+\psi^X_t(u),
\end{eqnarray*}
where the second equation follows from Lemma \ref{Lemma: Gammazerlegung}, the third from 
Lemmas \ref{Lemma: Gammas querdrift} and \ref{Lemma: Driftregel fuer Levy Exponenten}, 
the fourth from Lemma \ref{Lemma: Levy exponenten Formel} and the last from 
Lemmas \ref{Lemma: Levy exponeten Formel} and  \ref{Lemma: Driftregel fuer Levy Exponenten}.
\end{proof}

\begin{corollary}[Consistency condition]\label{Korollar: Abstrakte Konsistenzbedingung} 
 If $\osm$ satisfies the conditional expectation condition, then it satisfies the consistency condition.
\end{corollary}
\begin{proof}
 Lemma \ref{Lemma: Driftbedingung in delta und sigma} and the definition of $\delta$ yield
 $$\Psi_{t-}(t,u)=\delta_{t}(t,u)+\psi_t^X(u)+\psi_t^{(X,M)}(u,0)=\psi_t^X(u).$$
\end{proof}

\begin{corollary}[Drift condition]\label{Korollar: Driftbedingung} If $\osm$ satisfies the 
conditional expectation condition, then it satisfies the drift condition.
\end{corollary}
\begin{proof}
 This follows from Lemma \ref{Lemma: Driftbedingung in delta und sigma} and 
Corollary \ref{Korollar: Abstrakte Konsistenzbedingung}.
\end{proof}

\begin{lemma}\label{Lemma: Gamma Darstellung}
 If the option surface model satisfies the consistency condition, then
$$\Gamma_t(T,u)=\int_t^T\Psi_t(r,u)dr$$ 
for any $T\in\mathbb R_+$, $t\in[0,T]$, $u\in\mathbb R$.
\end{lemma}
\begin{proof}
This is a direct consequence of Lemmas \ref{Lemma: Gammazerlegung} and \ref{Lemma: Gammas querdrift}.
\end{proof}

\begin{lemma}\label{Lemma: Psi ist Martingal}
  If the option surface model satisfies the drift condition, then
$$\Phi_t(T,u):=\exp(iuX_t+\Gamma_t(T,u))$$
defines a local martingale $(\Phi_t(T,u))_{t\in[0,T]}$ for any $u\in\mathbb R$, $T\in\mathbb R_+$.
\end{lemma}
\begin{proof}
Fix $T, u$ and define $Z_t:=iuX_t+\Gamma_t(T,u)$. By the drift condition and 
Lemmas \ref{Lemma: Levy exponeten Formel} -- \ref{Lemma: Driftregel fuer Levy Exponenten} we have
\begin{eqnarray*}
 0&=&\psi_t^{(X,M)}(u,-i\sigma_t(T,u))+\delta_t(T,u)\\
 &=&\psi_t^{(X,\sigma(T,u)\mal\! M)}(u,-i)+\delta_t(T,u)\\
 &=&\psi_t^{(iuX,\Gamma(T,u))}(-i,-i)\\
 &=&\psi_t^{iuX+\Gamma(T,u)}(-i)\\
 &=&\psi_t^Z(-i).
\end{eqnarray*}
 Hence $\exp(Z)$ is a local martingale by Lemma \ref{Lemma: exp(X) Martingal gdw psi(-i)=0}.
\end{proof}

\begin{lemma}\label{Lemma: Bedingungen gdw conditional expectation}
 $\osm$ satisfies the drift and consistency conditions if and only if it satisfies the conditional expectation condition.
\end{lemma}
\begin{proof}
$\Leftarrow:$ This is a restatement of Corollaries \ref{Korollar: Driftbedingung} and 
\ref{Korollar: Abstrakte Konsistenzbedingung}.

$\Rightarrow:$ Fix $u\in\mathbb R$, $T\in\mathbb R_+$. Lemma \ref{Lemma: Negativitaet von LKM}
implies that the absolute value of
$$\Phi_t(T,u):=\exp\!\left(iuX_t+\int_t^T\Psi_t(r,u)dr\right)$$ 
is bounded by $1$. By Lemmas \ref{Lemma: Gamma Darstellung} and \ref{Lemma: Psi ist Martingal}, 
$\Phi(T,u)$ is a local martingale and hence a martingale. This yields 
$$\Phi_t(T,u)=E(\Phi_T(T,u)\vert\F_t)=E(e^{iuX_T}\vert\F_t).$$
\end{proof}

\begin{lemma}\label{Lemma: conditional expectation implies martingale property}
 If the option surface model $\osm$ satisfies the conditional expectation condition, then $S=e^X$ is a martingale.
\end{lemma}
\begin{proof}
 For $T\in\mathbb R_+$, $t\in[0,T]$ we have
$$E(e^{iu(X_T-X_t)}\vert\F_t)=\exp\left(\int_t^T\Psi_t(r,u)dr\right).$$
Since $u\mapsto\int_t^T\Psi_t(r,u)(\omega)dr$ is in $\Pi$, we have
$E(e^{X_T-X_t}\vert\F_t)=1$,
cf.\ Remark \ref{Bemerkung: Vergleich zu PII}.
\end{proof}

\begin{lemma}\label{Lemma: conditional expectation impliziert risiko neutral}
 If the option surface model $\osm$ satisfies the conditional expectation condition, it is strongly risk neutral.
\end{lemma}
\begin{proof}
 Lemma \ref{Lemma: conditional expectation implies martingale property} implies that $e^X$ is a 
martingale and in particular that $e^{X_t}$ is integrable for any $t\in\mathbb R_+$. 
Let $T\in\mathbb R_+$, $t\in[0,T]$. We define
\begin{eqnarray*}
 \widetilde C(K)&:=&E((e^{X_T}-K)^+\vert\F_t),\\
 \widetilde{\scr O}(x)&:=&e^{-(x+X_t)}\widetilde{C}(e^{x+X_t})-(e^{-x}-1)^+,\\
 Y &:=& X_T-X_t
\end{eqnarray*}
 for any $K\in\mathbb R_+$. Obviously we have
$$\widetilde{\scr O}(x)=\bigg\{\begin{array}{cc}E((e^{Y-x}-1)^+\vert\F_t) & \mathrm{if}\ x\geq0,\\ 
E((1-e^{Y-x})^+\vert\F_t) & \mathrm{if}\ x<0\end{array}$$
and $E(e^Y\vert\F_t)=1$. Hence Corollary \ref{Proposition: Fourier transformierte der Call-Preise}, 
the conditional expectation condition, and the definition of $\scr O$ yield
\begin{eqnarray*}
\F\{x\mapsto\widetilde{\scr O}(x)\}(u)&=&\frac{1-E(e^{iuY}\vert\F_t)}{u^2+iu}\\
&=&\frac{1-\exp\left(\int_t^T\Psi_t(r,u)dr\right)}{u^2+iu}
\end{eqnarray*}
as well as
\begin{eqnarray*}
\widetilde{\scr O}(x)&=&\scr F^{-1}\left\{u\mapsto\frac{1-\exp\left(\int_t^T\Psi_t(r,u)dr\right)}{u^2+iu}\right\}(x)\\
&=&\scr O_t(T,x).
\end{eqnarray*}
Thus we have
$$\widetilde C(K)=C_t(T,K)$$
for any $K\in\mathbb R_+$. Consequently, the option surface model $\osm$ is strongly risk neutral.
\end{proof}

\begin{proofofmaintheorem}
$(1)\Rightarrow(2)$ is obvious.\\
$(2)\Rightarrow(3)$ has been shown in Lemma \ref{Lemma: Bedingte charackteristische Funktion}.\\
$(3)\Leftrightarrow(4)$ is the conclusion of Lemma \ref{Lemma: Bedingungen gdw conditional expectation}.\\
$(3)\Rightarrow(1)$ has been shown in Lemma \ref{Lemma: conditional expectation impliziert risiko neutral}.
\end{proofofmaintheorem}

\subsection{Musiela parametrisation}
In practice one may prefer to parametrise the codebook in terms of
time-to-maturity $x:=T-t$  instead of maturity $T$,
which is referred to as {\em Musiela parametrisation} in interest rate theory.
However, in order to express the corresponding codebook dynamics, 
we need some additional regularity.
\begin{proposition}\label{p:musiela}
 Let $\osm$ be an option surface model such that
 \begin{enumerate}
  \item $M$ is of the form $M_t=N_t+\int_0^tv_sds$ with some locally square-integrable martingale
  $N$ and an integrable predictable process $v$,
 \item $T\mapsto \alpha_t(T,u),\beta_t(T,u),\Psi_0(T,u)$ are continuously differentiable for any $t\in\rp$, $u\in\rr$,
 \item $\int_0^T|\partial_r\beta_t(r,u)|^2dr<\infty$ for any $t,T\in\mathbb R_+$, $u\in\mathbb R$,
 \item $\int_0^t\sup_{r\in[0,T]}\vert \partial_r\alpha_s(r,u)+ \partial_r\beta_s(r,u)v_s\vert ds <\infty$ 
 for any $t,T\in\rp$, $u\in\rr$,
  \item $((\int_0^T(\vert \beta^j_t(r,u)\vert^2 
  +\vert \partial_r\beta^j_t(r,u)\vert^2)dr)^{1/2})_{t\in\rp}\in L^2_\mathrm{loc}(N^j)$ 
  for any $T\in\rp$, $u\in\rr$, $j\in\{1,\dots,d\}$,
  where $L^2_\mathrm{loc}(N^j)$ is defined as in \cite[I.4.39]{js.87}. 
 \end{enumerate}
Define
  $\check\alpha_t(x,u) := \alpha_t(t+x,u)$,
  $\check\beta_t(x,u) := \beta_t(t+x,u)$, 
  $\check\Psi_t(x,u) := \Psi_t(t+x,u)$
for any $t\in\rp$, $x\in\rp$, $u\in\rr$. For any fixed $u\in\rr$, the mapping 
$x\mapsto\check\Psi_t(x,u)$ is differentiable for $dt$-almost all $t\in\rp$ and we have
 \begin{eqnarray*}
  \check\Psi_t(x,u) = \check\Psi_0(x,u) + 
  \int_0^t\left(\check\alpha_s(x,u) + \partial_x\check\Psi_s(x,u)\right)ds + \int_0^t\check\beta_s(x,u) M_s
\end{eqnarray*}
for any $t\in\rp$, $x\in\rr$, $u\in\rr$.
\end{proposition}
\begin{proof}
Since
\begin{eqnarray*}
\Psi_t(T,u)&=&\Psi_0(T,u)+\int_0^{t\wedge T}\alpha_s(T,u)ds+\int_0^{t\wedge T}\beta_s(T,u)dM_s\\
&=&\Psi_0(T,u)+\int_0^{t\wedge T}(\alpha_s(T,u)+\beta_s(T,u)v_s)ds+\int_0^{t\wedge T}\beta_s(T,u)dN_s,
\end{eqnarray*}
we can assume w.l.o.g.\ that $M$ is a  locally square-integrable martingale.
By localisation, it even suffices to consider square-integrable martingales.

By Fubini's theorem for the stochastic integral (cf.\ \cite[Theorem IV.65]{protter.04}), we have
\begin{eqnarray*}
  \check\Psi_t(x,u) &=& \check\Psi_0(t+x,u) + \int_0^t \check\alpha_s(t-s+x,u)ds 
  + \int_0^t\check\beta_s(t-s+x,u)dM_s \\
                    &=& \check\Psi_0(x,u) + \int_0^t \check\alpha_s(x,u)ds + \int_0^t\check\beta_s(x,u)dM_s \\
  &&{}+ \int_0^t \partial_x\check\Psi_0(r+x,u)dr + \int_0^t\int_s^t\partial_x\check\alpha_s(r-s+x,u)drds \\ 
  &&{}+ \int_0^t\int_s^t\partial_x\check\beta_s(r-s+x,u)drdM_s \\
                     &=& \check\Psi_0(x,u) + \int_0^t \check\alpha_s(x,u)ds + \int_0^t\check\beta_s(x,u)dM_s \\
  &&{}+ \int_0^t \partial_x\check\Psi_0(x+r,u)dr + \int_0^t\int_0^r\partial_x\alpha_s(r+x,u)dsdr \\ 
  &&{}+ \int_0^t\int_0^r\partial_x\beta_s(r+x,u)dM_sdr \\
                    &=& \check\Psi_0(x,u) + \int_0^t \check\alpha_s(x,u)ds + \int_0^t\check\beta_s(x,u)dM_s \\
  &&{}+ \int_0^t \Big(\partial_x\check\Psi_0(x+r,u) + \partial_x\int_0^r\alpha_s(r+x,u)ds \\ 
  &&{}+ \partial_x\int_0^r\beta_s(r+x,u)dM_s\Big)dr \\
                    &=& \check\Psi_0(x,u) + \int_0^t\left(\check\alpha_s(x,u) + \partial_x\check\Psi_s(x,u)\right)ds 
                    + \int_0^t\check\beta_s(x,u) M_s
\end{eqnarray*}
 for any $t,x\in\rp$, $u\in\rr$ where the fourth equality is explained below. 
 For fixed $t,x\in\rp$, $u\in\rr$ define Hilbert spaces $H:=L^2([0,t+x],\mathbb R)$ and  
$$H_1:=\{ f\in H:f\text{ is differentiable and }f'\in H\}$$
with norm
$$\|f\|_{H_1}:=\sqrt{\|f\|_{H}^2+\|f'\|_{H}^2}.$$
 The mapping $r\mapsto \beta_t(r+x,u)$ is in $H_1$ by assumption. 
 Since $\partial_r:H_1\rightarrow H$, $f\mapsto f'$  is linear and continuous,
 \cite[Theorem 8.7(v)]{peszat.zabczyk.07} yields
 $$ \int_0^t\partial_r\beta_s(\cdot,u)dM_s = \partial_r\int_0^t\beta_s(\cdot,u)dM_s$$
 and hence
 \begin{eqnarray*}
  \int_0^t\int_0^r\partial_x\beta_s(r+x,u)dM_sdr
   = \int_0^t\partial_x\int_0^r\beta_s(r+x,u)dM_sdr.
 \end{eqnarray*}

 \end{proof}

\section{Constructing models from building blocks}\label{Abschnitt: Beispiele und Existenzresultate}
In this section we turn to existence and uniqueness results for option surface models
which are driven by a subordinator $M$.

\subsection{Building blocks}
Theorem \ref{Satz: Bedingungen impl. Vollstaendig} indicates thet neither the drift part 
$\alpha$ of the codebook nor the dynamics of the return process $X$ can be chosen arbitrarily 
if one wants to end up with a risk-neutral option surface model. What ingredients do we need in 
order to construct such a model? It seems natural to consider volatiliy processes $\beta$ that are 
functions of the present state of the codebook, i.e.\
$$\beta_t(T,u)(\omega) = b(t,\Psi_{t-}(\cdot,\cdot)(\omega))(T,u)$$
for some deterministic function $b:\mathbb R_+\times \raum{\Pi}\rightarrow \raum{\Pi}$,
 where $\raum{\Pi}$ denotes some suitable space of conceivable codebook states, i.e.\ 
essentially of functions $\mathbb R_+\rightarrow\Pi$. It is specified below. In order to hope 
for uniqueness, we need to fix the initial values $X_0$ and $\Psi_0(\cdot,\cdot)$, function $b$ 
and the law of the driving process $M$. The drift part in \eqref{Gleichung:neu2} need not be 
specified as it is implied by the drift condition. But we need some information on $X$. 
Although its dynamics seem to be determined by the consistency condition, the joint 
behaviour of $X$ and $M$ is not. The latter, however, is needed for the drift condition 
\eqref{Gleichung: Driftbedingung} resp.\ \eqref{G:Driftbedingung, einf.}. In view of (\ref{e:letztedrift}), 
we assume that the joint law of $M$ and the dependent part $\Xp$ of $X$ relative to $M$ 
in the sense of Section \ref{Abschnitt: Semimartingalprojektion} are given. More specifically, 
we suppose 
that $(\Xp,M)$ is a L\'evy process with given L\'evy exponent $\psi^{(\Xp,M)}=\gamma$. 
The components of $M$ are supposed to be subordinators.
Altogether, we suggest to construct models based on a quadrupel $(x_0,\psi_0,b,\gamma)$, 
where $x_0\in\mathbb R$ and $\psi_0\in \raum{\Pi}$ stand for the initial states of the return 
process and the codebook, respectively.

In order to derive existence and uniqueness results, we still need to specify the domain and codomain of $b$. 
For ease of notation, we focus on one-dimensional driving processes $M$. 
The vector-valued case can be treated along the same lines.

Let $E$ denote the set of continuous functions $\psi:\mathbb R\rightarrow\mathbb C$  and
$$\Vert\psi\Vert_m:=\sup\left\{|\psi(u)|:|u|\le m\right\}$$
for $m\in\rp$.
By $\scr L^1(\mathbb R_+,E)$ we denote the set of measurable functions 
$\psi:\mathbb R_+\times\mathbb R\rightarrow\mathbb C$ such that $\psi(T,\cdot)\in E$ and
$$\|\psi\|_{T,m}:=\int_0^T\Vert\psi(r,\cdot)\Vert_m dr<\infty$$
for any $T,m\in\mathbb R_+$. For $\psi\in \scr L^1(\mathbb R_+,E)$ we set
$$ [\psi] := \{ \varphi\in\scr L^1(\rp,E): \psi(T,\cdot)=\varphi(T,\cdot) \text{ for almost any }T\in\rp\}.$$
Moreover, we define the space
 $$\raum{E} := \{ [\psi] : \psi\in\scr L^1(\rp,E) \}.$$
as usual. Finally, we set
\begin{eqnarray*}
 \raum{\Pi} &:=& \left\{\psi\in\raum{E}:\int_t^T\psi(r,\cdot)dr\in\Pi\text{ for any }0\leq t\leq T<\infty\right\}, 
\end{eqnarray*}
where we refer to the Bochner integral in the sense of Definition \ref{d:frechetintegral}
and Example \ref{ex:frechet}.
%
\begin{lemma}\label{l:banach}
The following statements hold:
 \begin{enumerate}
  \item $(E,\Vert\cdot\Vert_m)$ is a complete and separable semi-normed space for any $m\in\rp$.
\item $(\raum{E},\Vert\cdot\Vert_{T,m})$ is a complete and separable semi-normed space for any $T,m\in\mathbb R_+$.
If $x\in\raum{E}$ with $\|x\|_{n,n}=0$ for any $n\in\mathbb N$, we have $x=0$.
Moreover, if $(x_k)_{k\in\mathbb N}$ is a Cauchy sequence in $\raum{E}$ relative to any 
$\|\cdot\|_{n,n}$, $n\in\mathbb N$,
there exists $x\in\raum{E}$ such that 
$\lim_{k\to\infty}\|x_k-x\|_{n,n}=0$, $n\in\mathbb N$. 
Consequently,
$(\raum{E},d)$ is a separable Fr\'echet space for the metric
 $$d(\psi,\varphi):=\sum_{n\in\mathbb N}2^{-n}(1\wedge \Vert\psi-\varphi\Vert_{n,n}),\quad \psi,\varphi\in \raum{E}.$$
  \item $\Pi\subset E$ is a convex cone. If $A$ is a Borel subset of $\rr$, $\mu$ a finite measure on $A$, and 
$\psi:A\times\mathbb R\rightarrow\mathbb C$ is measurable with $\psi(r,\cdot)\in\Pi$ and 
$\int_A\Vert\psi(r,\cdot)\Vert_m\mu(dr)<\infty$ for all $m\in\mathbb N$, then the mapping
$u\mapsto\int_A\psi(r,u)\mu(dr)$
is in $\Pi$. 
  \item If $\psi\in\raum{E}$ and $\psi(T,\cdot)\in\Pi$ for almost all $T\in\mathbb R_+$, 
then $\psi\in\raum{\Pi}$.  
  \item For any increasing function $X:\mathbb R_+\rightarrow\mathbb R_+$ 
and any locally $X$-integrable function $\eta:\mathbb R_+\rightarrow\raum{\Pi}$, 
we have $\int_0^t\eta_s dX_s\in\raum{\Pi}$ for any $t\in\mathbb R_+$. 
Here, we refer to Bochner integration on $\raum{E}$, cf.\ Example \ref{ex:frechet} and Definition \ref{d:frechetintegral}.
  \item $\Pi$ is a Borel subset of $E$ and consequently $\raum{\Pi}$ is  a 
Borel subset of $\raum{E}$ (relative to the Borel-$\sigma$-field generated by the metric $d$).
 \end{enumerate}
\end{lemma}
\begin{proof}
\begin{enumerate}
\item This follows from the fact that the continuous functions on $[-m,m]$ are a separable
Banach space relative to the uniform norm.
\item $(\raum{E},\Vert\cdot\Vert_{T,m})$ is a complete and separable semi-normed space 
because the corresponding Lebesgue-Bochner space of integrable functions on $[0,T]$
with values in the Banach space of continuous functions $[-m,m]\to\mathbb C$ is a Banach space.

Let $Q$ be a countable dense set in $E$. Define 
$$S:=\left\{\sum_{j=1}^n q_j1_{(a_j,b_j]}: n\in\mathbb N,q\in Q^n,a,b\in\mathbb Q^n\right\}.$$ 
$S$ is dense in $\raum{E}$ because $S$ is obviously dense in
$$T:=\left\{\sum_{j=1}^n q_j1_{A_j}: n\in\mathbb N,q\in Q^n,A_1,\dots,A_n\in\scr B(\mathbb R_+)\right\}$$
and $T$ is dense in $\raum{E}$, cf.\ Section \ref{su:frechet} of the appendix.
This shows separability of $(\raum{E},\Vert\cdot\Vert_{T,m})$.
The remaining statements are straightforward.
 \item $\Pi$ is obviously a convex cone.  
Define 
$$\Psi:\mathbb R\rightarrow\mathbb C,\quad u\mapsto \int_A\psi(r,u)\mu(dr).$$
Then $\Psi$ is continuous and it is the pointwise limit of L\'evy exponents. Hence L\'evy's continuity theorem 
(see \cite[Theorem 3.6.1]{lukacs.70}) together with \cite[Theorem 5.3.3]{lukacs.70} yield that $\Psi$ is the 
characteristic exponent of an infinitely divisible random variable $X$.

Fix the truncation function $h:\mathbb R\rightarrow\mathbb R,x\mapsto x1_{\{\vert x\vert\leq 1\}}$. 
For all $r\in A$ let $(b_r,c_r,F_r)$ be the L\'evy-Khintchine triplet corresponding to $\psi(r,\cdot)$. 
A detailed analysis of the proof of \cite[Lemma II.2.44]{js.87} yields integrability of $b$ and $c$ and that 
$F$ is a transition kernel satisfying $\int_A\int(|x|^2\wedge 1)F_r(dx)\mu(dr)<\infty$.

In order to prove $\Psi\in\Pi$ we have to show that $Ee^X=1$. Let $(B,C,\nu)$ be the triplet corresponding to 
$\Psi$ and $h$. 
Then \cite[Theorem II.4.16]{js.87} yields
$$\exp(\Psi(u)) = \exp\left(iuB-\frac{u^2}{2}C+\int_{\mathbb R}(e^{iux}-1-iuh(x))\nu(dx)\right) $$
as well as
\begin{eqnarray*}
&&\exp(\Psi(u)) \\
&=& \exp\left(\int_A \psi(r,u)\mu(dr)\right) \\
 &=& \exp\left(iu\int_Ab_rdr-\frac{u^2}{2}\int_Ac_rdr+\int_A\int(e^{iux}-1-iuh(x))F_r(dx)\mu(dr)\right)
\end{eqnarray*}
for any $u\in\mathbb R$. From \cite[Lemma II.2.44]{js.87} we obtain
 $B=\int_Ab_r\mu(dr)$, $C=\int_Ac_r\mu(dr)$, and $\nu(G) = \int_AF_r(G)\mu(dr)$ for $G\in\scr B$. 
Consequently, we have
\begin{eqnarray*}
E(e^X) &=& \exp\left(B+\frac{1}{2}C+\int (e^x-1-h(x)) \nu(dx)\right) \\
 &=& \exp\Bigg(\int_A\left(b_r+\frac{1}{2}c_r+\int_{(-\infty,1]} (e^x-1-h(x)) F_r(dx)\right)\mu(dr) \\
&&+ \int_{(1,\infty)} (e^x-1) \nu(dx)\Bigg).
\end{eqnarray*}
Tonelli's theorem yields $\int_{(1,\infty)} (e^x-1) \nu(dx) = \int_A\int_{(1,\infty)} (e^x-1)F_r(dx)\mu(dr)$. Hence
$$E(e^X) =  \exp\left(\int_A\left(b_r+\frac{1}{2}c_r+\int (e^x-1-h(x)) F_r(dx)\right)\mu(dr)\right) = 1.$$

\item This is a consequence of Statement 3.

\item This is a consequence of Statement 3 as well.

\item 
$E$ is a metric space relative to 
$$\delta(\psi,\varphi):=\sum_{n\in\mathbb N}
2^{-n}(1\wedge \|\psi-\varphi\|_n).$$
Let $C\subset E$ denote the set of all L\'evy exponents. 
L\'evy's continuity theorem (cf.\ \cite[Theorem 3.6.1]{lukacs.70}) and 
\cite[Theorem 5.3.3]{lukacs.70} imply that $C$ is closed in $E$
and in particular a Borel subset in $E$. 
The function 
$$f:C\rightarrow \overline{\mathbb R}_+,
\quad\varphi\mapsto \frac{1}{2\pi}\int_{-\infty}^{\infty}e^x\int_{-\infty}^{\infty}
\exp\left(\varphi(u)-\frac{1}{2}(iu+u^2)-iux\right)dudx$$
is well defined and measurable, where $\overline{\mathbb R}_+:=\mathbb R_+\cup\{\infty\}$ . 
Indeed, observe that for $\varphi\in C$ the function 
$$u\mapsto \varphi(u)-\frac{1}{2}(iu+u^2)$$
is a L\'evy exponent and thus \cite[Theorem 3.2.2]{lukacs.70} yields that 
$$x\mapsto \frac{1}{2\pi}\int_{-\infty}^{\infty}\exp\left(\varphi(u)-\frac{1}{2}(iu+u^2)-iux\right)du$$
is a density function. 

Let $L$ be a L\'evy process with L\'evy exponent $\varphi\in C$ and $W$ be an independent Brownian motion 
with diffusion coefficient $1$ and drift rate $-1/2$. 
\cite[Theorem 3.2.2]{lukacs.70} yields that 
$$p:\mathbb R\rightarrow\mathbb R_+,\quad x \mapsto \frac{1}{2\pi}\int_{-\infty}^{\infty}
\exp\left(\varphi(u)-\frac{1}{2}(iu+u^2)-iux\right)du$$
is the density function of $L_1+W_1$. Thus we have
\begin{eqnarray*}
 E(e^{L_1}) &=& E(e^{L_1+W_1}) 
  = \int_{-\infty}^{\infty}e^xp(x)dx 
  = f(\varphi).
\end{eqnarray*}
Hence $\Pi = f^{-1}(1)$, which implies that $\Pi$ is measurable in $E$. 
It remains to be shown that $\raum{\Pi}$ is a measurable set in $\raum{E}$. For $a,b\in\mathbb R_+$ define
the continuous and hence measurable map 
$$I_{a,b}:\raum{E}\rightarrow E,\quad  \psi \mapsto \int_{a}^{b}\psi(r,\cdot)dr.$$
We obviously have
$$\raum{\Pi}\subset M:=\bigcap \{ I_{q_1,q_2}^{-1}(\Pi) : q_1,q_2\in\mathbb Q_+,q_1\leq q_2\}.$$
We show that the two sets are in fact equal.
Let $\psi\in M$ and $t,T\in\rp$ with $t< T$. Then $\psi\in\raum{E}$ and hence $\int_t^T \psi(r,\cdot)dr\in E$. 
Let $(q_n)_{n\in\mathbb N},(p_n)_{n\in\mathbb N}$ be sequences of rational numbers 
such that $q_n\downarrow t$, $p_n\uparrow T$, and $q_n\leq p_n$. Then
$$I_{t,T}\psi = I_{q_0,p_0}\psi + \sum_{n\in\mathbb N}(I_{q_{n+1},p_{n+1}}-I_{q_{n},p_{n}})\psi.$$
Defining the finite measure $\mu$ on $\mathbb N$ by
 $\mu(\{n\}):= 1/n^2$
and the function 
$$\gamma:\mathbb N\rightarrow \Pi,\quad n\mapsto n^2(I_{q_{n+1},p_{n+1}}-I_{q_{n},p_{n}})\psi,$$
we obtain
 $$ I_{t,T}\psi = I_{q_0,p_0}\psi + \int_{\mathbb N} \gamma d\mu.$$
Hence Statement 3 yields $I_{t,T}\psi \in \Pi$, which in turn implies $\psi\in\raum{\Pi}$. 

\end{enumerate}
\end{proof}
%
We are now ready to formalise the notion of building blocks. 

\begin{definition}\label{D:Bausteine}
 We call a quadruple $(x_0,\psi_0,b,\gamma)$ {\em building blocks} of an option surface model if
\begin{enumerate}
 \item $x_0\in\mathbb R_+$,
 \item $\psi_0\in\raum{\Pi}$,
 \item $b:\mathbb R_+\times\raum{E}\rightarrow\raum{E}$ is measurable
(relative to  the $\sigma$-fields $\B(\rp)\otimes\B(\raum{E})$ and $\B(\raum{E})$, where
$\B(\raum{E})$ denotes the Borel-$\sigma$-field on the metric space
$(\raum{E},d)$ as introduced in Lemma \ref{l:banach}),
\item $b$ maps $\mathbb R_+\times\raum{\Pi}$ on a subset of $\raum{\Pi}$,
\item $b$ is {\em locally Lipschitz} in the sense that 
for any $T\in\rp$ there are $T_0,m_0\in\rp$ such that for any $\widetilde T\geq T_0$, $m\geq m_0$
there exists $c\in\mathbb R_+$ such that
$$\Vert b(t,\psi_1)-b(t,\psi_2)\Vert_{\widetilde T,m}\leq c\Vert\psi_1-\psi_2\Vert_{\widetilde T,m}$$
holds for any $\psi_1,\psi_2\in\raum{E}$ and any  $t\in[0,T]$,
 \item $b(t,\psi)(T,u)=0$ for any $\psi\in\raum{E}$, $u\in\mathbb R$ and any $t,T\in\mathbb R_+$ with $t>T$,
\item $\sup_{t\in[0,T]}\|b(t,0)\|_{T,m}<\infty$ for any $T,m\in\rp$,
 \item $\gamma:\mathbb R\times(\rr+i\rp)\rightarrow\mathbb C$ is the extended L\'evy exponent 
on $\mathbb R\times(\rr+i\rp)$(cf.\ Section \ref{Unterabschnitt: L\'evy exponenten}) of an 
$\mathbb R^{1+1}$-dimensional L\'evy process $(\Xp,M)$ such that
\begin{enumerate}
 \item $M$ is a pure jump subordinator, i.e.\ $M_t=\sum_{s\leq t}\Delta M_s$, $t\in\mathbb R_+$,
 \item $\Xp$ is the dependent part of $\Xp$ relative to $M$,
 \item $\gamma$ is differentiable, and
 \item $\partial_2\gamma:\mathbb R\times(\rr+i\rp)\rightarrow\mathbb C$ is bounded and Lipschitz continuous.
\end{enumerate}
\end{enumerate}
\end{definition}

\begin{remark}
By Remark \ref{Remark smoothness} the smoothness conditions (c,d) on $\gamma$ are satisfied if 
both $X^\|$ and $M$ have finite second moments.
\end{remark}

Our goal is to find corresponding risk-neutral option surface models in the following sense.
\begin{definition}\label{Definition: kompatibel}
 An option surface model $\osm$ is said to be {\em compatible} with building blocks $(x_0,\psi_0,b,\gamma)$ if
\begin{itemize}
 \item $X_0=x_0$,
 \item $\Psi_0=\psi_0$,
 \item 
$\rp\to\raum{E}$, $t\mapsto \Psi_{t}(\omega)$ is a well-defined, a.s.\ c\`adl\`ag mapping,
\item
$\beta_{t}(\omega) = b(t,\Psi_{t-}(\omega))$ for $dP\otimes dt$-almost any 
$(\omega,t)\in \Omega\times\mathbb R_+$, 
 \item $\psi^{(\Xp,M)}=\gamma$ on $\mathbb R\times(\rr+i\rp)$,
where $\Xp$ denotes the dependent part of $X$ relative to $M$.
\end{itemize}
\end{definition}

\begin{remark}
In practice one may be interested in coefficients $b$ of the form
\begin{equation}\label{e:musiela}
 b(t,\psi)(T,u):=\left\{\begin{array}{ll}
\check b(\psi^{(t)})(T-t,u) & \mbox{ if } t\leq T,\\
0 & \mbox{ if }t> T,
\end{array}\right.
\end{equation}
where 
$\psi^{(t)}\in\raum{E}$ is defined by
$$\psi^{(t)}({x},u):=\psi(t+{x},u),\quad {x}\in\rp, \quad u\in\rr$$
for $\psi\in\raum{E}$.
In line with Proposition \ref{p:musiela}, $(\Psi_t^{(t)})_{t\in\rp}$ may be called
{\em Musiela parametrisation} of the codebook process $(\Psi_t)_{t\in\rp}$.
In other words, the Mu\-sie\-la codebook refers to a function of the remaining life time 
${x}=T-t$ rather than maturity $T$ of the claim.

Function $b$ in (\ref{e:musiela}) satisfies 
Conditions 3--7 in Definition \ref{D:Bausteine}
if $\check b:\raum{E}\rightarrow\raum{E}$  maps $\raum{\Pi}$ on a subset of $\raum{\Pi}$
and if $\check b$ is {\em locally Lipschitz} in the sense that
there exist $x_0,m_0\in\rp$ such that for any $x\geq x_0,m\geq m_0$ there exists $c\in\mathbb R_+$ such that
$$\Vert \check b(\psi_1)- \check b(\psi_2)\Vert_{\widetilde x,m}\leq c\Vert\psi_1-\psi_2\Vert_{\widetilde x,m}$$
holds for any $\psi_1,\psi_2\in\raum{E}$ and any $\widetilde x\in[x_0,x]$. 
\end{remark}

\subsection{Existence and uniqueness results}\label{Abschnitt: Existenzergebnis HJM}
For building blocks $(x_0,\psi_0,b,\gamma)$
consider the stochastic differential equation (SDE)
  \begin{eqnarray}\label{G:Psi SDE}
d\Psi_t = a(t,\Psi_{t-})dt+b(t,\Psi_{t-})dM_t,\quad\Psi_0=\psi_0
\end{eqnarray}
in $\raum{E}$, where $M$ denotes a subordinator with L\'evy exponent $\gamma(0,\cdot)$ and
\begin{eqnarray}\label{G:Driftbedingung, Markovsch}
a(t,\psi)(T,u) &:=& -\partial_T\left(\gamma\Big(u,-i\int_{t\wedge T}^T\tilde b(t,\psi)(r,u)dr\Big)\right) \\
   &=& i\partial_2\gamma\Big(u,-i\int_t^T\tilde b(t,\psi)(r,u)dr\Big)\tilde b(t,\psi)(T,u)1_{[t,\infty)}(T). \nonumber
\end{eqnarray}
with
\begin{eqnarray}\label{e:tilb}
 \tilde b(t,\psi)(r,u)
:=\big(\Re(b(t,\psi)(r,u))\wedge 0\big)
+i \Im (b(t,\psi)(r,u)).
\end{eqnarray}
Note that $\tilde b(t,\psi)= b(t,\psi)$
for $\psi\in\Pi$ by Lemma \ref{Lemma: Negativitaet von LKM}.
In view of Equations \eqref{e:letztedrift} and \eqref{Gleichung:neu2}, 
any compatible codebook process should solve \eqref{G:Psi SDE}. 
We start by showing that \eqref{G:Psi SDE} allows for a unique solution in $\raum{E}$.

\begin{proposition}\label{Proposition: SDE ist eindeutig}
Let $(x_0,\psi_0,b,\gamma)$ be building blocks.
\eqref{G:Driftbedingung, Markovsch} defines a measurable function 
$a:\mathbb R_+\times\raum{E}\rightarrow\raum{E}$. 
 Let $T,T_0,m_0,\widetilde T,m$ be as in Definition \ref{D:Bausteine}(5). 
 Then there are a constant $C\in\rp$ and for any  $\Vert\cdot\Vert_{\widetilde T,m}$-bounded set 
 $B\subset \raum{E}$ a constant $c\in\rp$ such that $a$ satisfies the Lipschitz and linear growth conditions
\begin{eqnarray*} 
\Vert a(t,\psi_1) - a(t,\psi_2) \Vert_{\widetilde T,m} &\leq& c\Vert \psi_1-\psi_2\Vert_{\widetilde T,m}, \\
 \Vert a(t,\psi)  \Vert_{\widetilde T,m} &\leq& C(1+\Vert \psi\Vert_{\widetilde T,m})
\end{eqnarray*}
for any $t\in[0,T]$, $\psi_1,\psi_2\in B$, $\psi\in \raum{E}$. 
 \item 
\sloppy{If $(X^\|,M)$ denotes an $\rr^2$-valued L\'evy process with L\'evy exponent $\gamma$
(implying in particular that $M$ is a subordinator), 
then the SDE \eqref{G:Psi SDE} has a unique c\`adl\`ag $\raum{E}$-valued solution,
in the sense of $\raum{E}$-valued processes and integrals, cf.\ Section \ref{s:frechetsde}. 
The joint law of $(X^\|,M,\Psi)$ on
$$\left(\mathbb D(\rr^2)\times \mathbb D(\raum{E}), \scr D(\rr^2)\otimes \scr D(\raum{E})\right)$$
is uniquely determined by $(x_0,\psi_0,b,\gamma)$. Here, $(\mathbb D(\rr^2), \scr D(\rr^2))$
and 
$(\mathbb D(\raum{E})$, $\scr D(\raum{E}))$ 
denote the Skorohod spaces of c\`adl\`ag functions
on $\rr_+$ with values in the Polish spaces $\rr^2$ and $\raum{E}$, respectively
(cf.\ \cite[Section 3.5]{ethier.kurtz.86}).}
\end{proposition}
\begin{proof}
\begin{enumerate}
\item
\sloppy{From the representation 
\begin{equation}\label{e:a}
 a(t,\psi)(T,u)
=i\partial_2\gamma\Big(u,-i\int_t^T\tilde b(t,\psi)(r,u)dr\Big)\tilde b(t,\psi)(T,u)1_{[t,\infty)}(T)
\end{equation}
and boundedness of $\partial_2\gamma$ one concludes that $a(t,\psi)\in\raum{E}$.
Moreover, $a$ is the composition of the measurable mapping
$\rp\times\raum{E}\to\rp\times\raum{E}$, $(t,\psi)\mapsto(t,b(t,\psi))$ 
and the continuous and hence measurable mapping
$\rp\times\raum{E}\to\raum{E}$, $(t,\psi)\mapsto f(t,\psi)$ defined by
$$f(t,\psi)(T,u)
=i\partial_2\gamma\Big(u,-i\int_t^T\widetilde \psi(r,u)dr\Big)\widetilde \psi(T,u)1_{[t,\infty)}(T),$$
where $\widetilde\psi$ is defined by truncating the real part of $\psi$ as in (\ref{e:tilb}).}

Let $T,T_0,m_0,\widetilde T,m,c$ be as in Definition \ref{D:Bausteine}(5) such that
\begin{eqnarray}\label{e:lipb}
 \Vert b(t,\psi_1)-b(t,\psi_2)\Vert_{\widetilde T,m} &\leq& c \Vert \psi_1-\psi_2\Vert_{\widetilde T,m}
\end{eqnarray}
for any $\psi_1,\psi_2\in \raum{E}$, $t\in[0,T]$. Let $B$ be a bounded set and
$L$ a Lipschitz constant of $\partial_2\gamma$. Then 
\begin{eqnarray*}
 H &:=& \sup_{(R,u)\in S} \left\vert i\partial_2\gamma\left(u,-i\int_t^R\widetilde\psi_1(r,u)dr\right)-i
\partial_2\gamma\left(u,-i\int_t^R\widetilde\psi_2(r,u)dr\right)\right\vert \\
 &\leq& L\sup_{(R,u)\in S} \left\vert \int_t^R\widetilde\psi_1(r,u)dr-\int_t^R\widetilde\psi_2(r,u)dr\right\vert \\
 &\leq& L\sup_{(R,u)\in S} \int_t^R\vert\psi_1(r,u)-\psi_2(r,u)\vert dr \\
 &\leq& L\Vert \psi_1-\psi_2\Vert_{\widetilde T,m}
\end{eqnarray*}
with $S=[t,\widetilde T]\times[-m,m]$. Let $c_1\in\mathbb R_+$ be a bound for the set $B$,
$$c_2 := \sup_{t\in[0,T]}\Vert b(t,0)\Vert_{\widetilde T,m} + cc_1,$$
and $c_3$ be a bound for $\partial_2\gamma$. Observe that
 $ \Vert b(t,\psi) \Vert_{\widetilde T,m} \leq c_2$
 for any $\psi\in B, t\in[0,T]$. Submultiplicativity of the uniform norm yields
\begin{eqnarray*}
 \lefteqn{\Vert f(t,\psi_1)-f(t,\psi_2)\Vert_{\widetilde T,m}}\\ 
&=& \int_0^{\widetilde T}\sup_{u\in[-m,m]}\vert (f(t,\psi_1)-f(t,\psi_2))(r,u)\vert dr \\
   & \leq & \int_0^{\widetilde T} \Big(H\sup_{u\in[-m,m]}\vert \psi_1(r,u)\vert \\
&&{}+ \sup_{u\in[-m,m]}\left\vert \partial_2\gamma\left(u,-i\int_t^R\widetilde\psi_2(r,u)dr\right)\right\vert 
\vert (\psi_1-\psi_2)(r,u)\vert \Big)dr\\
   &\leq& L\Vert \psi_1-\psi_2\Vert_{\widetilde T,m} \Vert\psi_1\Vert_{\widetilde T,m}
   +c_3 \Vert\psi_1-\psi_2\Vert_{\widetilde T,m} \\
   &\leq& (Lc_2+c_3) \Vert \psi_1-\psi_2\Vert_{\widetilde T,m}.
\end{eqnarray*}
for any $t\in[0,T]$ and any  $\psi_1,\psi_2\in\raum{E}$ which are bounded by $c_2$. 
Since $a$ is the composition of $(t,\psi)\rightarrow (t,b(t,\psi))$ and $f$ and since $b$ is Lipschitz continuous, 
the first inequality follows.

For the second inequality note that
\begin{eqnarray*}
\Vert a(t,\psi)  \Vert_{\widetilde T,m} 
&\leq& c_3\Vert b(t,\psi)\Vert_{\widetilde T,m} \\
&\leq&c_3\left(\Vert b(t,0)\Vert_{\widetilde T,m} 
+\Vert b(t,\psi)-b(t,0)\Vert_{\widetilde T,m}\right) \\
&\leq& c_3\left(\Vert b(t,0)\Vert_{\widetilde T,m} 
+c\Vert \psi\Vert_{\widetilde T,m}\right) \\
&\leq& C(1+\Vert\psi\Vert_{\widetilde T,m})
\end{eqnarray*}
by (\ref{e:a}, \ref{e:lipb}) for $C:=c_3 (\sup_{t\in[0,T]}\Vert b(t,0)\Vert_{\widetilde T,m} \vee c)$.

\item \sloppy{For fixed $\omega\in\Omega$
SDE (\ref{G:Psi SDE}) is a pathwise equation in the Fr\'echet space $(\raum{E},d)$, driven by two increasing
functions. Existence and uniqueness under the present Lipschitz and growth conditions follows from 
Corollary \ref{C:Ex und Eind Satz}.} Pathwise uniqueness of a solution to SDE (\ref{G:Psi SDE}) 
now implies uniqueness in law. This follows along the same lines as for $\rr^d$-valued SDE's driven by a Wiener process, 
cf.\ e.g.\ \cite[Theorem IX.1.7 and Exercise IV.5.16]{revuz.yor.99}. For the proof of \cite[Exercise IV.5.16]{revuz.yor.99} 
one may note that the law of the Bochner integral $\int_0^\cdot a(t,\Psi_{t-})dt$ (and likewise the law of the integral 
with respect to $M$) is determined by the law of all random vectors of the form 
$$\left(\int_0^{t_1}f_1(a(s,\Psi_{s-}))ds,\dots, \int_0^{t_d}f_d(a(s,\Psi_{s-}))ds\right),$$
where $d\in\mathbb N$, $t_1,\dots,t_d\in\rp$, and $f_1,\dots,f_d$ denote continuous linear functionals on $\raum{E}$.

\end{enumerate}
\end{proof}

We can now state an existence and uniqueness result for compatible option surface models.
The condition in Statement 2 of the following theorem means essentially that
\begin{itemize}
 \item  the current codebook state $(T,u)\mapsto\Psi_t(T,u)$ must look like the exponent of a PII
whose exponential is a martingale and
\item $u\mapsto\Psi_{t-}(t,u)$ is the local exponent of some process $X$ whose exponential is a
martingale and whose dependent part $\Xp$ relative to $M$ is of the form in Definition \ref{D:Bausteine}(8).
\end{itemize}
The first requirement makes sense because of the very idea of a codebook in the present
L\'evy setup. The second condition, on the other hand, naturally appears through the
consistency condition.

\begin{theorem}\label{Satz: starker Existenzsatz}
 Let $(x_0,\psi_0,b,\gamma)$ be building blocks.
\begin{enumerate}
 \item \sloppy{Any two compatible risk-neutral option surface models $\osm$ resp.\  $\osmzwei$ coincide in law, 
i.e.\ $(\Psi,X,M)$ and $(\widetilde\Psi,\widetilde X, \widetilde M)$ have the same law
on $\mathbb D(\raum{E})\times\mathbb D(\rr^2)$.}
\item If a compatible risk-neutral option surface model $\osm$ exists, then
the $\raum{E}$-valued process $\eta$ defined by
\begin{equation}\label{e:eta}
 \eta_t(T,u):=\Phi_t(T,u)-\gamma(u,0)1_{[0,t]}(T),\quad t,T\in\rp,\quad u\in\rr,
\end{equation}
has values in $\raum{\Pi}$. Here, $\Phi$ denotes 
the $\raum{E}$-valued solution to SDE \eqref{G:Psi SDE} from 
Proposition \ref{Proposition: SDE ist eindeutig}.
 \item Let $\Phi$ denote
the $\raum{E}$-valued solution to SDE \eqref{G:Psi SDE} from 
Proposition \ref{Proposition: SDE ist eindeutig}.
If the $\raum{E}$-valued process $\eta$ in (\ref{e:eta}) has values in $\raum{\Pi}$,
there exists a compatible risk-neutral option surface model.
\end{enumerate}
\end{theorem}
\begin{proof}
\begin{enumerate}
 \item 
{\em Step 1:}
Let $\osm$ be a risk-neutral option surface model which is compatible with the building blocks. 
Denote by $\Xp$ the dependent part of $X$ relative to $M$. By compatibility we have 
$\psi^{(\Xp,M)}=\gamma.$ Thus $(\Xp,M)$ is a L\'evy process. Let $(\scr G_t)_{t\in\mathbb R_+}$ 
be the filtration generated by $(\Xp,M)$, i.e.\ 
$$\scr G_t=\bigcap_{s>t}\sigma({(\Xp,M)}_r:r\leq s).$$
Since the option surface model $\osm$ is risk neutral, Theorem \ref{Satz: Bedingungen impl. Vollstaendig} 
yields that it satisfies the drift condition, the consistency condition and the conditional expectation condition. 
Moreover, compatibility and the drift condition imply
\begin{eqnarray*}
 \beta_t(T,u) &=& b(t,\Psi_{t-})(T,u),\\
 \alpha_t(T,u) &=& -\partial_T\left(\gamma\Big(u,-i\int_t^T b(t,\Psi_{t-})(r,u)dr\Big)\right)
\end{eqnarray*}
a.s\ for any $u\in\mathbb R$ and almost any $T\in\rp$,  $t\in[0,T]$. 
(\ref{Gleichung:neu2}) implies that $\Psi$ solves the SDE 
\begin{equation}\label{e:ba}
 d\Psi_t = a(t,\Psi_{t-})dt + b(t,\Psi_{t-})dM_t,\quad\Psi_0=\psi_0.
\end{equation}
pointwise for any $T\in\rp$, $u\in\rr$.
By compatibility, $\Psi_t,\Psi_{t-}$ are $\raum{E}$-valued random variables. 
It is not hard to see that Equation (\ref{e:ba}) holds also in the sense of 
$\raum{E}$-valued processes. 
By Proposition \ref{Proposition: SDE ist eindeutig} $\Psi$ is the unique pathwise solution to the SDE. 
Thus $\Psi$ is adapted to the filtration $(\scr G_t)_{t\in\mathbb R_+}$.

{\em Step 2:}
Define a filtration $(\scr H_t)_{t\in\rp}$ via 
$\scr H_t:=\bigcap_{s>t}(\F_s\vee\G_\infty)$.
We show that $X^\perp-X_0=X-X^\|-X_0$ is a $\G_\infty$-conditional PII with respect to filtration 
$(\scr H_t)_{t\in\rp}$.
Indeed, adaptedness follows from the fact that both $X$ and $X^\perp$ 
are adapted to the original filtration $(\F_t)_{t\in\rp}$. 
By definition of conditional PII's in \cite[Section II.6.1]{js.87} it remains to be shown that
\begin{equation}\label{e:CPII}
 E\left(f(X^\perp_r-X^\perp_s)ZY\right)=E\left(E\big(f(X^\perp_r-X^\perp_s)\big|\G_\infty\big)
E(Z|\G_\infty)Y\right)
\end{equation}
for any $s\le r$, any bounded measurable function $f:\rr\to\rr$, 
any bounded $\scr H_s$-measurable random variable $Z$ 
and any bounded $\G_\infty$-measurable function $Y$.
By right-continuity of $X^\perp$, it suffices to consider only
$\F_s\vee\G_\infty$-measurable $Z$.
Standard measure theory yields that we can focus on functions of the form
$f(x)=e^{iux}$ for any $u\in\rr$ and $Z$ of the form $Z=1_F1_G$ with $F\in\F_s$ and 
$G\in\G_\infty$.
In view of (\ref{e:CPII}), it even suffices to discuss $Z=1_F$ 
because the second factor can be moved to $Y$.
Moreover, we may replace $X^\perp_r$ by $X^\perp_{r\wedge\tau_n}$, 
where the $\G_\infty$-measurable stopping times $\tau_n,n\in\mathbb N$ are defined by
$$\tau_n:=\inf\left\{\tilde t\ge s:\Re\Big(\int_s^{\tilde t}\Psi_{t}(t,u)dt\Big)\ge n\right\}.$$
Finally, $\G_\infty$-mesurability of $\Psi$ implies that we can write $Y$ as 
$$Y=\widetilde Y\exp\left(-\int_s^{r\wedge\tau_n}(\Psi_t(t,u)-\gamma(u,0))dt\right)$$
with some bounded $\G_\infty$-measurable $\tilde Y$.
The consistency condition and local independence of $X^\|,X^\perp$ imply that
$$\Psi_t(t,u)-\gamma(u,0)=\psi^X_t(u)-\psi_t^{X^\|}(u)=\psi^{X^\perp}_t(u), \quad u\in\rr$$
outside some $dP\otimes dt$-null set.
As above, standard measure theory
 yields that it suffices to consider 
$\widetilde Y$ of the form
$$\widetilde Y=\exp\left(i\int_0^Tv(t)d(X^\|,M)_t\right)$$
with $T\in[r,\infty)$ and bounded measurable $v=(v_1,v_2):[0,T]\to\rr^2$.
If we set 
$\G^+_s:=\sigma((X^\|,M)_t-(X^\|,M)_s:t\ge s)$, we have $\G_\infty=\G_s\vee\G^+_{s}$.
Moreover, $\G^+_s$ is independent of $\F_s$ because $(X^\|,M)$ 
is a L\'evy process with respect to filtration $(\F_t)_{t\in\rp}$.
Since $Z=1_F$ is $\F_s$-measurable, we have 
$E(Z|\G_\infty)=E(Z|\G_s)$,
cf.\ e.g.\ \cite[Satz 54.4]{bauer.78}.
This yields
\begin{eqnarray}\label{e:indep1}
 \lefteqn{E\left(E\big(f(X^\perp_{r\wedge\tau_n}-X^\perp_s)\big|\G_\infty\big)
E(Z|\G_\infty)Y\right)}\nonumber\\
&=&E\left(f(X^\perp_{r\wedge\tau_n}-X^\perp_s)YE(Z|\G_\infty)\right)\nonumber\\
&=&E\left(f(X^\perp_{r\wedge\tau_n}-X^\perp_s)YE(Z|\G_s)\right)\nonumber\\
&=&E\left(E\big(f(X^\perp_{r\wedge\tau_n}-X^\perp_s)Y\big|\G_s\big)Z\right).
\end{eqnarray}
It remains to be shown that
$E(f(X^\perp_{r\wedge\tau_n}-X^\perp_s)Y|\F_s)$
is in fact $\G_s$-measurable because together with (\ref{e:indep1}) this implies 
\begin{eqnarray*}
E\left(E\big(f(X^\perp_{r\wedge\tau_n}-X^\perp_s)\big|\G_\infty\big)
E(Z|\G_\infty)Y\right)
&=&E\left(E\big(f(X^\perp_{r\wedge\tau_n}-X^\perp_s)Y\big|\F_s\big)Z\right)\\
&=&E\left(f(X^\perp_{r\wedge\tau_n}-X^\perp_s)ZY\right)
\end{eqnarray*}
as claimed in (\ref{e:CPII}).

To this end, note that \normalsize{
\begin{eqnarray*}
\lefteqn{f(X^\perp_{r\wedge\tau_n}-X^\perp_s)Y}\\
&=&\exp\Bigg(iu (X^\perp_{r\wedge\tau_n}-X^\perp_s)
+i\int_0^Tv(t)d(X^\|,M)_t-\int_s^{r\wedge\tau_n}\psi^{X^\perp}_t(u)dt\Bigg)\\
&=&\exp\Bigg(i\int_0^T(u1_{(s,r\wedge\tau_n]}(t),v_1(t),v_2(t))d(X^\perp,X^\|,M)_t\\
&&{}-\int_0^T\psi^{(X^\perp,X^\|,M)}_t(u1_{(s,r\wedge\tau_n]}(t),v_1(t),v_2(t))dt\Bigg)\\
&&{}\times\exp\left(\int_0^T\psi^{(X^\|,M)}_t(v_1(t),v_2(t))dt\right)\\
&=&M_TD,
\end{eqnarray*}}
where $\psi^{(X^\perp,X^\|,M)},\psi^{(X^\|,M)}$
 denote local exponents in the sense of Definition \ref{Definition: lokaler Exponent 1. Version},
\begin{eqnarray*}
 M&=&\exp\Bigg(i\int_0^\cdot(u1_{(s,r\wedge\tau_n]}(t),v_1(t),v_2(t))d(X^\perp,X^\|,M)_t\\
&&{}-\int_0^\cdot\psi^{(X^\perp,X^\|,M)}_t(u1_{(s,r\wedge\tau_n]}(t),v_1(t),v_2(t))dt\Bigg),
\end{eqnarray*}
and $D$ stands for the remaining factor.
Since $D$ is deterministic and
$M$ is a bounded local martingale and hence a martingale, we have
\begin{equation}\label{e:msd}
 E(M_TD|\F_s)=M_sD,
\end{equation}
which is $\G_s$-measurable as desired.

{\em Step 3:}
Using the notation of Step 2, we show that
\begin{equation}\label{e:cpiicf}
 E\big(f(X^\perp_r-X^\perp_s)\big|\G_\infty\big)
=\exp\left(\int_s^{r}(\Psi_t(t,u)-\gamma(u,0))dt\right).
\end{equation}
Indeed, first note that we may replace $r$ with $r\wedge\tau_n$, $n\in\mathbb N$ by
right-continuity. Choosing $\G_\infty$-measurable $Y$ as in Step 2, we obtain using
(\ref{e:msd}):
\begin{eqnarray*}
E\big(f(X^\perp_{r\wedge\tau_n}-X^\perp_s)Y\big) 
&=&D\\
&=&E(\widetilde Y)\\
&=&E\left(\exp\Big(\int_s^{r\wedge\tau_n}(\Psi_t(t,u)-\gamma(u,0))dt\Big)Y\right),
\end{eqnarray*}
which yields the assertion.

{\em Step 4:}
We now show uniqueness of the law of $(X^\perp,X^\|,M,\Psi)$, which implies 
uniqueness of the law of $(\Psi,X,M)$.
To this end, observe that $(X^\|,M,\Psi)$ is $\G_\infty$-measurable whereas 
the conditional law of $X^\perp$ given $\G_\infty$ is determined by
the fact that $X^\perp-X_0$ is a $\G_\infty$-conditional PII with conditional characteristic function
(\ref{e:cpiicf}). Therefore, it suffices to prove uniqueness of the law of $(X^\|,M,\Psi)$.
This uniqueness, on the other hand, follows from Step 1 and Statement 2 in 
Proposition \ref{Proposition: SDE ist eindeutig}.

\item In Step 1 of the proof of Statement 1 it is shown that the codebook process $\Psi$
of (\ref{Gleichung:neu2}) solves SDE (\ref{G:Psi SDE}), i.e.\ it coincides with $\Phi$.
It suffices to show $\int_t^T\eta_s(r,\cdot)dr\in\Pi$ separately for $s<t\le T$ and for $t\le T\le s$.

By definition of option surface models, we have $\int_t^T\Psi_s(r,\cdot)dr\in\Pi$ for $s\le t\le T$.
Since $\eta_s(r,\cdot)=\Psi_s(r,\cdot)$ for $s<r$, this yields $\int_t^T\eta_s(r,\cdot)dr\in\Pi$ for $s<t\le T$.

For $r\le s$ outside some Lebesgue-null set, we have
\begin{eqnarray*}
 \eta_s(r,u)&=&\Psi_r(r,u)-\gamma(u,0)\\
&=&\psi^X_r(u)-\gamma(u,0)\\
&=&\psi^{X^\|}_r(u)+\psi^{X^\perp}_r(u)-\gamma(u,0)\\
&=&\psi^{X^\perp}_r(u), \quad u\in\rr,
\end{eqnarray*}
where we used the consistency condition in the second equality.
By Lemma \ref{Lemma: Semimartingalprojektion} and Remark \ref{r:pi} 
$u\mapsto\psi^{X^\perp}_t(u)$ is in $\Pi$. Lemma \ref{l:banach}(3)
yields that $\int_t^T\eta_s(r,\cdot)dr\in\Pi$ for $t\le T\le s$.
\item
{\em Construction of the codebook process}: By Theorem \ref{S:Ex von PII} there is a 
L\'evy process $(\Xp,M)$ on a complete filtered probability space 
$(\Omega^{(1)},\scr F^{(1)},(\scr F^{(1)}_t)_{t\in\mathbb R_+},P^{(1)})$ 
such that its (extended) L\'evy exponent is $\gamma$. 
Let $\Psi$ be the $\raum{E}$-valued c\`adl\`ag solution to the SDE
$$d\Psi_t = a(t,\Psi_{t-})dt + b(t,\Psi_{t-})dM_t,\quad\Psi_0=\psi_0$$
given by Proposition \ref{Proposition: SDE ist eindeutig}. $\Psi$ is an
$\raum{\Pi}$-valued process because this even holds for $\eta$  by assumption. 
It is not hard to find versions 
\begin{eqnarray*}
 \alpha_t(T,u) &:=& a(t,\Psi_{t-})(T,u), \\
 \beta_t(T,u) &:=& b(t,\Psi_{t-})(T,u)
\end{eqnarray*}
 for any  $t,T\in\mathbb R_+$, $u\in\mathbb R$ and a version of $\Psi$ such that
\begin{equation}\label{e:SDEp}
\Psi_t(T,u)=\Psi_0(T,u)+\int_0^{t\wedge T}\alpha_s(T,u)ds
+\int_0^{t\wedge T}\beta_s(T,u)dM_s
\end{equation}
for any $t,T\in\mathbb R_+$, $u\in\mathbb R$ almost surely.
More precisely, one can choose versions of the $\raum{E}$-valued processes
$\Psi,(a(t,\Psi_{t-}))_{t\in\rp},(b(t,\Psi_{t-}))_{t\in\rp}$ such that 
for any $T,u$ the $\mathbb C$-valued process $\Psi(T,u) $ is adapted, almost surely c\`adl\`ag,
and satisfies (\ref{e:SDEp}) almost surely.

{\em Construction of the return process}: 
As usual, let $(\mathbb D,\D,(\scr D_t)_{t\in\mathbb R_+})$ denote the Skorokhod space 
of real-valued c\`adl\`ag functions. Let $X^\perp$ be the canonical process on $\mathbb D$
and set
\begin{eqnarray*}
 \Omega &:=& \Omega^{(1)}\times\mathbb D,\\
 \scr F &:=& \scr F^{(1)}\otimes\scr D,\\
 \scr F_t &:=& \bigcap_{s>t}\big(\scr F_s^{(1)}\otimes\scr D_s\big).
\end{eqnarray*}
Fix $\omega_1\in\Omega^{(1)}$. Theorem \ref{S:Ex von PII} yields 
that there is a probability measure 
$P^{(2)}(\omega_1,\cdot)$ on $(\mathbb D,\scr D)$ such that 
$X_0^\perp=x_0$ a.s.\ and 
$X^\perp-X^\perp_0$ is a PII with characteristic function 
\begin{equation}\label{e:etadoppelt}
 u\mapsto\exp\left(\int_0^t\eta_\infty(r,u)(\omega_1)dr\right)=\exp\left(\int_0^t\eta_t(r,u)(\omega_1)dr\right).
\end{equation}
Measurability of $\eta$ implies that  $P^{(2)}$ is a transition kernel from 
$(\Omega^{(1)},\F^{(1)})$ to $(\mathbb D,\scr D)$. 
Therefore, 
$$P(d(\omega_1,\omega_2)):=(P^{(1)}\otimes P^{(2)})(d(\omega_1,\omega_2))
:=P^{(1)}(d\omega_1)P^{(2)}(\omega_1,d\omega_2)$$ 
defines a probability measure $P$ on $(\Omega,\F)$.

By abuse of notation we will use the same letters for the process $M,\Psi,\Xp,X^\perp$ 
embedded in the filtered probability space 
$(\Omega,\scr F,(\scr F_t)_{t\in\mathbb R_+},P)$,
i.e.\ we denote e.g.\ the process 
$((\omega_1,\omega_2),t)\mapsto M_t(\omega_1)$
again by $M$. Set
\begin{equation}\label{e:Xnormal}
 X:=\Xp+X^\perp
\end{equation}

Observe that $X^\perp-X_0$ is an 
$\scr F^{(1)}\otimes \{\emptyset,\mathbb D\}$-conditional PII 
relative to the filtration $(\G_t)_{t\in\rp}$ defined by
$$\G_t:=\bigcap_{s>t}\big(\scr F^{(1)}\otimes\scr D_s\big).$$
Denote by $(b,c,K)$  its local characteristics relative to some truncation function 
$h$. Then \cite[Theorem II.6.6]{js.87} yields
$$ iub_t(\omega)-\frac{u^2}{2}c_t(\omega)+\int(e^{iux}-1-iuh(x))K_t(\omega,dx)
=\eta(t,u)(\omega_1)$$ 
for almost all $\omega=(\omega_1,\omega_2)\in\Omega$, $t\in\mathbb R_+$, $u\in\mathbb R$. 
Both $X^\perp$ and $(b,c,K)$ are $(\scr F_t)_{t\in\mathbb R_+}$-adapted. From Proposition 
\ref{P:erhalt der lokalen char} 
it follows that $X^\perp$ is a semimartingale with respect to this smaller filtration 
$(\scr F_t)_{t\in\mathbb R_+}$ with the same local characteristics $(b,c,K)$.

We want to show that $(\Xp,M)$ is a L\'evy process on $\fraum$ as well.
By right-continuity of $(\Xp,M)$ it suffices to prove that
$$E\big(U\big\vert\scr F_s^{(1)}\otimes\scr D_s\big)=E(U)$$
for any $s,t\in\mathbb R_+$ with $ s\leq t$ and any $u\in\mathbb R^2$,
where 
$$U:=\exp\left(iu\big((\Xp,M)_t-(\Xp,M)_s\big)\right).$$
It suffices to show that
$$E(UVW)=E(U)E(VW)$$
for any bounded 
$\scr F^{(1)}_s\otimes\{\emptyset,\mathbb D\}$-measurable $V$
and any bounded 
$\{\emptyset,\Omega^{(1)}\}\otimes \scr D_s$-measurable $W$.
The conditional law of $(X^\perp_r)_{r\le s}$ given 
$\scr F^{(1)}\otimes \{\emptyset,\mathbb D\}$ is
$\scr F^{(1)}_s\otimes \{\emptyset,\mathbb D\}$-measurable 
because $\eta(r,u)$ is $\scr F^{(1)}_s$-measurable for any $r\le s$.
This implies 
$$E\big(W\big|\scr F^{(1)}\otimes \{\emptyset,\mathbb D\}\big)
=E\big(W\big|\scr F^{(1)}_s\otimes \{\emptyset,\mathbb D\}\big)$$
because $W$ is a measurable function of $(X^\perp_r)_{r\le s}$.
Moreover, $U$ is independent of $\scr F^{(1)}_s\otimes \{\emptyset,\mathbb D\}$
because $(\Xp,M)$ is a L\'evy process on $\Omega^{(1)}$.
This yields 
\begin{eqnarray*}
E(UVW)
&=&E\left(UVE\big(W\big|\scr F^{(1)}\otimes \{\emptyset,\mathbb D\}\big)\right)\\
&=&E\left(UVE\big(W\big|\scr F^{(1)}_s\otimes \{\emptyset,\mathbb D\}\big)\right)\\
&=&E\left(E\big(U\big|\scr F^{(1)}_s\otimes \{\emptyset,\mathbb D\}\big)VW\right)\\
&=&E(U)E(VW)
\end{eqnarray*}
as desired.

{\em Compatibility of the constructed model}: We have $\beta_t=b(t,\Psi_{t-})$, $X_0=x_0$, $\Psi_0=\psi_0$ 
and $\gamma=\psi^{(\Xp,M)}$. We must show that the dependent part 
of $X$ relative to $M$ is $\Xp$. Since $\Xp$ is the dependent part of $\Xp$ relative to $M$, it remains to be 
shown that $M$ and $X^\perp$ are locally independent. Since $M$ is a subordinator, it suffices to prove that
$$P(\exists t\in\mathbb R_+:\Delta M_t\neq0,\Delta X^\perp_t\neq0)=0.$$
Let $J:=\{ t\in\rp:\Delta M_t\neq 0\}$ denote the set of jump times of $M$. Then $J$ is almost surely countable and
\begin{eqnarray*}
P\big(\exists t\in\mathbb R_+:\Delta M_t\neq0,\Delta X^\perp_t\neq0\big) 
&\leq&  E\left(\sum_{s\in J}1_{\{\Delta X^\perp_s\neq0\}}\right)\\
&=& E\left( \sum_{s\in J} P\big(\Delta X^\perp_s\neq 0\big\vert M\big)\right)\\
&=& 0
\end{eqnarray*}
because 
\begin{eqnarray*}
E\big(\exp(iu \Delta X^\perp_s)\big\vert \scr F^{(1)}\otimes\{\emptyset,\mathbb D\}\big) 
= \exp\left(\int_{s-}^s\eta(r,u)dr\right) =0,\quad u\in\rr
\end{eqnarray*}
and hence $P(\Delta X^\perp_s\neq 0\vert M)=0$.

{\em Risk neutrality of the constructed model}: 
The constructed model satisfies the consistency and the drift condition. 
Hence Theorem \ref{Satz: Bedingungen impl. Vollstaendig} yields risk-neutrality.

\end{enumerate}
\end{proof}

Examples illustrating the previous result are to be found in Sections \ref{Abschnitt:Verschwindendes Beta} 
and \ref{Abschnitt: Affine Modelle} below. In general, however, it is not obvious why the 
solution to SDE \eqref{G:Psi SDE} should satisfy 
the condition in Statement 3 of Theorem \ref{Satz: starker Existenzsatz}.
If this is not the case, a compatible risk-neutral option surface model does not exist. 
As a way out, we introduce a 
weaker form of compatibility, 
which assumes \eqref{G:Psi SDE} to hold only up to some maximal stopping time.
For related discussions on stochastic invariance problems, we refer the reader to \cite{filipovic.al.09,filipovic.al.12}.

\begin{definition}\label{Definition: schwach kompatibel}
\begin{enumerate}
 \item An option surface model $\osm$ is called {\em $\tau$-weakly compatible} with building blocks 
$(x_0,\psi_0,b,\gamma)$ if
\begin{itemize}
 \item $\tau$ is a stopping time,
 \item $X_0=x_0$,
 \item $\Psi_0 = \psi_0$,
 \item $\rp\to\raum{E}$, $t\mapsto \Psi_{t}(\omega)$ is  well defined and a.s.\ c\`adl\`ag,
 \item $\beta_{t}(\omega) = b(t,\Psi_{t-}(\omega))$ for $dP\otimes dt$-almost any 
$(\omega,t)\in [\![0,\tau]\!]$, 
 \item $\psi^{(X^\Vdash,M)}(u,v)=\gamma(u,v)$ for $(u,v)\in\mathbb R\times(\rr+i\rp)$,
where $X^\Vdash$ denotes some process which coincides on
$[\![0,\tau]\!]$
with  the dependent part $\Xp$ of $X$ relative to $M$.
\end{itemize}
\item Let $\osm$ denote an option surface model which is $\tau$-weakly compatible 
with building blocks $(x_0, \psi_0,b,\gamma)$. It is called {\em maximal weakly compatible} if
\begin{itemize}
 \item 
\begin{equation}\label{e:tau}
 \tau =\inf\big\{t\in\mathbb R_+:\eta_t\notin\raum{\Pi}\big\}\quad\text{a.s},
\end{equation}
where $\Phi$ denotes the unique $\raum{E}$-valued solution to  SDE \eqref{G:Psi SDE} 
from Proposition \ref{Proposition: SDE ist eindeutig} and $\eta$ is defined as in (\ref{e:eta}),
\item $t\mapsto\Psi_t(T,u)$ from (\ref{Gleichung:neu2}) stays constant after $\tau$.
\end{itemize}
\end{enumerate}
\end{definition}

We can now state our general existence and uniqueness result.

\begin{theorem}\label{Satz: schwacher Existenz und Eindeutigkeitssatz}
 Let $(x_0,\psi_0,b,\gamma)$ be building blocks. 
\begin{enumerate}
 \item There exists a maximal weakly compatible and risk-neutral option surface model.
 \item \sloppy{ Any two maximal weakly compatible and risk-neutral option surface models 
$\osm$, $\osmzwei$ coincide in law, i.e.\ $(\Psi,X,M)$ and $(\widetilde\Psi,\widetilde X, \widetilde M)$ have the same law
on $\mathbb D(\raum{E})\times\mathbb D(\rr^2)$.}
 \item If a compatible risk-neutral option surface model exists, then any maximal 
weakly compatible risk-neutral option surface model is in fact compatible.
\end{enumerate}
\end{theorem}
\begin{proof}
\begin{enumerate}
 \item 
 Let $(X^\Vdash,M)$ be a L\'evy process with characteristic exponent $\gamma$. 
Define $a$  as in Proposition \ref{Proposition: SDE ist eindeutig},
$\Phi$ as the unique $\raum{E}$-valued solution to SDE \eqref{G:Psi SDE},
and the $\raum{E}$-valued adapted c\`adl\`ag process $\eta$  as in (\ref{e:eta}).
\cite[Problem 2.9.1, Theorem 2.1.6]{ethier.kurtz.86} and \cite[Lemma I.1.19]{js.87} 
yield that there is a stopping time $\tau$ which satisfies Equation \eqref{e:tau}. 
Set 
\begin{eqnarray}
 \Psi_t &:=& \Phi_{t\wedge\tau},\nonumber\\
 \alpha_t(T,u) &:=& a(t,\Psi_{t-})(T,u)1_{[\![ 0,\tau]\!]}(t)\label{e:dreia},\\
 \beta_t(T,u) &:=& b(t,\Psi_{t-})(T,u)1_{[\![ 0,\tau]\!]}(t).\label{e:dreib}
\end{eqnarray}
More specifically, it is not hard to find versions of the right-hand sides of (\ref{e:dreia},\ref{e:dreib})
such that (\ref{e:SDEp}) holds up tp $\tau$.
We have $\Psi_0=\psi_0$.
Along the same lines as in the proof of Theorem \ref{Satz: starker Existenzsatz} 
we can now construct a return process $X$ such that $\osm$ is a $\tau$-weakly compatible 
and risk-neutral option surface model. 
More specifically, $\eta$ in (\ref{e:etadoppelt}) must be replaced by
$\eta^\tau_t(\cdot,\cdot):=\eta_{t\wedge\tau}(\cdot,\cdot)$
and $X$ in (\ref{e:Xnormal}) by $X:=(X^\Vdash)^\tau+X^\perp$.
The option surface model $\osm$ is in fact maximal weakly compatible.
\item 
\sloppy{By Proposition \ref{Proposition: SDE ist eindeutig}(2), 
the law of $(X^\Vdash,M,\Phi)$ is uniquely determined by $(x_0,\psi_0,b,\gamma)$, 
where $\Phi$ denotes the solution to SDE \eqref{G:Psi SDE}.
\cite[Problem 2.9.1, Theorem 2.1.6]{ethier.kurtz.86} and \cite[Lemma I.1.19]{js.87} 
yield that there is a stopping time $\tau$ which satisfies Equation \eqref{e:tau}.
Consider now the following stopped variant of SDE \eqref{G:Psi SDE}:}
 \begin{eqnarray}\label{G:Psi SDEtau}
d\Psi_t 
&=& a(t,\Psi_{t-})1_{[\![ 0,\tau]\!]}(t)dt+b(t,\Psi_{t-})1_{[\![ 0,\tau]\!]}(t)dM_t,\quad\Psi_0=\psi_0.
\end{eqnarray}
The uniqueness statement follows now along the same lines as in the proof of
Theorem \ref{Satz: starker Existenzsatz} if we replace \eqref{G:Psi SDE} by \eqref{G:Psi SDEtau},
$X^\|$ by $(X^\Vdash)^\tau$, and $\gamma(u,0)$ by $\gamma(u,0)1_{[\![ 0,\tau]\!]}$.
\item Let $\osm$ be a maximal weakly compatible and risk-neutral option surface model, 
$\Phi$ the unique $\raum{E}$-valued solution to SDE \eqref{G:Psi SDE}, and
$\tau$ as in (\ref{e:tau}).
Theorem \ref{Satz: starker Existenzsatz}(2) yields $\tau=\infty$
because there exists some compatible and risk-neutral option surface model.
This implies that $\osm$ is a compatible option surface model.

\end{enumerate}
\end{proof}

\subsection{Vanishing coefficient process $\beta$}\label{Abschnitt:Verschwindendes Beta}
The simplest conceivable codebook model \eqref{Gleichung:neu} is obtained 
for building blocks $(x_0,\psi_0,b,\gamma)$ where $b=0$ or equivalently $\gamma=0$. 
Not surprisingly, it leads to constant codebook processes and hence to the simple model class 
that we used to motivate option surface models in Section \ref{Abschnitt: ZeitinhomogeneLevys}.

\begin{corollary}\label{Satz: Existenzsatz fuer PII Modelle}
 Let $(x_0,\psi_0,b,\gamma)$ be building blocks with $\gamma=0$. 
Then there is a compatible risk-neutral option surface model $\osm$. 
For any such model, $X-X_0$ is a PII with characteristic function
$$E\big(\exp(iu(X_T-X_0))\big) 
= \exp\left(\int_0^T\psi_0(r,u)dr\right),\quad u\in\mathbb R, \quad T\in\rp.$$
In particular, the law of $X$ is uniquely determined.
\end{corollary}
\begin{proof}
Let $\Psi$ be the solution to SDE \eqref{G:Psi SDE} given by Proposition \ref{Proposition: SDE ist eindeutig}. 
Then we have 
$\Psi_t(T,u)=\psi_0(T,u)$
for all $t,T\in\mathbb R_+,u\in\mathbb R$ because $M=0$ and $\alpha=0$. In particular, 
$\Psi$ is an $\raum{\Pi}$-valued process. Thus the existence of a compatible risk-neutral option surface model $\osm$
follows from Theorem \ref{Satz: starker Existenzsatz}. Theorem \ref{Satz: Bedingungen impl. Vollstaendig} 
yields that $\osm$ satisfies the conditional expectation condition. Hence
$$E\left(e^{iu(X_T-X_t)}\Big\vert\scr F_t\right)=\exp\left(\int_t^T\psi_0(r,u)dr\right).$$
In particular, $X-X_0$ is a PII by definition.
\end{proof}

The Black-Scholes model is obtained for a particular choice of the initial state of the codebook.
\begin{example}[Black-Scholes model]
If we choose $\psi_0(T,u):=-(iu+u^2)\sigma^2/2$ in Theorem \ref{Satz: Existenzsatz fuer PII Modelle} 
for some $\sigma>0$, we obtain 
$E(e^{iuX_T})=\exp(iuX_0-iu{\frac{\sigma^2}{2}}T-{\frac{u^2}{2}}\sigma^2T)$,
which means $X_T\sim N(X_0-{\frac{\sigma^2}{2}}T,\sigma^2T)$, $T\in\mathbb R_+$.
Put differently, the return process $X$ is Brownian motion with drift rate $-\sigma^2/2$ and volatility $\sigma$.
\end{example}

\subsection{Deterministic coefficient process $\beta$}\label{Abschnitt: Affine Modelle}
In this section we consider building blocks $(x_0,\psi_0,b,\gamma)$ where $b$ depends on the time parameter only. 
Then $a$ defined as in Equation \eqref{G:Driftbedingung, Markovsch} also depends only on time. 
Thus SDE \eqref{G:Psi SDE} is solved by mere integration.
In the following, we omit the redundant argument $\psi$ and write $b(t)$, $a(t)$ for $b(t,\psi)$, $a(t,\psi)$, respectively.

\begin{corollary}\label{Korollar: Affine I}
 Let $(x_0,\psi_0,b,\gamma)$ be building blocks such that $b$ is constant in its second variable, 
i.e.\ $\beta_t(T,u):=b(t,\psi)(T,u)$ does not depend on $\psi$. 
\begin{enumerate}
 \item \sloppy{Then there exists a maximal weakly compatible risk-neutral option surface model. 
Any two such models $\osm, \osmzwei$ coincide in law,
i.e.\ $(\Psi,X,M)$ and $(\widetilde\Psi,\widetilde X, \widetilde M)$ have the same law
on $\mathbb D(\raum{E})\times\mathbb D(\rr^2)$.}
\item Let $a$ be defined as in  \eqref{G:Driftbedingung, Markovsch}.
The model in Statement 1 is compatible
if and only if the deterministic mappings
\begin{equation}\label{e:phi}
 \varphi_t:(T,u)\mapsto\psi_0(T,u) + \int_0^t a(s)(T,u)ds - \gamma(u,0)1_{[0,t]}(T)
\end{equation}
are in $\raum{\Pi}$ for any $t\in\rp$.
\end{enumerate}
\end{corollary}
\begin{proof}
\begin{enumerate}
\item The first assertion follows from Theorem \ref{Satz: schwacher Existenz und Eindeutigkeitssatz}. 
\item $\Rightarrow$:
Suppose that a compatible risk-neutral option surface model exists.
Fix $t\in\rp$. Since the mapping $\varphi_t$ in (\ref{e:phi}) is in $\raum{E}$, it remains to be shown that
$\int_{T_1}^{T_2}\varphi_t(r,\cdot)dr\in\Pi$ for any $t\in\rp$ and any $T_1\le T_2$.
For $\eta_t$ as in (\ref{e:eta}) we have that
\begin{eqnarray}\label{e:etaf}
 \int_{T_1}^{T_2}\eta_t(r,\cdot)dr-\int_{T_1}^{T_2}
 \varphi_t(r,\cdot)dr&=&\int_{T_1}^{T_2}\int_0^{t\wedge T_2}b(s)(r,\cdot)dM_sdr\nonumber\\
 &=&\int_0^{t\wedge T_2}\int_{T_1}^{T_2}\beta_s(r,\cdot)drdM_s
\end{eqnarray}
is in $\Pi$ by Lemma \ref{l:banach}(5).
Corollary \ref{Korollar: essinf und beschraenkte Integration} yields
$$\bigg\vert\int_0^{t\wedge T_2}\int_{T_1}^{T_2}\beta_s(r,u)drdM_s\bigg\vert(\omega_n)
\underset{n\rightarrow\infty}{\longrightarrow}0$$
for some sequence $(\omega_n)_{n\in\mathbb N}$ in $\Omega$. The proof actually shows that the same sequence 
can be chosen for all $u\in\mathbb R$. Hence there is a sequence $(\omega_n)_{n\in\mathbb N}$ in $\Omega$ 
such that $$\left( \int_{T_1}^{T_2}\eta_t(r,\cdot)dr\right)(\omega_n)\in\Pi$$
for all $n\in\mathbb N$ and 
$$\lim_{n\to\infty}\left( \int_{T_1}^{T_2}\eta_t(r,u)dr\right)(\omega_n)=\int_{T_1}^{T_2}\varphi_t(r,u)dr$$ 
for all $u\in\mathbb R$. 
Since the continuous mapping $u\mapsto\int_{T_1}^{T_2}\varphi_t(r,u)dr$ 
is the pointwise limit of characteristic exponents 
of infinitely divisible distributions, \cite[Theorem 5.3.3]{lukacs.70} yields that it is a 
characteristic exponent of an infinitely divisible distribution as well. 
By (\ref{e:etaf}) it is a difference of functions in $\Pi$ and hence itself in $\Pi$.

$\Leftarrow$: Suppose conversely that the mappings $\varphi_t$ in (\ref{e:phi}) are in $\raum{\Pi}$. 
Since
\begin{eqnarray*}
  \eta_t(T,u)&=&\Phi_t(T,u)-\gamma(u,0)1_{[0,t]}(T)\\
  &=&\varphi_t(T,u)+\int_0^tb(t)(T,u)dM_t,
 \end{eqnarray*}
 the assumption and Lemma \ref{l:banach}(5) yield
  $\eta_t \in\raum{\Pi}$.

\end{enumerate}
\end{proof}

The condition in Statement 2 of Corollary \ref{Korollar: Affine I} means that the initial codebook state $\psi_0$ 
must be greater or equal than
$$\mu_t:(T,u)\mapsto\gamma(u,0)1_{[0,t]}(T)-\int_0^t a(s)(T,u)ds$$
in the sense that $\psi_0-\mu_t\in\raum{\Pi}$ for any $t\in\rp$.
Put differently, the initial option prices must be large enough to 
allow for a compatible risk-neutral option surface model. 

\begin{remark}
For deterministic $\beta$ as in Corollaries \ref{Satz: Existenzsatz fuer PII Modelle} and \ref{Korollar: Affine I} 
it may not be obvious
why one should require the c\`adl\`ag property of the codebook in
Definition \ref{Definition: kompatibel}. However, in this case it holds automatically.
Indeed, let $(x_0,\psi_0,b,\gamma)$ be building blocks and $\osm$ a compatible risk-neutral option surface model. 
Then $\alpha$ and $\beta$ are deterministic and
$$\Psi_t(T,u) = \Psi_0(T,u)+\int_0^t\alpha_s(T,u)ds + \int_0^t\beta_s(T,u)dM_s.$$
It is not hard to conclude that
$$\Psi_t = \Psi_0+\int_0^t\alpha_sds + \int_0^t\beta_sdM_s$$
holds in the sense of $\raum{E}$-valued integrals as well.
Since the right-hand side is c\`adl\`ag, $t\mapsto \Psi_t$ is c\`adl\`ag as well.
\end{remark}

If $b$ in Corollary \ref{Korollar: Affine I} is constant in Musiela parametrisation, 
i.e.\ if it is of the form (\ref{e:musiela}) with a constant $\check b$, 
the compatibility condition in Statement 2
of this Corollary can be simplified.

\begin{lemma}\label{l:local}
   Assume that $b$ is of the form
   $$b(t)(T,u)=\check b(T-t,u)1_{[t,\infty)}(T), \quad t,T\in\rp, u\in\rr$$
   for some $\check b\in\raum{\Pi}$
   and let $a$ be defined by
   Equation \eqref{G:Driftbedingung, Markovsch} as usual.
   Then 
    \begin{equation}\label{e:mindest}
  (T,u)\mapsto -\int_0^Ta(s)(T,u)ds + \gamma(u,0)
 \end{equation}
 is in $\raum{\Pi}$.
   If the mapping $\varphi_\infty\in\raum{E}$ defined by
   $$\varphi_\infty(T,u) := \psi_0(T,u) + \int_0^Ta(s)(T,u)ds - \gamma(u,0)$$
   is in $L^1(\mathbb R_+,\Pi)$ as well, then $\varphi_t$ from (\ref{e:phi})
   is in  $L^1(\mathbb R_+,\Pi)$  for any $t\in\rp$.
 \end{lemma}
 \begin{proof}
  First note that $\varphi_t\in L^1(\mathbb R_+,E)$ for any $t\in\overline\rr_+$.
Since
  \begin{eqnarray*}
   a(t)(T,u) &=& -\partial_T\left(\gamma\left(u,-i\int_{t\wedge T}^T\check b(r-t,u)dr\right) \right)\\
             &=& -\partial_T\left(\gamma\left(u,-i\int_{(T-t)\wedge 0}^{T-t}\check b(r,u)dr\right)\right) \\
             &=& \partial_t\left(\gamma\left(u,-i\int_{(T-t)\wedge 0}^{T-t}\check b(r,u)dr\right)\right)
  \end{eqnarray*}
  we have
  \begin{equation} \label{e:inta}
   \int_0^Ta(s)(T,u)ds =\gamma(u,0)-\gamma\left(u,-i\int_{0}^{T}\check b(r,u)dr\right)
  \end{equation}
and  
  \begin{eqnarray*}
    \int_0^ta(s)(T,u)ds &=& \gamma\left(u,-i\int_{(T-t)\wedge 0}^{T-t}\check b(r,u)dr\right)
    -\gamma\left(u,-i\int_{0}^{T}\check b(r,u)dr\right) \\
           &=& \gamma\left(u,-i\int_{(T-t)\wedge 0}^{T-t}\check b(r,u)dr\right) + \int_0^Ta(s)(T,u)ds - \gamma(u,0)
  \end{eqnarray*}
  for any $t,T\in\mathbb R_+$, $u\in\mathbb R$. Consequently
  \begin{equation}\label{e:phigamma}
   \varphi_t(T,u) = \varphi_\infty(T,u) + \gamma\left(u,-i\int_{0}^{T-t}\check b(r,u)dr\right)1_{[t,\infty)}(T) \\
  \end{equation}
  for any $t,T\in\mathbb R_+$, $u\in\mathbb R$. \cite[Theorem 30.1]{sato.99} 
  yields that 
  $$u\mapsto\gamma\left(u,-i\int_{0}^{T-t}\check b(r,u)dr\right)1_{[t,\infty)}(T)$$ 
  is in $\Pi$ and hence the second summand of $\varphi_t$ in (\ref{e:phigamma})
  is in $\raum{\Pi}$ by Lemma \ref{l:banach}(4). 
  Since this holds for the first summand as well, the second statement follows.
  Similarly, we have that the last term in (\ref{e:inta}) is in $\raum{\Pi}$,
  which yields the first statement.
 \end{proof}

 Since (\ref{e:mindest}) is in $\raum{\Pi}$,  the condition
 in Lemma \ref{l:local} means that the initial codebook $\psi_0$ must be the sum
 of this minimal codebook (\ref{e:mindest}) and any other element of $\raum{\Pi}$.
 
If $\beta$ is of product form, we can establish a link to affine Markov processes in the sense of \cite{filipovic.05}.

\begin{theorem}\label{Satz: BNS artige Modelle}
 Let $(x_0,\psi_0,b,\gamma)$ be building blocks such that
\begin{enumerate}
 \item $$b(t)(T,u)=\varphi(u)\exp\left(-\int_{t}^T\lambda(s)ds\right)1_{[t,\infty)}(T)$$
for some $\varphi\in\Pi$ and some continuous function $\lambda:\mathbb R_+\rightarrow\mathbb R$,
\item $\psi_0(T,u)$ is continuous in $T$ for fixed $u$,
\item the compatibility condition  in Statement 2 of Corollary \ref{Korollar: Affine I} holds.
\end{enumerate}
Then $(X,Z)$ is a time-inhomogeneous affine process in the sense of \cite{filipovic.05},
where $\osm$ denotes a compatible risk-neutral option surface model (cf.\ Corollary \ref{Korollar: Affine I})
and 
$$Z_t:=\int_0^t\exp\left(-\int_s^t\lambda(s)ds\right)dM_s.$$
\end{theorem}
\begin{proof}
Let $\Xp$ be the dependent part of $X$ relative to $M$ and $X^{\perp}:=X-\Xp$. 
The local exponent of $X^\perp$ satisfies
\begin{eqnarray*}
\psi_t^{X^\perp}(u)&=&\psi_t^{X}(u)-\psi^{\Xp}(u)\\
&=&\Psi_0(t,u)+\int_0^t\alpha_s(t,u)ds+\int_0^{t-}\beta_s(t,u)dM_s-\psi^{\Xp}(u)\\
&=&\Psi_0(t,u)+\int_0^t\alpha_s(t,u)ds+\varphi(u) Z_{t-}-\psi^{\Xp}(u)\\
&=&\Psi_0(t,u)-\int_0^t\partial_t\left(\gamm\bigg(u,-i\int_s^t\beta_s(r,u)dr\bigg)\right)ds\\
&&{}+\varphi(u) Z_{t-}-\psi^{\Xp}(u),
\end{eqnarray*}
where we used the consistency condition in the second and the
drift condition \eqref{e:letztedrift} in the last equality.
Since
$$dZ_t=-\lambda(t)Z_tdt+dM_t$$
and by Lemma \ref{Lemma: Semimartingalprojektion}, we obtain for the local exponent of $(X,M,Z)$:
\begin{eqnarray*}
\psi^{(X,M,Z)}_t(u,v,w)
&=&\psi^{(X,M)}_t(u,v+w)-iw\lambda(t) Z_{t-}\\
&=&\psi^{(\Xp,M)}(u,v+w)-iw\lambda(t) Z_{t-}+\psi^{X^\perp}_t(u)\\
&=&\Phi_0(t;u,v,w)+\Phi_1(t;u,v,w) Z_{t-}
\end{eqnarray*}
with
\begin{eqnarray*}
\Phi_0(t;u,v,w)&:=&\psi^{(\Xp,M)}(u,v+w)-\psi^{\Xp}(u)\notag\\
&&{}+\Psi_0(t,u)-\int_0^t\partial_t\left(\gamm\bigg(u,-i\int_s^t\beta_s(r,u)dr\bigg)\right)ds,\label{e:Phi0}\\
\Phi_1(t;u,v,w)&:=&\varphi(u)-iw\lambda(t).\label{e:Phi1}
\end{eqnarray*}
This implies that $(u,v,w)\mapsto\Phi_0(t;u,v,w)+\Phi_1(t;u,v,w)Z_{t-}$ is a L\'evy exponent on $\mathbb R^{3}$. 
Since $\essinf M_t=0$ we have $\essinf Z_t=0$ by Corollary \ref{Korollar: essinf und beschraenkte Integration}. 
At the end of this proof we show that $(u,v,w)\mapsto\Phi_0(t;u,v,w)$ is a L\'evy exponent for fixed $t$. 
The same holds for $\Phi_1$. 
Relative to some truncation function $h$, denote by 
$(\beta_t^{(0)},\gamma_t^{(0)},\kappa_t^{(0)})$, $(\beta_t^{(1)},\gamma_t^{(1)},\kappa_t^{(1)})$ 
L\'evy-Khintchine triplets on $\mathbb R^{3}$ which correspond to $\Phi_0(t;\cdot)$ and $\Phi_1(t;\cdot)$ respectively. 
Observe that $\Phi_0(t;u,v,w)$ and $\Phi_1(t;u,v,w)$ are continuous in $t$ for fixed $(u,v,w)$. 
L\'evy's continuity theorem and \cite[Theorem VII.2.9]{js.87}  imply that 
$(\beta_t^{(0)},\gamma_t^{(0)},\kappa_t^{(0)})$, $(\beta_t^{(1)},\gamma_t^{(1)},\kappa_t^{(1)})$ 
are continuous in $t$ in the sense of Conditions $[\beta_1],[\gamma_1],[\delta_{1,3}]$ in that theorem. 
A detailed inspection of the arguments shows that this weaker continuity suffices for the 
proof of \cite[Proposition 4.1]{filipovic.05}. The assertion follows now from \cite[Theorem 2.14]{filipovic.05}.

Let $t\in\mathbb R_+$. Since $\essinf Z_t=0$ there is a sequence $\omega_n\in\Omega$ such that 
$\psi_n:=\psi^{(X,M,Z)}_t(\omega_n)$ is a L\'evy exponent and $Z_{t-}(\omega_n)\rightarrow 0$
for $n\to\infty$. Then $\psi_n\rightarrow \Phi_0(t;\cdot)$ locally uniformly. 
Thus \cite[Proposition 2.5 and  Lemma 7.8]{sato.99} yield that $\Phi_0(t;\cdot)$ is a L\'evy exponent.
\end{proof}

Finally, we consider more specific choices of $\varphi,\lambda$.

\begin{corollary}\label{Korollar: BNS}
Let $(x_0,\psi_0,b,\gamma)$ be building blocks such that 
\begin{eqnarray*}
b(t,\psi)(T,u) &=& b(t)(T,u) :=\varphi(u)e^{-\lambda(T-t)}1_{[t,\infty)}(T),\\
\gamma(u,v)&=&\eta(\delta u+v)-iu\eta(-\delta i)
\end{eqnarray*}
for $t,T\in\rp,  \psi\in\raum{\Pi},u\in\rr, v\in\rr+i\rp$,
where $\lambda\in(0,\infty)$, $\delta\in\mathbb R_-$, 
$\varphi(u):=-(u^2+iu)/2$, $u\in\mathbb C$, 
and $\eta:\rr+i\rp\to\mathbb C$ denotes the extended L\'evy exponent 
of a pure-jump subordinator with finite second moments.
Suppose that
$$\psi_0(T,u)=\psi^L(T,u) + \eta\left(\delta u-i\varphi(u)\frac{1-e^{-\lambda T}}{\lambda}\right)
-iu\eta(-\delta i)$$
for some  $\psi^L\in\raum{\Pi}$ 
(which implies $\psi_0\in\raum{\Pi}$ because it is the sum of two objects in $\raum{\Pi}$).

Then there is a compatible risk-neutral  option surface model $\osm$. Moreover, it can be chosen 
such that there is a standard Wiener process $W$ and a time-inhomogeneous L\'evy process $L$ 
with characteristic function 
$$E(e^{iuL_T})=\exp\left(\int_0^T\psi^L(r,u)dr\right),\quad  T\in\mathbb R_+,u\in\mathbb R,$$
$W, L, M$ are independent, and
\begin{eqnarray}
 dX_t & = & dL_t-\bigg(\frac{1}{2}Z_t+\eta(-\delta i)\bigg)dt+\sqrt{Z_t}dW_t+\delta d M_t,\label{e:bns1}\\
 dZ_t & = & -\lambda Z_tdt+dM_t\label{e:bns2}
\end{eqnarray}
holds with $X_0=x_0$, $Z_0=0$.
\end{corollary}
\begin{proof}
{\em Step 1:}
Let $Y$ be a L\'evy process with L\'evy exponent 
$$\psi:\mathbb C\rightarrow\mathbb C,\quad u\mapsto i\delta u+\varphi(u)\frac{1-e^{-\lambda T}}{\lambda}$$
and $M$ an independent subordinator with exponent $\eta$. Set $U_t:=Y_{M_t}-t\eta(-i\delta)$. 
Observe that $\eta(-i\delta)\in\mathbb R_-$. Then \cite[Theorem 30.1]{sato.99} yields that $U$ is a 
L\'evy process with L\'evy exponent
$$u\mapsto \eta\left(\delta u-i\varphi(u)\frac{1-e^{-\lambda T}}{\lambda}\right)-iu\eta(-\delta i).$$
Moreover, \cite[Theorem 30.1]{sato.99} also implies that 
$$P(U_t\in B) = \int P(Y_s-t\eta(-i\delta)\in B) P^{M_t}(ds)$$
for any $B\in\scr B$. Thus
\begin{eqnarray*}
E (e^{U_1}) &=& e^{-\eta(-i\delta)}\int E(e^{Y_s}) P^{M_1}(ds) \\ 
 &=& e^{-\eta(-i\delta)}\int \exp(\psi(-i)s) P^{M_1}(ds) \\
 &=& e^{-\eta(-i\delta)}\int \exp(\delta s) P^{M_1}(ds) \\
 &=& e^{-\eta(-i\delta)}\exp(\eta(-i\delta)) \\
 &=& 1,
\end{eqnarray*}
which implies that the L\'evy exponent of $U$ is an element of $\Pi$. Lemma \ref{l:banach}(4) yields that
$$(T,u)\mapsto \eta\left(\delta u-i\varphi(u)\frac{1-e^{-\lambda T}}{\lambda}\right) -iu\eta(-\delta i)$$
is an element of $\raum{\Pi}$. 

{\em Step 2:}
Let $W,L,M$ be independent L\'evy processes on a filtered probability space $\fraum$ such that $W$ 
is a Brownian motion, $L$ is a PII with characteristic function
$$E(e^{iuL_T})=\exp\left(\int_0^T\psi^L(r,u)dr\right),\quad  T\in\mathbb R_+,u\in\mathbb R,$$
and $M$ a L\'evy process with characteristic exponent $\eta$. Let $(X,Z)$ be a solution to the system of 
SDE's (\ref{e:bns1}, \ref{e:bns2}). The dependent part of $X$ relative to $M$ is 
$$\Xp:=(\delta M_t-\eta(-\delta i)t)_{t\in\rp}$$ 
because $W,L$ are independent of $M$
and 
$\psi^{X^\|}(-i)=\eta(-\delta i)-\eta(-\delta i)=0$.
Moreover, $(\Xp,M)$ has L\'evy exponent $\gamma$. Define 
\begin{eqnarray*}
\Psi_t(T,u) &:=& \psi_0(T,u) + \int_0^ta(s)(T,u)ds + \varphi(u)e^{-\lambda(T-(t\wedge T))}Z_{t\wedge T}\\
  &=&\psi_0(T,u) + \int_0^ta(s)(T,u)ds + \int_0^tb(s)(T,u)dM_s
\end{eqnarray*}
with 
$$a(t)(T,u):=\eta'\left(\delta u-i\varphi(u)\frac{1-e^{-\lambda(T-t)}}{\lambda}\right)
i\varphi(u)e^{-\lambda(T-t)}1_{[t,\infty)}(T),$$
cf.\ Section \ref{su:diff}.

By local independence of $W,L,M$ (cf.\ Corollary \ref{co:PIISPII}) we have
\begin{eqnarray*}
\psi_t^X(u) &=& \psi_t^L(u) + \varphi(u)Z_{t-} + \eta(\delta u)-iu\eta(-\delta i)\\
            &=& \psi_0(t,u) + \int_0^ta(s)(t,u)ds + \varphi(u)Z_{t-}\\
            &=& \Psi_{t-}(t,u).
\end{eqnarray*}
The condition on the initial codebook implies that the mapping
$$(T,u)\mapsto  \psi_0(T,u) + \int_0^Ta(s)(T,u)ds - \gamma(u,0)=\psi^L(T,u)$$
is in $\raum{\Pi}$.
By Lemma \ref{l:local} we have that
$\varphi_t$ from (\ref{e:phi}) is in $\raum{\Pi}$ as well for any $t\in\rp$.
As in the proof of Corollary \ref{Korollar: Affine I}(2) it follows that
$\Psi_t$ has values in $\raum{\Pi}$ for any $t\in\rp$.

Thus $\osm$ with $\alpha_t(T,u):=a(t)(T,u)$, $\beta_t(T,u):=b(t)(T,u)$ is a
compatible option surface model which satisfies the consistency condition and the
drift condition (\ref{e:letztedrift}). 
By Theorem \ref{Satz: Bedingungen impl. Vollstaendig} it is risk neutral,
which yields the claim.
\end{proof}
\begin{remark}
Up to the additional time-inhomogeneous L\'evy process $L$, the stock price model in
(\ref{e:bns1}, \ref{e:bns2}) is a special case of the so-called BNS model of \cite{barndorff.shephard.01}.
If we consider more general functions $\varphi$, then, again up to the additional PII $L$, we end up 
with the CGMY extension of the BNS model from \cite{carr.al.03}, cf.\ also \cite{kallsen.04}. 
\end{remark}

\section{Carmona \& Nadtochiy's 'Tangent L\'evy market models'}\label{Abschnitt: CNs Modell}
In \cite{carmona.nadtochiy.12} and its extension \cite[Section 5]{carmona.nadtochiy.11}, 
Carmona and Nadtochiy (CN) developed independently a HJM-type approach for option prices with overlap to ours. 
Their simple model class in the sense of Step~(4) in Section \ref{Abschnitt2: HJM philosophie} 
is based on time-inhomogeneous L\'evy processes as well. These can be described uniquely by their 
L\'evy density and their diffusion coefficient because the drift is determined by the martingale condition 
for the stock under the risk neutral measure. Instead of the characteristic exponent from \eqref{Gleichung: 2.3} 
CN use this L\'evy density together with the diffusion coefficient as the codebook $(\kappa_t(T,x),\Sigma_t(T))$. 
Since we basically allow for the same class of simple models, their framework can be embedded into ours. 
Indeed, there is a transformation $A$ that converts their codebook into ours, given by
$$A(\kappa_t(T,x),\Sigma_t(T)):=-\frac{u^2+iu}{2}\Sigma^2_t(T)+\int(e^{iux}-1-iu(e^x-1))\kappa_t(T,x)dx.$$
Since the simple models are parametrised differently, the drift condition in the two approaches differ. 
The condition in the CN framework looks a little more complex because it involves convolutions and 
differential operators of second order.

CN focus on \ito\ processes for modelling the codebook process, which roughly corresponds to choosing $M$ as 
Brownian motion in our setup. Surprisingly, their approach leads to a constant diffusion coefficient 
$\Sigma_t(T)=\Sigma_0(T)$, cf.\ \cite[Section 5]{carmona.nadtochiy.11}.  
This constant diffusion coefficient implies that the continuous martingale part of the stock price process 
follows a time-inhomogeneous Brownian motion rather than a more general continuous semimartingale.
This phenomenon does not occur in our setup
if the codebook is driven by a subordinator $M$, cf.\ e.g.\ Corollary \ref{Korollar: BNS}.

With regards existence and uniqueness of models given basic building blocks, CN and we provide different answers. 
Our Theorems \ref{Satz: starker Existenzsatz} and \ref{Satz: schwacher Existenz und Eindeutigkeitssatz} 
imply existence and uniqueness for a subordinator $M$ and 
given $\beta_t$ is a sufficiently regular function of time and the current state 
of the codebook. By contrast, CN consider a different situation in their \cite[Theorem 16]{carmona.nadtochiy.12} 
where they assume that the process $\beta$ in their codebook dynamics
$$d\kappa_t=\alpha_tdt+\beta_tdB_t,$$
is given beforehand. This does not allow for the case that $\beta$ depends on the current state $\kappa$ 
of the codebook itself, which occurs e.g.\ in the example in Section 6 of \cite{carmona.nadtochiy.11} 
and is treated separately.

Both CN and we provide basically one non-trivial example, based on more or less deterministic $\beta$. 
In order to ensure existence of a compatible option surface model we assume the initial codebook to be 
large enough whereas CN slow down the codebook process when necessary.

\appendix
\section{Local characteristics and local exponents}\label{s:localchar}
 In this section we define and recall some properties of local characteristics and local exponents. 
 
\subsection{Local characteristics}\label{su:lc}
Let $X$ be an $\mathbb R^d$-valued semimartingale with integral characteristics $(B,C,\nu)$ in the sense of \cite{js.87} 
relative to some fixed truncation function $h:\mathbb R^d\rightarrow\mathbb R^d$. By \cite[I.2.9]{js.87} 
there exist a predictable $\mathbb R^d$-valued process $b$, a predictable 
$\mathbb R^{d\times d}$-valued process $c$, a kernel $K$ from 
$(\Omega\times\mathbb R,\scr P)$ to $(\mathbb R^d,\B)$, and a predictable increasing process $A$ such that 
$$dB_t=b_tdA_t,\quad dC_t=c_tdA_t,\quad \nu(dt,dx)=K_t(dx)dA_t.$$ 
If $A_t=t$, we call the triplet $(b,c,K)$ {\em local} or {\em differential characteristics} of $X$ 
relative to truncation function $h$. Most processes in applications as e.g.\ diffusions, L\'evy processes etc.\ 
allow for local characteristics. In this case $b$ stands for a drift rate, $c$ for a diffusion coefficient, and $K$ 
for a local L\'evy measure representing jump activity. If they exist, the local characteristics are unique up to a 
$dP\otimes dt$-null set on $\Omega\times\mathbb R_+$.

\begin{proposition}[\ito's formula for local characteristics]\label{Prop: Ito-formel diff}
Let $X$ be an $\mathbb R^d$-valued semimartingale with local characteristics $(b,c,K)$ and 
$f:\mathbb R^d\rightarrow\mathbb R^n$ a $C^2$-function. Then the triplet $(\widetilde b,\widetilde c,\widetilde K)$ 
defined by
\begin{eqnarray*}
 \widetilde b_t & = & Df(X_{t-})^\top b_t+\frac{1}{2}\sum_{j,k=1}^n\partial_j\partial_kf(X_{t-})c_t^{jk}\\
&&{}+\int\left( \widetilde h(f(X_{t-}+x)-f(X_{t-}))-(Df(X_{t-}))^\top h(x)\right)K_t(dx),\\
 \widetilde c_t & = & (Df(X_{t-}))^\top c_tDf(X_{t-}), \\
 \widetilde K_t(A) & = & \int1_A(f(X_{t-}+x)-f(X_{t-})) K_t(dx),\quad A\in\B^n\mbox{ with }0\notin A,
\end{eqnarray*}
is a version of the local characteristics of $f(X)$ with respect to a truncation function $\widetilde h$ on $\mathbb R^n$. 
Here, $\partial_j$ etc. denote partial derivatives relative to the $j$'th argument.
\end{proposition}
\begin{proof}
 See \cite[Proposition 2.5]{kallsen.04}.

\end{proof}

\begin{proposition}\label{Prop: Integrationsformel diff}
 Let $X$ be an $\mathbb R^d$-valued semimartingale with local characteristics $(b,c,K)$ and let 
$\beta=(\beta^{ij})_{i\in\{1,\dots,d\},j\in\{1,\dots,n\}}$ be a $\mathbb R^{d\times n}$-valued 
predictable process such that $\beta^{\cdot i}\in L(X)$ for $i\in\{1,\dots,n\}$. Then the triplet 
$(\widetilde b,\widetilde c,\widetilde K)$ defined by
\begin{eqnarray*}
 \widetilde b_t & = & \beta_t^\top b _t+ \int \left(\widetilde h(\beta_t^\top  x)-\beta_t^\top  h(x)\right)K_t(dx),\\
 \widetilde c_t & = & \beta_t^\top  c_t \beta_t,\\
 \widetilde K_t(A) & = & \int1_A(\beta_t^\top  x) K_t(dx),\quad A\in\B^n\mbox{ with }0\notin A,
\end{eqnarray*}
is a version of the local characteristics of the $\mathbb R^n$-valued semimartingale 
$\beta\mal X:=(\beta^{\cdot 1}\mal X,\dots,\beta^{\cdot n}\mal X)$  
with respect to the truncation function $\widetilde h$ on $\mathbb R^n$,
\end{proposition}
\begin{proof}
 See \cite[Proposition 2.4]{kallsen.04}.

\end{proof}

\subsection{Local exponents}\label{Unterabschnitt: L\'evy exponenten}
\begin{definition}\label{d:levyexp}
 Let $(b,c,K)$ be a L\'evy-Khintchine triplet on $\mathbb R^d$ relative to some truncation function 
$h:\mathbb R^d\rightarrow\mathbb R^d$. We call the mapping $\psi:\mathbb R^d\rightarrow\mathbb C,$
\begin{equation}\label{e:levyexp}
 \psi(u):=iub-\frac{1}{2}u^\top cu+\int(e^{iux}-1-iuh(x))K(dx)
\end{equation}
{\em L\'evy exponent} corresponding to $(b,c,K)$. By \cite[II.2.44]{js.87}, 
the L\'evy exponent determines the triplet $(b,c,K)$ uniquely. If $X$ is a L\'evy process 
with L\'evy-Khintchine triplet $(b,c,K)$, we call $\psi$ the {\em characteristic} or {\em L\'evy exponent} of $X$.
If (\ref{e:levyexp}) exists for all $u\in U\supset \rr^d$,
we call $\psi:U\to\mathbb C$ defined by (\ref{e:levyexp}) 
{\em extended L\'evy exponent} of $X$
on $U$.
\end{definition}

In the same vein, local characteristics naturally lead to local exponents.

\begin{definition}\label{Definition: lokaler Exponent 1. Version}
 If $X$ is an $\mathbb R^d$-valued semimartingale with local characteristics $(b,c,K)$, we write 
\begin{eqnarray}\label{Gleichung: L\'evy exponent}
\psi^X_t(u):=iub_t-\frac{1}{2}u^\top c_tu+\int(e^{iux}-1-iuh(x))K_t(dx),\quad u\in\mathbb R^d
\end{eqnarray}
for the L\'evy exponent corresponding to $(b_t,c_t,K_t)$. We call the family of predictable processes 
$\psi^X(u):=(\psi^X_t(u))_{t\in\mathbb R_+},u\in\mathbb R^d$ {\em local exponent} of $X$. 
(\ref{Gleichung: L\'evy exponent}) implies that $u\mapsto\psi^X_t(u)$ is the characteristic exponent of a L\'evy process.
\end{definition}

The name {\em exponent} is of course motivated by the following fact.

\begin{remark}\label{Bemerkung: PII und charakteristische Funktion}
 If $X$ is a semimartingale with deterministic local characteristics $(b,c,K)$, it is a PII and we have
 $$E(e^{iu(X_T-X_t)}\vert\F_t)=E(e^{iu(X_T-X_t)})=\exp\left(\int_t^T\psi_s^X(u)ds\right)$$
for any $T\in\mathbb R_+$, $t\in[0,T]$, $u\in\mathbb R^d$, cf.\ \cite[II.4.15]{js.87}.
\end{remark}

We now generalize the notion of local exponents to complex-valued semimartingales and more general arguments.

\begin{definition}\label{Definition: lokaler Exponent}
 Let $X$ be a $\mathbb C^d$-valued semimartingale and $\beta$ a $\mathbb C^d$-valued $X$-integrable process. 
We call a predictable $\mathbb C$-valued process $\psi^X(\beta)=(\psi_t^X(\beta))_{t\in\mathbb R_+}$
{\em local exponent} of $X$ at $\beta$ if $\psi^X(\beta)\in L(I)$ and 
$(\exp(i\beta\mal X_t-\int_0^t\psi_s^X(\beta)ds))_{t\in\mathbb R_+}$ is a complex-valued local martingale. 
We denote by $\scr U^X$ the set of processes $\beta$ such that the local exponent $\psi^X(\beta)$ exists.
\end{definition}

From the following lemma it follows that $\psi^X(\beta)$ is unique up to a $dP\otimes dt$-null set.

\begin{lemma}\label{Lemma: Eindeutigkeit des Exponenten}
 Let $X$ be a complex-valued semimartingale and $A,B$ complex-valued predictable processes of finite variation 
with $A_0=0=B_0$ and such that $\exp(X-A)$ and $\exp(X-B)$ are local martingales. Then $A=B$ up to indistinguishability.
\end{lemma}
\begin{proof}
 Set $M:=e^{X-A}$, $N:=e^{X-B}$, $V:=e^{A-B}$. Integration by parts yields that
$$M_-\mal V=MV-V\mal M-M_0V_0=N-V\mal M-M_0$$
is a local martingale. Therefore $V=1+\frac{1}{M_-}\mal(M_-\mal V)$ is a predictable local martingale 
with $V_0=1$ and hence $V=1$, cf.\ \cite[I.3.16]{js.87}.

\end{proof}

The following result shows that Definition \ref{Definition: lokaler Exponent} truly generalizes 
Definition \ref{Definition: lokaler Exponent 1. Version}.

\begin{proposition}\label{Proposition: lokale Exponentenformel}
 Let $X$ be an $\mathbb R^d$-valued semimartingale with local characteristics $(b,c,K)$. 
Suppose that $\beta$ is a $\mathbb C^d$-valued predictable and $X$-integrable process. 
If $\beta$ is $\mathbb R^d$-valued for any $t\in\mathbb R_+$, then $\beta\in\scr U^X$. 
Moreover there is equivalence between
\begin{enumerate}
 \item $\beta\in\scr U^X$,
\item $\int_0^t\int 1_{\{-\Im(\beta_sx)>1\}}e^{-\Im(\beta_sx)}K_s(dx)ds<\infty$ almost surely for any $t\in\rp$.
\end{enumerate}
In this case we have 
\begin{eqnarray}\label{Gleichung: lokale exponenten Formel}
 \psi_t^X(\beta)=i\beta_tb_t-\frac{1}{2}\beta_t^\top c_t\beta_t+\int (e^{i\beta_tx}-1-i\beta_th(x)) K_t(dx)
\end{eqnarray}
outside some $dP\otimes dt$-null set.
\end{proposition}
\begin{proof} If $\beta$ is $\mathbb R^d$-valued, then Statement (2) is obviously true. Thus we only need to 
prove the equivalence and (\ref{Gleichung: lokale exponenten Formel}).
For real-valued $i\beta$ the equivalence follows from \cite[Lemma 2.13]{kallsen.shiryaev.00b}.
The complex-valued case is derived similarly.	
For real-valued $i\beta$ (\ref{Gleichung: lokale exponenten Formel})
is shown in \cite[Theorems 2.18(1,6) and 2.19]{kallsen.shiryaev.00b}.
The general case follows along the same lines.	

\end{proof}

(\ref{Gleichung: lokale exponenten Formel}) implies that the local exponent of $X$ at any $\beta\in\scr U^X$ 
is determined by the triplet $(b,c,K)$ and hence by the local exponent of $X$ in the sense of 
Definition \ref{Definition: lokaler Exponent 1. Version}.

\begin{corollary}\label{Korollar: vertraeglichkeit negativer anteile}
 Let $(X,M)$ be a $1+d$-dimensional semimartingale with local exponent $\psi^{(X,M)}$ such that 
$M$ is a L\'evy process whose components are subordinators. Then $\beta\in\scr U^{(X,M)}$
for any $\mathbb R\times (\mathbb R+ i\mathbb R_+)^d$-valued $(X,M)$-integrable
process $\beta$.
\end{corollary}
\begin{proof}
 This follows immediately from Proposition \ref{Proposition: lokale Exponentenformel}.

\end{proof}

\begin{definition}\label{Definition: lokale unabh}
 Let $X^{(1)},\dots,X^{(n)}$ be semimartingales which allow for local characteristics. 
We call them $X^{(1)},\dots,X^{(n)}$ {\em locally independent} if 
$$\scr U^{(X^{(1)},\dots,X^{(n)})}\cap(L(X^{(1)})\times\dots\times L(X^{(n)}))
=\scr U^{X^{(1)}}\times\dots\times\scr U^{X^{(n)}}$$ 
for
\begin{eqnarray*}
 L(X^{(1)})\times\dots\times L(X^{(n)})
&:=&\{\beta=(\beta^{(1)},\dots,\beta^{(n)})\mbox{ complex-valued}:\\
&& \beta^{(i)} \mbox{ $X^{(i)}$-integrable for }i=1,\dots,n\} 
\end{eqnarray*}
and 
$$\psi^{(X^{(1)},\dots,X^{(n)})}(\beta)=\sum_{j=1}^n\psi^{X^{(j)}}(\beta^{(j)})$$
outside some $dP\otimes dt$-null set
for any $\beta=(\beta^{(1)},\dots,\beta^{(n)})\in\scr U^{(X^{(1)},\dots,X^{(n)})}$.
\end{definition}

The following lemma provides alternative characterisations of local independence.
For ease of notation we consider two semimartingales but the extension 
to arbitrary finite numbers is straightforward.

\begin{lemma}\label{Lemma: Bed fuer lok unabh}
 Let $(X,Y)$ be an $\mathbb R^{m+n}$-valued semimartingale with local characteristics $(b,c,K)$ 
and denote by $(b^X,c^X,K^X)$ resp.\ $(b^Y,c^Y,K^Y)$  local characteristics of $X$ resp.\ $Y$. 
We have equivalence between
\begin{enumerate}
 \item $X$ and $Y$ are locally independent,
\item 
\begin{equation}\label{e:both}
 \psi^{(X,Y)}(u,v)=\psi^X(u)+\psi^Y(v),\quad (u,v)\in\mathbb R^{m+n}
\end{equation}
outside some $dP\otimes dt$-null set,
\item 
$$c=\left(\begin{array}{cc}c^X &0\\0& c^Y\end{array}\right)$$
and 
$$K(A)=K^X(\{x:(x,0)\in A\})+K^Y(\{y:(0,y)\in A\}),\quad A\in\B^{m+n}$$
outside some $dP\otimes dt$-null set.
\end{enumerate}
\end{lemma}
\begin{proof}
{\em (1)$\Rightarrow$(2)}: This is obvious by Proposition \ref{Proposition: lokale Exponentenformel}.

{\em (2)$\Rightarrow$(3)}: Both sides of (\ref{e:both}) are L\'evy exponents for fixed 
$(\omega,t)\in\Omega\times\mathbb R_+$.
Indeed, the triplet corresponding to 
$(u,v)\mapsto(\psi^X_t(u)+\psi^Y_t(v))(\omega)$
is
$(b_t,\widetilde c_t,\widetilde K_t)(\omega)$ with
$$\widetilde c_t=\left(\begin{array}{cc}c_t^X &0\\0& c_t^Y\end{array}\right)$$
and
$$\widetilde K_t(A)=K_t^X(\{x:(x,0)\in A\})+K_t^Y(\{y:(0,y)\in A\}),\quad A\in\B^{m+n}.$$
Since the L\'evy exponent determines the triplet uniquely (cf.\ \cite[II.2.44]{js.87}), 
the assertion follows.

{\em (3)$\Rightarrow$(1)}:
If $\beta^X$ is $X$-integrable and $\beta^Y$ is $Y$-integrable,
 then $\beta=(\beta^X,\beta^Y)$ is $(X,Y)$-integrable and 
$\beta\mal(X,Y)=\beta^X\mal X+\beta^Y\mal Y$. The characterisation in 
Proposition \ref{Proposition: lokale Exponentenformel}
yields $\beta\in\scr U^{(X,Y)}$ for such $\beta=(\beta^X,\beta^Y)$
if and only if $\beta^X\in\scr U^X$, $\beta^Y\in\scr U^Y$.
In addition, 
$\psi^{(X,Y)}(\beta)=\psi^X(\beta^X)+\psi^Y(\beta^Y)$
follows from (\ref{Gleichung: lokale exponenten Formel}) 

\end{proof}

\begin{corollary}
 Let $X^{(1)},\dots,X^{(n)}$ be locally independent semimartingales and $Q\overset{\mathrm{loc}}{\ll} P$ 
another probability measure. Then $X^{(1)},\dots,X^{(n)}$ are locally independent semimartingales relative to $Q$.
\end{corollary}
\begin{proof}
 This follows from Lemma \ref{Lemma: Bed fuer lok unabh} and \cite[III.3.24]{js.87}.

\end{proof}

\begin{corollary}\label{co:PIISPII}
 If $(X^{(1)},\dots,X^{(n)})$ is a L\'evy process or, more generally, a PII allowing for local characteristics, 
then $X^{(1)},\dots,X^{(n)}$ are independent if and only if they are locally independent.
\end{corollary}
\begin{proof}
By Remark \ref{Bemerkung: PII und charakteristische Funktion} the characteristic function $\varphi_{X_t}$ of 
$X_t:=(X_t^{(1)},\dots,X_t^{(n)})$ is given by
$$\varphi_{X_t}(u^1,\dots,u^n)=\exp\left(\int_0^t\psi_s^{(X^{(1)},\dots,X^{(n)})}(u^1,\dots,u^n)ds\right).$$
Thus independence of $X_t^{(1)},\dots,X^{(n)}_t$ is equivalent to
$$\psi_t^{(X^{(1)},\dots,X^{(n)})}(u^1,\dots,u^n)=\sum_{k=1}^n\psi_t^{X^k}(u^k),$$
for Lebesgue-almost any $t\in\mathbb R_+$ and any $(u^1,\dots,u^n)\in\mathbb R^n$. By 
Lemma \ref{Lemma: Bed fuer lok unabh} this in turn is equivalent to local independence of $X_t^{(1)},\dots,X^{(n)}_t$.

\end{proof}

\begin{lemma}\label{Lemma: Negativitaet von LKM}
 If $\varphi\in\Pi$, then $\Re(\varphi(u))\leq0$ for any $u\in\mathbb R$, where $\Pi$ is defined in 
Section \ref{Abschnitt: Optionsfl\"achenmodelle}.
\end{lemma}
\begin{proof}
For $\varphi\in\Pi$ we have 
\begin{eqnarray}
\varphi(u)=-\frac{u^2+iu}{2}c+\int(e^{iux}-1-iu(e^x-1))K(dx)\label{Gleichung: Pi definitionsgleichung}
\end{eqnarray} 
with some L\'evy measure $K$ and some $c\in\mathbb R_+$. The real part of the first term is 
obviously negative and the real part of the integrand is negative as well.

\end{proof}
\begin{remark}
 If we extend the domain of $\varphi$ to $\mathbb R+i[-1,0]$ by keeping the representation 
(\ref{Gleichung: Pi definitionsgleichung}), then the conclusion of Lemma \ref{Lemma: Negativitaet von LKM} 
is still correct. However, this fact is not used in this paper.
\end{remark}

The following four lemmas follow immediately from the definition of local exponents.
\begin{lemma}\label{Lemma: exp(X) Martingal gdw psi(-i)=0}
 Let $X$ be a $\mathbb C$-valued semimartingale that allows for local characteristics. Then there is equivalence between
\begin{enumerate}
 \item $\exp(X)$ is a local martingale,
 \item $-i\in\scr U^X$ and $\psi^X(-i)=0$ outside some $dP\otimes dt$-null set.
\end{enumerate}
\end{lemma}

\begin{remark}\label{r:pi}
If $X$ in Lemma \ref{Lemma: exp(X) Martingal gdw psi(-i)=0} is real-valued and if
$e^X$ is a local martingale, 
Proposition \ref{Proposition: lokale Exponentenformel} 
yields that the mapping $\rr\to \mathbb C$, $u\mapsto\psi^X_t(u)$ is in $\Pi$ outside some $dP\otimes dt$-null set.
\end{remark}

\begin{lemma}\label{Lemma: Levy exponeten Formel}
 Let $X$ be a $\mathbb C^d$-valued semimartingale, $\beta$ a $\mathbb C^d$-valued and $X$-integrable process,
and $u\in\mathbb C$. Then $u\beta\in\scr U^X$ if and only if $u\in\scr U^{\beta\stackrel{\text{\Large .}}{}\! X}$. 
In that case we have
$$\psi^X(u\beta)=\psi^{\beta\stackrel{\text{\Large .}}{}\! X}(u).$$
\end{lemma}

\begin{lemma}\label{Lemma: Levy exponenten Formel}
 Let $X,Y$ be $\mathbb C^d$-valued semimartingales and $u\in\mathbb C^d$. Then $u\in\scr U^{X+Y}$ 
if and only if $(u,u)\in\scr U^{(X,Y)}$. In this case we have
$$\psi^{X+Y}(u)=\psi^{(X,Y)}(u,u)$$
outside some $dP\otimes dt$-null set.
\end{lemma}

\begin{lemma}\label{Lemma: Driftregel fuer Levy Exponenten}
 Let $X,Z$ be  $\mathbb C^d$-valued semimartingales and $\beta,\gamma$ predictable 
$\mathbb C^d$-valued processes such that 
\begin{enumerate}
\item $\gamma$ has $I$-integrable components,
\item $\beta\gamma$ is $I$-integrable,
\item $Z_t=Z_0+\int_0^t\gamma_s ds+X_t$.
\end{enumerate}
Then $\beta\in\scr U^Z$ if and only if $\beta\in\scr U^X$. In this case $\psi^Z(\beta)=\psi^X(\beta)+i\beta\gamma$ 
outside some $dP\otimes dt$-null set.
\end{lemma}

\begin{proposition}\label{P:erhalt der lokalen char}
 Let $\fraum$ be a filtered proability space and $(\scr G_t)_{t\in\mathbb R_+}$  a sub-filtration of 
$(\scr F_t)_{t\in\mathbb R_+}$. Moreover, let $X$ be an $\mathbb R^d$-valued semimartingale 
with integrable characteristics $(B,C,\nu)$  such that $X$ is $(\scr G_t)_{t\in\mathbb R_+}$-adapted and 
$B,C,\nu(A)$ are $(\scr G_t)_{t\in\mathbb R_+}$-predictable for any $A\in\B(\mathbb R^d)$. 
Then $X$ is a semimartingale with  integral characteristics $(B,C,\nu)$ relative to 
filtration $(\scr G_t)_{t\in\mathbb R_+}$as well.
\end{proposition}
\begin{proof}
 Let $u\in\mathbb R^d$ and define the process
$$A(u)_t:=iuB_t-\frac{1}{2}u^\top C_tu+\int(e^{iux}-1-iuh(x))\nu_t(dx).$$
\cite[Theorem II.2.42]{js.87} yields that $e^{iuX}-e^{iuX_-}\mal A(u)$ is an 
$(\scr F_t)_{t\in\mathbb R_+}$-local martingale. The process 
$Y:=e^{iuX_-}\mal A(u)$ is a $(\scr G_t)_{t\in\mathbb R_+}$-predictable process by assumption with c\`adl\`ag paths. 
This implies that it is  $(\scr G_t)_{t\in\mathbb R_+}$-locally bounded 
because it is c\`adl\`ag and predictable,
cf.\ \cite[Lemma A.1]{kallsen.98a}. If $(\tau_n)_{n\in\mathbb N}$ denotes a corresponding sequence of stopping times,
then
$(e^{iuX}-Y)^{\tau_n}$ is bounded and hence it is a $(\scr F_t)_{t\in\mathbb R_+}$-martingale. 
\cite[Corollaire 9.16]{jacod.79} yields that it is a $(\scr G_t)_{t\in\mathbb R_+}$-martingale. 
Therefore $e^{iuX}-Y$ is a $(\scr G_t)_{t\in\mathbb R_+}$-local martingale. 
Consequently, \cite[Theorem II.2.42]{js.87} yields that $X$ is a $(\scr G_t)_{t\in\mathbb R_+}$-semimartingale 
with integral characteristics $(B,C,\nu)$.

\end{proof}

\subsection{Semimartingale decomposition relative to a semimartingale}
\label{Abschnitt: Semimartingalprojektion}
Let $(X,Y)$ be an $\mathbb R^{1+d}$-valued semimartingale with local characteristics $(b,c,K)$, 
written here in the form
\begin{equation}\label{e:form}
 b=\left(\begin{array}{c}b^X\\b^Y\end{array}\right),
\quad c:=\left(\begin{array}{cc}c^X & c^{X,Y}\\c^{Y,X}& c^Y\end{array}\right).
\end{equation}
Suppose that $\int_0^t\int_{(1,\infty)}e^xK_s(dx)ds<\infty$ for any $t\in\mathbb R_+$ or, equivalently, 
$e^X$ is a special semimartingale. We set
\begin{eqnarray*}
 \Xp_t&:=&\log\E\!\left((c^{X,Y}(c^Y)^{-1})\mal Y^c_t+f*(\mu^{(X,Y)}-\nu^{(X,Y)})_t\right)
\end{eqnarray*}
for any $t\in\mathbb R_+$, where $c^-$ denotes the pseudoinverse of a matrix $c$ in the sense of 
\cite{albert.72}, $Y^c$ is the continuous local martingale part of $Y$, $\mu^{(X,Y)}$ resp.\ $\nu^{(X,Y)}$ 
are the random measure of jumps of $(X,Y)$ and its compensator, and 
$f:\mathbb R^{1+d}\rightarrow\mathbb R,(x,y)\mapsto (e^x-1)1_{\{y\neq0\}}$. We call $\Xp$ and 
$X^\perp:=X-\Xp$ the {\em dependent} resp.\ {\em independent part} of $X$ relative to $Y$.

\begin{lemma}\label{Lemma: Projektion}
$X\mapsto\Xp$ is a projection in the sense that $(\Xp)^{\parallel}=\Xp$.
Moreover, we have $(X+Z)^\parallel=\Xp$ if $Z$ is a semimartingale such that $Z,Y$ are locally
independent.
\end{lemma}
\begin{proof}
Observe that $(\Xp)^c=(c^{X,Y}(c^Y)^{-1})\mal Y^c$ by \cite[Lemma 2.6(2)]{kallsen.shiryaev.00b}.
Defining $c^{\Xp, Y}$ similarly as $c^{X,Y}$ in (\ref{e:form}), we have
$c^{\Xp, Y}=c^{X,Y}(c^Y)^{-1}c^Y=c^{X,Y}$.
Moreover, $f(\Delta \Xp_t,\Delta Y_t)=f(\Delta X_t,\Delta Y_t)$ for any $t\ge0$,
which implies 
$f*(\mu^{(\Xp,Y)}-\nu^{(\Xp,Y)})=f*(\mu^{(X,Y)}-\nu^{(X,Y)})$ by definition of stochastic
integration relative to compensated random measures.
Together, the first assertion follows.

Using the notation of (\ref{e:form}), note that 
$c^{(X+Z),Y}=c^{X,Y}+c^{Z,Y}=c^{X,Y}$ by Lemma \ref{Lemma: Bed fuer lok unabh}.
Lemma \ref{Lemma: Bed fuer lok unabh} also implies that $Z$ and $Y$ do not jump together
(outside some evanescent set) and hence
$f(\Delta (X+Z)_t,\Delta Y_t)=f(\Delta X_t,\Delta Y_t)$.
This implies 
$f*(\mu^{(\Xp,Y)}-\nu^{(\Xp,Y)})=f*(\mu^{(X,Y)}-\nu^{(X,Y)})$ 
and hence $(X+Z)^\parallel=\Xp$.

\end{proof}

\begin{lemma}\label{Lemma: Semimartingalprojektion}
 $e^{\Xp}$ is a local martingale. Moreover, $X^\perp$ and $(\Xp,Y)$ are locally independent semimartingales. 
Finally, $e^{X^\perp}$ is a local martingale if and only if $e^X$ is a local martingale.
\end{lemma}
\begin{proof}
 The first statement is obvious. The last statement follows from the first two and
from Lemma \ref{Lemma: exp(X) Martingal gdw psi(-i)=0}.
It remains to prove local independence of $X^\perp$ and $(\Xp,Y)$.
Denote the local characteristics of $(X^\perp,\Xp,Y)$ by 
$(b^{(X^\perp,\Xp,Y)},c^{(X^\perp,\Xp,Y)}, K^{(X^\perp,\Xp,Y)})$
and accordingly for $(X,\Xp,Y)$, $X^\perp$ etc. 
Set $\bar c:=c^{X,Y}(c^Y)^{-1}c^{Y,X}$. 
Since $(\Xp)^c=(c^{X,Y}(c^Y)^{-1})\mal Y^c$, we have
$$c^{(X,\Xp,Y)}=\left(\begin{array}{ccc}c^X &\bar c & c^{X,Y}\\
\bar c&\bar c&c^{X,Y}\\c^{Y,X}&c^{Y,X}&c^Y \end{array}\right)$$
and hence
$$c^{(X^\perp,\Xp,Y)}=\left(\begin{array}{ccc}c^X-\bar c &0 &0\\
0&\bar c&c^{X,Y}\\0&c^{Y,X}&c^Y \end{array}\right)$$
e.g.\ by Proposition \ref{Prop: Ito-formel diff}.
Moreover, 
\begin{eqnarray*}
   \Delta (X^\perp,\Xp,Y)_t
&=&1_{\{\Delta Y_t=0\}}(\Delta X_t,0,0)
+1_{\{\Delta Y_t\not=0\}}(0,\Delta X_t,\Delta Y_t)\\
&=&\left\{\begin{array}{ll}
 (\Delta X^\perp_t,0,0)& \mbox{ if }\Delta X^\perp_t\not=0,\\
(0,\Delta(\Xp,Y)_t)& \mbox{ if }\Delta(\Xp,Y)_t\not=0,\\
0 & \mbox{otherwise}
\end{array}\right.
  \end{eqnarray*}
yields
$$K^{(X^\perp,\Xp,Y)}(A) = K^{X^\perp}(\{x:(x,0,0)\in A\})+K^{(\Xp,Y)}(\{(x,z):(0,x,z)\in A\})$$
for $A\in\B^{2+d}$.
Lemma \ref{Lemma: Bed fuer lok unabh} completes the proof.

\end{proof}

\section{Technical Proofs}
\subsection{Option pricing by Fourier transform}\label{Abschnitt: Fourier transformation und Optionspreise}
By 
\begin{eqnarray}
\F f(u):=\lim_{C\rightarrow\infty}\int_{-C}^\infty f(x)e^{iux}dx
\label{Gleichung: Fourier Transformation} 
\end{eqnarray}
we denote the (left-)improper Fourier transform of a measurable function 
$f:\mathbb R\rightarrow\mathbb C$ for any $u\in\mathbb R$ such that the expression exists. 
If $f$ is Lebesgue integrable, then the improper Fourier transform and the ordinary Fourier transform 
(i.e. $u\mapsto \int f(x)e^{iux}dx$) coincide. In our application in 
Section \ref{Abschnitt: Notationen und Modelvorbereitung} the improper Fourier transform exists 
for any $u\in\mathbb R\setminus\{0\}$. Moreover, we denote by
\begin{eqnarray}
\F^{-1} g(x):=\frac{1}{2\pi}\left(\lim_{\epsilon\downarrow0}\int_\epsilon^\infty e^{-iux}g(u)du
+\lim_{\epsilon\downarrow0}\int_{-\infty}^{-\epsilon} e^{-iux}g(u)du\right)
\label{Gleichung: Inverse Fourier Transformation}  
\end{eqnarray}
an improper inverse Fourier transform, which is suitable to our application in 
Section \ref{Abschnitt: Notationen und Modelvorbereitung}.

\begin{lemma}\label{Lemma: Bedingte Fourier transform}
 Let $(\Omega,\F,P)$ be a probability space and $\G\subset\F$ a sub-$\sigma$-field. 
Furthermore suppose that $f:\Omega\times\mathbb R\rightarrow\mathbb R$ is $\F\otimes\B$-measurable, 
$m:\Omega\rightarrow\mathbb R$ is $\F$-measurable and  $$H(x):=1_{[0,m]}(x)f(x)-1_{[m,0)}(x)f(x)$$ is 
nonnegative with $E(\int_0^mf(x)dx)<\infty$. (Note that $\int_0^mf(x)dx$ is always nonnegative.) Then we have 
\begin{eqnarray}
 \F\{x\mapsto E(H(x)\vert\G)\}(u)=E\left(\int_0^m f(x)e^{iux}dx\bigg\vert\G\right),\quad u\in\rr
\end{eqnarray}
where the improper Fourier transform coincides with the ordinary Fourier transform.
\end{lemma}
\begin{proof}
Let $u\in\mathbb R$. From
\begin{eqnarray*}
\int_0^\infty H(x)e^{iux}dx & = & 1_{\{m\geq0\}}\int_0^mf(x)e^{iux}dx,\\
\int_{-\infty}^0 H(x)e^{iux}dx & = & 1_{\{m<0\}}\int_0^mf(x)e^{iux}dx
\end{eqnarray*}
it follows that $\int_{-\infty}^\infty H(x)e^{iux}dx=\int_0^mf(x)e^{iux}dx$. 
This implies 
\begin{eqnarray*}
 E\left(E\bigg(\int_{-\infty}^\infty H(x)dx\bigg\vert\G\bigg)\right)
=E\bigg(\int_{-\infty}^\infty H(x)dx\bigg)
=E\bigg(\int_0^mf(x)dx\bigg)<\infty
\end{eqnarray*}
and hence
$$\int_{-\infty}^\infty E(H(x)\vert\G)dx
=E\bigg(\int_{-\infty}^\infty H(x)dx\bigg\vert\G\bigg)<\infty.$$
Now we can apply Fubini's theorem and we get
$$\F\{x\mapsto E(H(x)\vert\G)\}(u)=E\bigg(\int_{-\infty}^\infty H(x)e^{iux}dx\bigg\vert\G\bigg)
=E\bigg(\int_0^m f(x)e^{iux}dx\bigg\vert\G\bigg).$$

\end{proof}

The next proposition is a modification of \cite[Proposition 1]{belomestny.reiss.05}, cf.\ also \cite{carr.madan.99}.
\begin{lemma}\label{Lemma: Fourier transformierte der Call-Preise}
 Let $(\Omega,\F,P)$ be a probability space and $\G\subset\F$ a sub-$\sigma$-field. Let $Y$ be a random variable 
such that $E(e^Y)<\infty$ and consider
$$\scr O(x):=\begin{cases}E((e^{Y-x}-1)^+\vert\G) & \text{ if }x\geq0,\\ 
E((1-e^{Y-x})^+\vert\G) & \text{ if }x<0.\end{cases}$$
Then we have
\begin{eqnarray*}
 \F\{x\mapsto\scr O(x)\}(u)&=&\frac{1}{iu}-\frac{E(e^Y\vert\G)}{iu-1}-\frac{E(e^{iuY}\vert\G)}{u^2+iu}
\end{eqnarray*}
and
\begin{eqnarray*}
\lefteqn{\F\{x\mapsto1_{\{x\geq-C\}}\scr O(x)\}(u)}\\
&=&\frac{1}{iu}-\frac{E\left(e^y\vert\G\right)}{iu-1}-\frac{1-E\left(e^{iu\left(Y\vee -C\right)}
\left(1+iu\left(e^{0\wedge(Y+C)}-1\right)\right)\Big\vert\G\right)}{u^2+iu}
\end{eqnarray*}
for any $C\in\mathbb R_+$, $u\in\mathbb R\setminus\{0\}$. If $E(e^Y\vert\G)=1$, then in particular
\begin{eqnarray*}
 \F\{x\mapsto\scr O(x)\}(u) & = & \frac{1-E(e^{iuY}\vert\G)}{u^2+iu}\\
\end{eqnarray*}
for any $u\in\mathbb R\setminus\{0\}$.
\end{lemma}
\begin{proof}
Let $C\in\mathbb R_+$, $u\in\mathbb R\setminus\{0\}$. We define $m:=(Y\vee -C)$, 
$f(x):=e^{Y-x}-1$, and $H(x):=1_{[0,m]}(x)f(x)-1_{[m,0)}(x)f(x)$. 
Then we have $1_{\{x\geq -C\}}\scr O(x)=E(H(x)\vert\G)$, $H\geq 0$, and
$$E\left(\int_0^mf(x)dx\right)=E(e^Y-m-e^{Y-m})<\infty.$$
Hence Lemma \ref{Lemma: Bedingte Fourier transform} yields
\begin{eqnarray*}
\int_{-C}^\infty\scr O(x)e^{iux}dx &=& \scr F\{x\mapsto E(H(x)\vert\G)\}(u)\\
& = & E\left(\int_0^m(e^{Y-x}-1)e^{iux}dx\Bigg\vert\G\right)\\
& = & E\left(\bigg[\frac{e^{Y+(iu-1)x}}{iu-1}-\frac{e^{iux}}{iu}\bigg]_{x=0}^m\Bigg\vert\G\right)\\
& = & \frac{1}{iu}-\frac{E(e^Y\vert\G)}{iu-1}-\frac{E\left(e^{ium}(1+iu(e^{Y-m}-1))\vert\G\right)}{u^2+iu}.
\end{eqnarray*}
Since $\vert e^{ium}(1+iu(e^{Y-m}-1))\vert\leq1+\vert u\vert$, we can apply Lebesgue's theorem and get
$$E\left(e^{ium}(1+iu(e^{Y-m}-1))\vert\G\right) \overset{C\rightarrow\infty}{\longrightarrow} 
E\left(e^{iuY}\vert\G\right).$$

\end{proof}

\begin{corollary}\label{Korollar: Abschaetzung der Fourier transformierten}
Let $Y$ be a random variable with $E(e^Y)<\infty$. Define
$$O(x):=\begin{cases} E\left((e^{Y-x}-1)^+\right) & \text{ if }x\geq0,\\ 
E\left((1-e^{Y-x})^+\right) & \text{ if }x<0.\end{cases}$$
Then we have
$$\left\vert\int_{-C}^\infty O(x)e^{iux}dx\right\vert\leq E(e^Y)+\frac{1+2\vert u\vert}{u^2}$$
for any $C\in\mathbb R_+$, $u\in\mathbb R\setminus\{0\}$.
\end{corollary}
\begin{proof}
 This follows from the second statement of Lemma \ref{Lemma: Fourier transformierte der Call-Preise}.

\end{proof}

\begin{proposition}\label{Proposition: Fourier transformierte der Call-Preise}
 Let $(\Omega,\F,P)$ be a probability space and $\G\subset\F$ a sub-$\sigma$-field. 
Let $Y$ be a random variable with $E(e^Y\vert\G)=1$ and define 
$$\scr O(x):=\begin{cases}E((e^{Y-x}-1)^+\vert\G) & \text{ if }x\geq0,\\ 
E((1-e^{Y-x})^+\vert\G) & \text{ if  }x<0.\end{cases}$$
Then we have
\begin{eqnarray*}
 \scr O(x) & = & \F^{-1}\bigg\{u\mapsto\frac{1-E(e^{iuY}\vert\G)}{u^2+iu}\bigg\}(x),\\
 E(e^{iuY}\vert\G) & = & 1-(u^2+iu)\F\{x\mapsto\scr O(x)\}(u)
\end{eqnarray*}
for any $u,x\in\mathbb R$.
\end{proposition}
\begin{proof}
The second equation is a restatement of Lemma \ref{Lemma: Fourier transformierte der Call-Preise}. 
Let $0<\alpha<1$ and define
$O(x):=e^{\alpha x}\scr O(x)$, $m:=Y$, $f(x):=e^{\alpha x}(e^{Y-x}-1)$, and 
$H(x):=1_{[0,m]}(x)f(x)-1_{[m,0)}(x)f(x)$ for any $x\in\mathbb R$. Then $H$ is nonnegative, 
$O(x)=E(H(x)\vert\G)$, and 
$$E\left(\int_0^mf(x)dx\right)=E\left(\frac{e^{\alpha Y}}{\alpha^2-\alpha}+
\frac{e^Y}{1-\alpha}+\frac{1}{\alpha}\right)<\infty.$$
Lemma \ref{Lemma: Bedingte Fourier transform} yields
\begin{eqnarray*}
 \scr F\{x\mapsto O(x)\}(u) &=& E\left(\int_0^mf(x)e^{iux}dx\bigg|\G\right) \\
&=& \frac{E\left(e^{(\alpha+iu)Y}\Big|\G\right)-1}{(\alpha+iu)^2-(\alpha+iu)}.
\end{eqnarray*}
We have $E(|e^{(\alpha+iu)Y}|)\leq E(1+e^Y)=2$ and thus $u\mapsto \scr F\{x\mapsto O(x)\}(u)$ is integrable. 
The Fourier inversion theorem yields
$$O(x)=\scr F^{-1}\{u\mapsto\scr F\{\tilde x\mapsto O(\tilde x)\}(u)\}(x)$$
because the ordinary inverse Fourier transform coincides with the improper inverse Fourier transform for 
Lebesgue-integrable functions. Define
$$g:\{z\in\mathbb C\setminus\{0\}:-1< \Re(z)\leq 0\}\rightarrow\mathbb C,\quad 
z\mapsto \frac{E(e^{-zY}|\G)-1}{z^2+z}.$$
$g$ is continuous and holomorphic in the interior of its domain. Let $0<\epsilon<\frac{1}{2}=:\alpha$ and define
\begin{eqnarray*}
 \gamma_{(1,\epsilon)}:[-1,1]\rightarrow\mathbb C,&&t\mapsto i\frac{t}{\epsilon}-\frac{1}{2},\\
 \gamma_{(2,\epsilon)}:[0,1]\rightarrow\mathbb C,&&t\mapsto i\frac{1}{\epsilon}-\frac{1-t}{2},\\
 \gamma_{(3,\epsilon)}:[0,1]\rightarrow\mathbb C,&&t\mapsto i(1-t)\left(\frac{1}{\epsilon}-\epsilon\right)+i\epsilon,\\
 \gamma_{(4,\epsilon)}:[0,\pi]\rightarrow\mathbb C,&&t\mapsto i\epsilon e^{it}, \\
 \gamma_{(5,\epsilon)}:[0,1]\rightarrow\mathbb C,&&t\mapsto it\left(\epsilon-\frac{1}{\epsilon}\right)-i\epsilon,\\
 \gamma_{(6,\epsilon)}:[0,1]\rightarrow\mathbb C,&&t\mapsto -i\frac{1}{\epsilon}-\frac{t}{2}
\end{eqnarray*}
as well as $\Gamma_\epsilon:=\sum_{k=1}^6\gamma_{(k,\epsilon)}$.
Cauchy's integral theorem yields
$$\int_{\Gamma_\epsilon}g(z)e^{xz}dz=0.$$
Moreover we have 
$$\frac{1}{2\pi i}\int_{\gamma_{(1,\epsilon)}}g(z)e^{xz}dz
\overset{\epsilon\rightarrow0}{\longrightarrow}O(x)e^{-{\frac{1}{2}}x}=\scr O(x)$$
and 
$$\int_{\gamma_{(k,\epsilon)}}g(z)e^{xz}dz\overset{\epsilon\rightarrow0}{\longrightarrow}0$$
for $k\in\{2,6\}$ and even for $k=4$ because $zg(z)\to0$ for $z\to0$.
Thus we conclude
$$\frac{1}{2\pi}\left(\int_{-{1/\epsilon}}^{-\epsilon}g(-iu)e^{-iux}du+\int_{\epsilon}^{1/\epsilon}
g(-iu)e^{-iux}du\right)$$
\begin{eqnarray*}
&=&\frac{1}{2\pi i}\int_{-\gamma_{(3,\epsilon)}-\gamma_{(5,\epsilon)}}g(z)e^{xz}dz  \\
&=&\frac{1}{2\pi i}\int_{\gamma_{(1,\epsilon)}-\Gamma_\epsilon+\gamma_(2,\epsilon)+\gamma_(4,\epsilon)
+\gamma_(6,\epsilon)}g(z)e^{xz}dz \\
& \overset{\epsilon\rightarrow0}{\longrightarrow} & \scr O(x).
\end{eqnarray*}
Since 
$$\int_{-\infty}^{-1/\epsilon}g(-iu)e^{-iux}du
+\int_{1/\epsilon}^{\infty}g(-iu)e^{-iux}du
\overset{\epsilon\rightarrow0}{\longrightarrow}0,$$
we have 
$$\frac{1}{2\pi}\left(\int_{-\infty}^{-\epsilon}g(-iu)e^{-iux}du
+\int_{\epsilon}^{\infty}g(-iu)e^{-iux}du\right)
\overset{\epsilon\rightarrow0}{\longrightarrow}\scr O(x)$$
and hence 
$$\F^{-1}\bigg\{u\mapsto\frac{1-E(e^{iuY}\vert\G)}{u^2+iu}\bigg\}(x)=\scr O(x).$$

\end{proof}

\begin{proposition}\label{Proposition: Fourier Transformation von Martingalen}
 Let $(N(x))_{x\in\mathbb R}$ be a family of nonnegative local martingales and $(\tau_n)_{n\in\mathbb N}$ 
a common localising sequence for all $N(x)$ such that
\begin{enumerate}
 \item $(\omega,x)\mapsto N_t(x)(\omega)$ is $\F\otimes\B$-measurable for all $t\in\mathbb R_+$,
 \item $x\mapsto N_{t}(x)(\omega)$ is right-continuous and $\int_{-C}^\infty N_t(x)(\omega) dx<\infty$ for all 
$t\in\mathbb R_+$,
\item $\lim_{C\rightarrow\infty}\int_{-C}^\infty e^{iux}N_t(x)(\omega) dx$ exists for all $\omega\in\Omega$, 
$t\in\mathbb R_+$, $u\in\mathbb R\setminus\{0\}$,
\item for any $n\in\mathbb N,t\in\mathbb R, u\in\mathbb R\setminus\{0\}$ there is an integrable random variable $Z$ 
such that $$\left\vert\int_{-C}^\infty e^{iux}N^{\tau_n}_t(x)(\omega) dx\right\vert\leq Z$$ for any $C\in\mathbb R_+$.
\end{enumerate}
Define the (improper, cf.\ (\ref{Gleichung: Fourier Transformation})) Fourier transform of $N$ by
$$X_{t}(u):=\F\{x\mapsto N_{t}(x)\}(u).$$
If $X(u)$ has c\`adl\`ag paths, then it is a local martingale for all $u\in\mathbb R\setminus\{0\}$ 
with common localising sequence $(\tau_n)_{n\in\mathbb N}$. 
\end{proposition}
\begin{proof}For any $C\in\mathbb R_+$, $u\in\mathbb R\setminus\{0\},\omega\in\Omega, t\in\mathbb R_+$ define 
$$X_t^C(u)(\omega):=\int_{-C}^\infty e^{iux}N_t(x)(\omega)dx.$$
Fix $n\in\mathbb N, C\in\mathbb R_+$, $u\in\mathbb R\setminus\{0\}$.  
Then $X^C_{t\wedge\tau_n(\omega)}(u)(\omega)=\int_{-C}^\infty e^{iux}N_{t\wedge\tau_n(\omega)}(x)(\omega)dx$ 
for any $t\in\mathbb R_+$, $\omega\in\Omega$.
Setting
\begin{eqnarray*}c(k,x) & := & \begin{cases}
           1_{\cos(x)>0}\cos(x) & \text{for }k=0,\\
           1_{\sin(x)>0}\sin(x) & \text{for }k=1,\\
           -1_{\cos(x)<0}\cos(x) & \text{for }k=2,\\
           -1_{\sin(x)<0}\sin(x) & \text{for }k=3,\\
          \end{cases}\\
I^C_t(k) & := & \int_{-C}^\infty c(k,ux)N_{t\wedge\tau_n}(x)dx
\end{eqnarray*}
yields $$X^C_{t\wedge\tau_n}(u)=\sum_{k=0}^3i^k I^C_t(k).$$
Since $c(k,\cdot):\mathbb R\rightarrow\mathbb R_+$ and hence $I^C(k)$ are positive, 
we can apply Tonelli's theorem and conclude that $I^C(k)$ is a martingale up to the c\`adl\`ag property for all 
$k\in\{0,1,2,3\}$. Thus $(t,\omega)\mapsto X^C_{t\wedge\tau_n}(u)(\omega)$ is a martingale up to the 
c\`adl\`ag property as well. By the definitions of $X^C$ and $X$ 
we have $X_t(u)(\omega)=\lim_{C\rightarrow\infty}X_t^C(u)(\omega)$ and thus we get
$$X_{t\wedge\tau_n(\omega)}(u)(\omega)=\lim_{C\rightarrow\infty}X_{t\wedge\tau_n(\omega)}^C(u)(\omega).$$
The fourth assumption on $N$ and Lebesgue's theorem yield
\begin{eqnarray*}
 E(X_{t\wedge\tau_n}(u)\vert\scr F_s)&=&E\left(\lim_{C\rightarrow\infty}X_{t\wedge\tau_n}^C(u)\Big\vert\scr F_s\right)\\
&=&\lim_{C\rightarrow\infty}E\left(X_{t\wedge\tau_n}^C(u)\Big\vert\scr F_s\right) \\
&=&X_{s\wedge\tau_n}(u)
\end{eqnarray*}
for $s\leq t$. Thus $(t,\omega)\rightarrow X_t(u)(\omega)$ is a local martingale.

\end{proof}

\subsection{Essential infimum of a subordinator}
\begin{definition}\label{Definition: infimum Prozess}
 Let $\psi$ be a deterministic local exponent (of some semimartingale $X$). We define the {\em infimum process} 
$E$ of $\psi$ by $E_t:=\essinf X_t$ for any $t\in\mathbb R_+$. 
\end{definition}

By \cite[III.2.16]{js.87}, $X$ in the previous definition is a PII whose law is determined by $\psi$.
Since $E_t$ is in turn determined by the law of $X_t$, the infimum process $E$ does not depend
on the particular choice of $X$.

\begin{lemma}\label{Lemma: Additivitaet und essinf}
 Let $X,Y$ be independent random variables. Then
$$\essinf (X+Y)=\essinf \ X+\essinf\ Y.$$
\end{lemma}
\begin{proof}
 Let $x:=\essinf\ X$ and $y:=\essinf\ Y$. We obviously have $\essinf(X+Y)\geq x+y$ because $X+Y\geq x+y$ 
almost surely. Independence yields
\begin{eqnarray*}
 P(X+Y\leq x+y+\epsilon) &\geq&P\left(X\leq x+\frac{\epsilon}{2},Y\leq y+\frac{\epsilon}{2}\right)\\
 &=&P\left(X\leq x+\frac{\epsilon}{2}\right)P\left(Y\leq y+\frac{\epsilon}{2}\right)\\
&>&0
\end{eqnarray*}
for any $\epsilon>0$.

\end{proof}

\begin{proposition}\label{Proposition: essinf und Subordinator}
 Let $X$ be a subordinator (i.e.\ an increasing L\'evy process) and 
$E_t:=\essinf X_t$ for any $t\in\mathbb R_+$. Then $E_t=tE_1\geq0$ for any $t\geq0$ 
and $E$ is the drift part of $X$ relative to the ``truncation'' function $h=0$. Moreover, $X-E$ is a subordinator.
\end{proposition}
\begin{proof} 
Since $X$ is a subordinator, we have $E_t\geq 0$ for any $t\in\mathbb R_+$. Moreover, $\essinf(X_t-X_s)=E_{t-s}$
 because $X_{t-s}$ has the same distribution as $X_t-X_s$. Since $X_s$ and $X_t-X_s$ are independent, 
Lemma \ref{Lemma: Additivitaet und essinf} yields $E_s+E_{t-s}=E_t$. The mapping $t\mapsto E_t$ is 
increasing because $X$ is a subordinator. Together we conclude $E_t=tE_1$. This implies that $X-E$ is a 
positive L\'evy process and hence a subordinator. 
By \cite[Theorem 21.5]{sato.99} the L\'evy-Khintchine triplet $(b,c,K)$ relative to ``truncation'' function $h=0$ 
exists and satisfies $c=0$. Moreover, $K$ and the random measure of jumps $\mu^X$ of $X$ are concentrated 
on $\mathbb R_+$. In view of
\cite[II.2.34]{js.87} we have $X_t=x*\mu^X_t+bt$ and thus we get $E_t=\essinf X_t\geq bt$. 
According to \cite[Theorem 21.5]{sato.99}, $X-E$ is a subordinator only if its drift rate $\widetilde b$ 
relative to $h=0$ is greater or equal $0$. Hence $bt-E_t=\widetilde bt\geq0$, which implies $bt=E_t$.

\end{proof}

\begin{corollary}\label{Korollar: essinf und beschraenkte Integration}
 Let $X$ be a $d$-dimensional semimartingale whose components $X^k$ are subordinators with 
essential infimum $E^k_t=\essinf X_t^k$ for $k=1,\dots,d$. For componentwise nonnegative bounded 
predictable processes $\varphi$ we have 
$$\essinf(\varphi\mal (X-E)_t)=0.$$
Moreover, for bounded predictable $\mathbb C^d$-valued processes $\varphi$ we have 
$$\essinf\left\vert \varphi\mal (X-E)_t\right\vert=0$$ 
for any $t\in\mathbb R_+$.
\end{corollary}
\begin{proof} The second statement is an application of the first statement because 
$$0\leq\left\vert \varphi\mal (X-E)_t\right\vert\leq
(|\varphi^1|,\dots,|\varphi^d|)\mal(X-E)_t.$$
Suppose that $\varphi$ is a componentwise nonnegative bounded predictable process. 
Proposition \ref{Proposition: essinf und Subordinator} yields that $X^k-E^k$ is a 
subordinator with essential infimum $0$. Hence we may assume w.l.o.g.\ that $E^k=0$. 
Since $\varphi$ is bounded, there is a constant $c\in\mathbb R_+$ such that $\varphi^k\leq c$ for $k=1,\dots,d$.
Hence $\varphi\mal X_t\leq c\sum_{k=1}^dX^k_t$ for any $t\in\mathbb R_+$. By 
Proposition \ref{Proposition: essinf und Subordinator} the drift part of $X^k$ is $0$ 
relative to the truncation function $h=0$.  Consequently, the drift part of the subordinator 
$L:=c\sum_{k=1}^dX^k$ is also $0$ relative to the truncation function $h=0$. 
Proposition \ref{Proposition: essinf und Subordinator} yields $\essinf L_t=0$ for any $t\in\mathbb R_+$. 
Thus we conclude
$$0\leq\essinf(\varphi\mal X_t)\leq\essinf L_t=0.$$

\end{proof}

\subsection{Differentiability of L\'evy exponents}\label{su:diff}
\begin{lemma}
 Let $X$ denote an $\rr^d$-valued L\'evy process with finite second moments
in the sense that $E(|X_t|^2)<\infty$ for some (and hence for any) $t>0$.
Then its L\'evy exponent is twice continuously differentiable with bounded 
second-order derivatives.
\end{lemma}
\begin{proof}
If $(b,c,K)$ denotes the L\'evy-Khintchine triplet of $X$,
the  L\'evy exponent $\psi$ of $X$  is of the form  (\ref{e:levyexp}).
The existence of second moments yields
$\int |x|^2K(dx)<\infty$ by \cite[Corollary 25.8, Definition 8.2]{sato.99}.
Dominated convergence implies that we may differentiate
under the integral sign and obtain
\begin{equation}\label{e:fuerbem}
 \partial_i\psi(u)=
ib_i-c^{i\cdot}u+\int\big(ix_ie^{iux}-ih_i(x)\big)K(dx),\quad i=1,\dots,d
\end{equation}
and
$$\partial_i\partial_j\psi(u)=
-c^{ij}-\int x_ix_je^{iux}K(dx),\quad i,j=1,\dots,d,$$
where $\partial_i$ denotes the partial derivative relative to $u_i$ and likewise for $j$.
Since 
$$2\int|x_ix_je^{iux}|K(dx)\leq \int x_i^2K(dx)+\int x_j^2K(dx),$$
the claim follows.
 
\end{proof}

\begin{remark}\label{Remark smoothness}
Consider an  $\rr^2$-valued L\'evy process $X,M$ such that
$M$ is a subordinator. If $X$ and $M$ have finite second moments, 
the statement of the previous lemma holds also for its extended
L\'evy exponent $\psi^{(X,M)}$ on $\rr\times(\rr+i\rp)$. This is shown along the same lines 
as above. Moreover, $\partial_2\psi^{(X,M)}$ is bounded in this case
by  (\ref{e:fuerbem}) and since $\int|x_2|K(dx)<\infty$.
\end{remark}

\subsection{Existence of PII}
\begin{theorem}\label{S:Ex von PII}
 Let $\psi\in\raum{E}$ such that $u\mapsto \int_t^T\psi(r,u)dr$ is a L\'evy exponent for any $t,T\in\rr$, $t\leq T$. 
Then there is a PII $X$ such that $Ee^{iuX_t}=\exp(\int_0^t\psi(r,u)dr)$ for any $t\in\rp$, $u\in\rr$.
\end{theorem}
\begin{proof}
Kolmogorov's extension theorem \cite[Theorem 1.8]{sato.99} yields that there is a stochastic process $Y$ with 
$Y_0=0$ a.s.\ and $Ee^{iu(Y_t-Y_s)}=\exp(\int_s^t\psi(r,u)dr)$ for any $s,t\in\rp$ with $s\leq t$, $u\in\rr$. 
This process satisfies (1) and (2) of \cite[Definition 1.6]{sato.99}. 
We show that it is an additive process in law in the sense of \cite[Definition 1.6]{sato.99}, i.e.\ it is in addition 
stochasically continuous. Let $\epsilon>0$ and $s\in\rp$. Let $\varphi:\rr\rightarrow[0,1]$ be infinitely differentiable, 
its support contained in $[-\epsilon,\epsilon]$, and $\varphi(0)=1$. 
Define 
$$\check\varphi:\mathbb R\rightarrow\mathbb C,u\mapsto \frac{1}{2\pi}\int_{-\infty}^\infty\varphi(x)e^{-iux}dx.$$
Then $\varphi(x) = \int_{-\infty}^\infty\check\varphi(u)e^{iux}du$ for any $x\in\mathbb R$. For $t>s$ we have
\begin{eqnarray*}
 P(\vert Y_t-Y_s\vert>\epsilon) &\leq& 1-E(\varphi(Y_t-Y_s)) \\
 &=& 1-\int_{-\infty}^\infty\check\varphi(u)Ee^{iu(Y_t-Y_s)}du \\
 &=& 1-\int_{-\infty}^\infty\check\varphi(u)\exp\left(\int_s^t\psi(r,u)dr\right)du \\
 &\underset{t\downarrow s}\longrightarrow& 1-\int_{-\infty}^\infty\check\varphi(u)du \\
 &=& 1-\varphi(0) = 0,
\end{eqnarray*}
where $t>s$ and the convergence follows from the dominated convergence theorem. Similar arguments yield
$$\lim_{t\uparrow s}P(\vert Y_t-Y_s\vert>\epsilon) = 0.$$
Thus $Y$ is stochastically continuous. \cite[Theorem 11.5]{sato.99} implies that there is a PII $X$ 
with the desired properties.

\end{proof}

\section{Bochner integrals and stochastic differential equations in Fr\'echet spaces}\label{s:frechetsde}
Option price surfaces are interpreted as elements of the Fr\'echet space 
$\raum{E}$ in Section \ref{Abschnitt: Beispiele und Existenzresultate}.
In order to derive existence and uniqueness results, we need to consider stochastic differential equations
in such spaces, cf.\ Section \ref{su:bochner} below.
These in turn rely on a properly defined Bochner integral., which is discussed in Section \ref{su:frechet}.

\subsection{Bochner integration in Fr\'echet spaces}\label{su:frechet}
Let $F$ be a vector space and $(\Vert\cdot\Vert_n)_{n\in\mathbb N}$ an increasing sequence 
of separable semi-norms on $F$ such that
\begin{enumerate}
 \item $\|x\|_n=0$ $\forall n\in\mathbb N$ holds only if $x=0$,
\item if $(x_k)_{k\in\mathbb N}$ is a $\|\cdot\|_n$-Cauchy sequence for all $n\in\mathbb N$, 
there exists $x\in F$ with $\lim_{k\to\infty}\|x_k-x\|_n=0$ for any $n\in\mathbb N$.
\end{enumerate}
Then
$$d(x,y):=\sum_{n\in\mathbb N}2^{-n}(1\wedge\Vert x-y\Vert_n)$$
defines a complete, translation-invariant, separable metric on the Fr\'echet space $F$.

\begin{remark}\label{R:Metrik abschaetzung}
 Let $x,y\in F$, $n\in\mathbb N$. Then
\begin{eqnarray*}
 1\wedge\Vert x-y\Vert_n &\leq& 2^nd(x,y)\quad\text{and}\\
 d(x,y) &\leq& \Vert x-y\Vert_n+2^{-n}.
\end{eqnarray*}
\end{remark}

\begin{example}\label{ex:frechet}
 We are mainly interested in the case that $F := \raum{E}$ and $\Vert\cdot\Vert_n:= \Vert\cdot\Vert_{n,n}$ 
as defined in Section \ref{Abschnitt: Beispiele und Existenzresultate} or, alternatively, $F=E$ itself in
Lemma \ref{l:banach}.
\end{example}

Fix a $\sigma$-finite measure space $(\Gamma,\scr G,\mu)$.
The goal of this section is to define a  Bochner integral $\int fd\mu$ for measurable functions 
$f:\Gamma\to F$ with values in the Fr\'echet space $F$, cf.\ Definition \ref{d:frechetintegral} below.
If $f$ is {\em simple} and {\em integrable} in the sense that it is 
a linear combination of indicators of sets in $\G$ with finite $\mu$-measure, the 
integral 
$\int fd\mu$
is naturally defined as a sum.

For fixed $n\in\mathbb N$ denote the set of measurable Bochner-integrable functions from 
$(\Gamma,\scr G,\mu)$ to the complete, separable, semi-normed space
$(F,\Vert\cdot\Vert_n)$ by $\scr L^1((\Gamma,\scr G,\mu),(F,\Vert\cdot\Vert_n))$. 
Recall that a measurable function $f:\Gamma\to F$ is called {\em Bochner integrable} (relative to $\|\cdot\|_n$)
if there is a sequence $(f^{(k)})_{k\in\mathbb N}$  of simple integrable functions with 
$\lim_{k\to\infty}\int \|f^{(k)}-f\|_nd\mu=0$.
Equivalently, $f:\Gamma\to F$ is measurable and $\int \|f\|_nd\mu< \infty$,
cf.\ e.g.\ Lemma \ref{L:einfache approximation} below.
In this case there is some $x\in F$ such that
$$\lim_{k\to\infty}\left\|\int f^{(k)}d\mu-x\right\|_n=0$$ 
for any such sequence.
This element $\int fd\mu:=x$ is called {\em ($\Vert\cdot\Vert_n$-)Bochner integral} of $f$.

 Note that we do not identify functions which are $\mu$-a.e. identical. Therefore
$\scr L^1((\Gamma,\scr G,\mu)$, $(F,\Vert\cdot\Vert_n))$ 
together with the semi-norm 
$\|f\|:= \int \Vert f\Vert_n d\mu$
is a complete semi-normed space, but in general not a Banach space.
Moreover,
{\em versions} of the $\Vert\cdot\Vert_n$-Bochner integral may differ by $\Vert\cdot\Vert_n$-distance zero.

The following lemma is needed in Section \ref{su:bochner}.
\begin{lemma}\label{L:Topologien}
 Let $\scr O_d, \scr O_n, n\in\mathbb N$ be the topologies generated by $d$ and 
$\Vert\cdot\Vert_n$, respectively. 
Then $\scr O_n\subset\scr O_d$ for any $n\in\mathbb N$.
For any $U\in\scr O_d$, $n_0\in\mathbb N$ there is a sequence $(V_n)_{n\geq n_0}$ such that $V_n\in\scr O_n$ and
 $U= \cup_{n\geq n_0}V_n$.
 In particular, the Borel $\sigma$-field corresponding to the metric $d$ is 
generated by the Borel $\sigma$-fields corresponding to  the semi-norms $\Vert\cdot\Vert_n$, $n\geq n_0$.
\end{lemma}
\begin{proof}
 Let $n\in\mathbb N$ and $d_n(x,y):=1\wedge \Vert x-y\Vert_n$ for any $x,y\in F$. 
Observe that the semi-metric $d_n$ generates the same topology as the semi-norm $\Vert\cdot\Vert_n$ because 
their balls of radius less or equal $1$ coincide. 
Remark \ref{R:Metrik abschaetzung} yields that $d_n\leq 2^nd$ and hence $\scr O_n\subset\scr O_d$.

The second inequality of Remark \ref{R:Metrik abschaetzung} implies that for any $x\in F$, $\epsilon>0$
 there is $n\geq n_0$ such that
  $B_{n}(x,\epsilon/2) \subset B_d(x,\epsilon)$,
where $B_n(x,\epsilon)$ resp.\ $B_d(x,\epsilon)$ denote the balls centered at $x$ 
with radius $\epsilon$ relative to the semi-norm $\Vert\cdot\Vert_n$ resp.\ the metric $d$.
 Let $U\in\scr O_d$. For any $x\in U$ choose $\epsilon_x>0$ such that $B_d(x,\epsilon_x)\subset U$
and $n_x\in\mathbb N$ with $n_x\geq n_0$ such that
  $B_{n_x}(x,\epsilon_x/2) \subset B_d(x,\epsilon_x)$.
 For any $n\geq n_0$ define
  $$V_n:=\bigcup \{B_{n_x}(x,\epsilon_x/2):x\in U,n_x=n\}\in\scr O_n.$$
 Since $V_n\subset U$, we have $\bigcup_{n\geq n_0} V_n\subset U$. Moreover,
  $$U=\bigcup_{n\geq n_0}\{ x:x\in U,n_x=n\}\subset \bigcup_{n\geq n_0}V_n.$$
 Hence $U=\bigcup_{n\geq n_0}V_n$.

\end{proof}

The $\Vert\cdot\Vert_n$-Bochner integrals are consistent in the following sense:
\begin{lemma}\label{L:Integralkompatibilitaet}
 Let $n\in\mathbb N$
and $f\in \scr L^1((\Gamma,\scr G,\mu),(F,\Vert\cdot\Vert_n))$. Then
$f\in \scr L^1((\Gamma,\scr G,\mu),(F,\Vert\cdot\Vert_k))$ 
for $k=1,\dots,n$. Moreover, if $x$ is a version of the $\Vert\cdot\Vert_n$-Bochner integral 
$\int fd\mu$, then $x$ is a version of the $\Vert\cdot\Vert_k$-Bochner integral $\int fd\mu$ for $k=1,\dots,n$.
\end{lemma}
\begin{proof}
 Let $k\in\{1,\dots,n\}$. Then $\Vert\cdot\Vert_k\leq\Vert\cdot\Vert_n$ and hence $\scr O_k\subset\scr O_n$. 
Thus $f$ is measurable with respect to the Borel $\sigma$-field generated by $\Vert\cdot\Vert_k$. We have
 $$\int \Vert f\Vert_kd\mu \leq \int \Vert f\Vert_nd\mu <\infty$$
 and hence $f\in \scr L^1((\Gamma,\scr G,\mu),(F,\Vert\cdot\Vert_k))$. 
Let $x$ be a version of the $\Vert\cdot\Vert_n$-Bochner integral $\int fd\mu$. 
Then there is a sequence $(f^{(k)})_{k\in\mathbb N}$ of simple integrable functions such that
  \begin{eqnarray*}
   \int \Vert f^{(k)}-f\Vert_n d\mu \rightarrow 0\quad\mbox{and}\quad
   \left\Vert \int f^{(k)}d\mu - x\right\Vert_n \rightarrow 0
  \end{eqnarray*}
 for $k\rightarrow\infty$. 
Since this also holds for $k$ instead of $n$,
we have that  $x$ is a version of the $\Vert\cdot\Vert_k$-Bochner integral $\int fd\mu$ as well.

\end{proof}

We are now ready to define the desired integral for Fr\'echet space-valued functions.
\begin{proposition}\label{P:Fr\'echet integral}
 Let 
  \begin{eqnarray*}
   \lefteqn{f\in \scr L^1((\Gamma,\scr G,\mu),(F,d)) 
:= \bigcap_{n\in\mathbb N} \scr L^1((\Gamma,\scr G,\mu),(F,\Vert\cdot\Vert_n))}\\
&=&\left\{f:\Gamma\to F:
f\mbox{ measurable with } \int \Vert f\Vert_n d\mu<\infty, n\in\mathbb N\right\}.
  \end{eqnarray*}
  Then there is one and only one $x\in F$ such that $x$ is a version of the 
$\Vert\cdot\Vert_n$-Bochner integral $\int fd\mu$ for any $n\in\mathbb N$. 
\end{proposition}
\begin{proof}
 Since $f\in \scr L^1((\Gamma,\scr G,\mu),(F,\Vert\cdot\Vert_n))$, 
there is a version $x_n$ of the $\Vert\cdot\Vert_n$-Bochner integral $\int fd\mu$ for any $n\in \mathbb N$. 
Lemma \ref{L:Integralkompatibilitaet} yields that
  $d(x_n,x_m) \leq \Vert x_n-x_m\Vert_k+2^{-k} = 2^{-k}$
 for any $n,m,k\in\mathbb N$ with $m,n\geq k$. Hence $(x_n)_{n\in\mathbb N}$ is a 
Cauchy sequence in $(F,d)$. Since $(F,d)$ is complete, it converges to some $x\in F$.

 Let $n\in\mathbb N$. Lemma \ref{L:Integralkompatibilitaet} yieds that $x_m$ is a version of the 
$\Vert\cdot\Vert_n$-Bochner integral $\int fd\mu$ and hence $\Vert x_m-x_n\Vert_n=0$ 
for any $m\in\mathbb N$ with $m\geq n$. By Remark \ref{R:Metrik abschaetzung} 
$d$-convergence implies $\Vert\cdot\Vert_n$-convergence. In particular,
 $\lim_{m\rightarrow\infty}\Vert x-x_m\Vert_n=0$. Together we have
  $$\Vert x-x_n\Vert_n = \lim_{m\rightarrow \infty} \Vert x_m-x_n\Vert_n =0.$$
  Hence $x$ is a version of the $\Vert\cdot\Vert_n$-Bochner integral $\int fd\mu$.

 Let $y\in F$ be a version of the $\Vert\cdot\Vert_n$-Bochner integral $\int fd\mu$ for any $n\in\mathbb N$. 
Then $\Vert x-y\Vert_n=0$ for any $n\in\mathbb N$. Since the sequence of semi-norms is separating, we have $x=y$.

\end{proof}

\begin{definition}[Bochner integral]\label{d:frechetintegral}
We call  $\scr L^1((\Gamma,\scr G,\mu),(F,d))$ the set of {\em Bochner-integrable} functions
and $\int fd\mu:=x$ from Proposition \ref{P:Fr\'echet integral} the corresponding {\em Bochner integral}.
\end{definition}

The next lemma is due to Pettis and can be found in a slightly different version in e.g.\ in 
\cite[Theorem II.2]{diestel.uhl.77}. It is a characterisation of measurability which also holds in 
separable semi-metric spaces. However, we are mainly interested in the additional bound that can be imposed 
on the approximating sequence. 
\begin{lemma}\label{L:einfache approximation}
 Let $(E,\Vert\cdot\Vert)$ be a separable semi-normed space and $f:\Gamma\rightarrow E$. 
Then we have equivalence between:
\begin{enumerate}
 \item $f$ is measurable,
\item there is a sequence $(f^{(n)})_{n\in\mathbb N}$ of measurable functions such that
$\lim_{n\rightarrow \infty}\Vert f^{(n)}(t)-f(t)\Vert = 0$ for any $t\in\Gamma$,
\item there is a sequence $(f^{(n)})_{n\in\mathbb N}$ of simple functions such that
$\lim_{n\rightarrow \infty}\Vert f^{(n)}(t)-f(t)\Vert = 0$ for any $t\in\Gamma$.
\end{enumerate}
In this case  the sequence of simple functions can be chosen such that
$\Vert f^{(n)}(t)\Vert \leq 2\Vert f(t)\Vert$ for any $t\in\Gamma$, $n\in\mathbb N$.
\end{lemma}
\begin{proof}
 $1\Rightarrow3$: Let $f$ be measurable and $(x_n)_{n\in\mathbb N}$  a dense sequence in $E$
with $x_0=0$, i.e.\ for any $y\in E$, $\epsilon>0$ there is $n\in\mathbb N$ such that $\Vert x_n-y\Vert<\epsilon$. 
Define closed sets
 $$C_{n,k} := \big\{ y\in E: \Vert y-x_k\Vert = \min\{ \Vert y-x_j\Vert:j=1,\dots,n\}\big\}$$
for any $n\in\mathbb N$ and $k=1,\dots,n$. 
Moreover, we define Borel-measurable sets 
$M_{n,k}:= C_{n,k}\setminus(C_{n,1}\cup\dots\cup C_{n,k-1})$ and 
simple functions
$$f^{(n)}:=\sum_{k=1}^nx_k1_{f^{-1}(M_{n,k})}$$
for any $n\in\mathbb N$ and $k=1,\dots,n$. Then
$$\Vert f(t)-f^{(n)}(t)\Vert = \min\{ \Vert f(t)-x_k\Vert:k=1,\dots,n\}\leq\Vert f(t)\Vert$$
and hence
$$\Vert f^{(n)}(t)\Vert\leq \Vert f^{(n)}(t)-f(t)\Vert+\Vert f(t)\Vert\leq 2\Vert f(t)\Vert$$
for any $n\in\mathbb N$, $t\in\Gamma$. Let $t\in\Gamma$ and $\epsilon>0$. 
Then there is $n_0\in\mathbb N$ such that $\Vert x_{n_0}-f(t)\Vert<\epsilon$. Hence
$$\Vert f^{(n)}(t) -f(t)\Vert\leq \Vert x_{n_0}-f(t)\Vert \leq \epsilon,\quad n\ge n_0.$$

$3\Rightarrow2$: This is obvious.

$2\Rightarrow1$:
We show that $f^{-1}(A)\in\scr C$ for any closed set $A\subset E$. 
We have
$A=\bigcap_{n\in\mathbb N} A_n$
for the open sets
$$A_n:=\{x\in E:\exists y\in A:\Vert x-y\Vert<1/n \}.$$
Hence
\begin{eqnarray*}
 f^{-1}(A) \subset  \bigcap_{n\in\mathbb N}\bigcup_{N\in\mathbb N}\bigcap_{k\geq N} (f^{(k)})^{-1}(A_n) 
   \subset \bigcap_{n\in\mathbb N}f^{-1}({A_n}) 
   \subset f^{-1}(\overline A)= f^{-1}(A),
\end{eqnarray*}
which implies
$$ f^{-1}(A) = \bigcap_{n\in\mathbb N}\bigcup_{N\in\mathbb N}\bigcap_{k\geq N} 
(f^{(k)})^{-1}({A_n})\in\scr C.$$

\end{proof}

\subsection{Stochastic differential equations with Fr\'echet space-valued processes}\label{su:bochner}
Let $(F,d)$ denote the Fr\'echet space of the previous section.
We identify right-continuous, increasing functions 
$X:\rp\to\rr$
with their corresponding Lebesgue-Stieltjes measure $\mu$ on $\rp$.
For $t\in\rp$ and $\varphi:\rp\to F$ we write
$$\varphi\in\scr L^1\big(([0,t],X),(F,d)\big)$$
if 
$\varphi1_{[0,t]}\in \scr L^1\big((\rp,\scr B(\rp),\mu)),(F,d)\big)$.
Moreover, we write
$$\int_0^t\varphi_sdX_s:= \int\varphi 1_{[0,t]}d\mu\in F$$
for the integral from Definition \ref{d:frechetintegral}.

For right-continuous increasing processes $X$ and measurable functions
$\varphi:\Omega\times\rp\to\rr$, both
$\varphi\in\scr L^1(([0,t],X),(F,d))$ and
$\int_0^t\varphi_sdX_s$ are to be interpreted in a pathwise sense,
i.e.\ for any fixed $\omega\in\Omega$.

\begin{lemma}\label{l:vorweg}
 Let $X:\Omega\times\mathbb R_+\rightarrow\mathbb R_+$ be an increasing right-continuous process.
 Moreover, let $\varphi:(\Omega\times\mathbb R_+,\scr A) \rightarrow (F,d)$
 be Borel measurable with
 $\varphi\in \scr L^1(([0,T],X),(F,d))$
for any $T\in\mathbb R_+$,
where $\scr A$ denotes the optional $\sigma$-field. Then
 $Y_t:= \int_0^t\varphi_sdX_s$
defines an adapted c\`adl\`ag process $Y$.
\end{lemma}
\begin{proof}
 The dominated convergence theorem applied to the $\Vert\cdot\Vert_n$-Bochner integrals yields the c\`adl\`ag property 
of the process $Y$.

 We now show that $Y$ is adapted. Let $T\in\mathbb R_+$ and $n\in\mathbb N$. 
Lemma \ref{L:Topologien} yields that
  $\Omega\times[0,T]\rightarrow F$, $(\omega,t)\mapsto \varphi_t(\omega)$
is $\scr F_T\otimes \scr B([0,T])$-$\scr B_n$-measurable, where $\scr B_n$ denotes the 
Borel $\sigma$-field generated by $\Vert\cdot\Vert_n$. 
Lemma \ref{L:einfache approximation} yields that there is a sequence of simple 
$\scr F_T\otimes\scr B([0,T])$-$\scr B_n$-measurable functions $(\varphi^{(k)})_{k\in\mathbb N}$ such that
\begin{eqnarray*}
 \Vert \varphi^{(j)}_t(\omega)\Vert_n &\leq& 2\Vert\varphi_t(\omega)\Vert_n, \\
 \Vert \varphi^{(k)}_t(\omega) -\varphi_t(\omega)\Vert_n &\rightarrow& 0 \quad\text{for }k\rightarrow \infty 
\end{eqnarray*}
for any $j\in\mathbb N$, $\omega\in \Omega$, $t\in[0,T]$. 
The random variable $\int_0^t \varphi^{(k)}_s dX_s$ is $\scr F_T$-$\scr B_n$-measurable 
for any $k\in\mathbb N$. 
Proposition \ref{P:Fr\'echet integral} yields that the $F$-valued Bochner integral
$Z:=\int_0^t\varphi_sdX_s$
is a version of the $\Vert\cdot\Vert_n$-Bochner integral. 
The integral inequality for Bochner integrals and the dominated convergence theorem yield
 $$\lim_{k\rightarrow\infty} \left\Vert \bigg(\int_0^t\varphi^{(k)}_sdX_s\bigg)(\omega) - Z(\omega)\right\Vert_n
\leq \lim_{k\rightarrow\infty} \left(\int_0^t \Vert\varphi^{(k)}_s-\varphi_s\Vert_n dX_s\right)(\omega) = 0$$
for any $\omega\in\Omega$. Hence Lemma \ref{L:einfache approximation} 
implies that $Z$ is $\scr F_T$-$\scr B_n$-measurable. 
Since $n$ was chosen arbitraily, 
Lemma \ref{L:Topologien} yields that $Z$ is $\scr F_T$-$\scr B_d$-measurable,
where $\scr B_d$ denotes the Borel $\sigma$-field generated by $d$.

\end{proof}

We now turn to existence and uniqueness of solutions to Banach space-valued SDE's. 
As usual, this follows e.g.\ under Lipschitz conditions. 
\begin{theorem}\label{S:Ex und Eind Satz}
Let $X=(X^1,\dots,X^d)$ be an $\mathbb R^d$-valued right-continuous process whose components
are nonnegative and increasing. Moreover, let $(E,\Vert\cdot\Vert)$ be a separable Banach space, 
$x\in E$ and
$a:\mathbb R_+\times E\rightarrow E^d$
measurable such that
$$\sup_{s\in[0,T]}\sum_{i=1}^d\Vert a_i(s,0)\Vert<\infty$$
and
$$\sup_{s\in[0,T]}\sum_{i=1}^d\Vert a_i(s,y)-a_i(s,z)\Vert < C_T\Vert y-z\Vert$$
for any $T>0$, some $C_T>0$ and any $y,z\in E$.
Then there is a unique $E$-valued c\`adl{\`a}g process $Y$ such that
\begin{eqnarray*}\label{E:SDE}
Y_t = x + \sum_{i=1}^d\int_0^ta_i(s,Y_{s-})dX^i_s, \quad t\ge0,
\end{eqnarray*}
where the right-hand side contains pathwise Bochner integrals.
This process is adapted to the filtration generated by $X$.
\end{theorem}
\begin{proof}
The proof uses the standard Picard iteration scheme (cf.\ \cite[Theorem I.1.1]{hartman.82}) to construct 
an adapted solution and Gr\"onwall's inequality (cf.\ \cite[Theorem A.5.1]{ethier.kurtz.86}) to show 
uniqueness among all possible solutions.
For ease of notation we assume that $d=1$.
The general case follows along the same lines.

Let $(\scr F_t)_{t\geq0}$ be the right-continuous filtration generated by $X$. 
W.l.o.g.\ we may assume that $X_0=0$. Let $g:\rp\to E$ be c\`adl\`ag.
For any $T\in\mathbb R_+$ and $s\leq T$ we have
\begin{eqnarray*}
\Vert a(s,g(s-))\Vert &\leq& C_T \Vert g(s-)\Vert+ \Vert a(s,0)\Vert\\
 &\leq & \sup_{s\in[0,T]} \left(C_T \Vert g(s)\Vert + \Vert a(s,0)\Vert\right)\\
 &<& \infty.
\end{eqnarray*} 
Hence the measurable function
$f:\mathbb R_+\rightarrow E$, $s\mapsto a(s,g(s-))$ is bounded on any $[0,T]$.
Thus $f$ is integrable on compact sets with respect to any finite Borel measure, 
e.g.\ with respect to the Lebesgue-Stieltjes measure of $X$. 
In particular, if $V$ is an adapted $E$-valued c{\`a}dl{\`a}g process, then the pathwise integrals
\begin{eqnarray*}
\Phi(V)_{t} &:=& x + \int_0^t a(s,V_{s-}) dX_s,\\
\Phi^\tau(V)_t &:=& x + \int_0^t 1_{[\![0,\tau]\!]}(s)a(s,V_{s-}) dX_s,\\
\Phi^{\tau-}(V)_t &:=& x + \int_0^t 1_{[\![0,\tau[\![}(s)a(s,V_{s-}) dX_s
\end{eqnarray*}
are c{\`a}dl{\`a}g adapted process for any stopping time $\tau$, cf.\ Lemma \ref{l:vorweg}. 
Observe that if $V$ is a fixed point of $\Phi^{\tau-}$, then $W:=\Phi^{\tau}(V)$ is a fixed point of $\Phi^{\tau}$.

Let $T\in\mathbb N$. \cite[I.1.18 and I.1.28]{js.87} yield that
$$\tau_n:=T\wedge \inf\left\{t\geq0: X_t\geq \frac{n}{2C_T}\right\}$$
is a stopping time for any $n\in\mathbb N$. 
Moreover, $\tau_0=0$, $\tau_n\leq\tau_{n+1}$, $\tau_{n}= T$ for large $n$ and we have 
$$1_{\{\tau_n\neq T\}}C_T\vert X_{\tau_{n+1}-}-X_{\tau_n}\vert \leq \frac{1}{2}$$
for any $n\in\mathbb N$. Let $n\in\mathbb N$ such that there is a fixed point $V^{n}$ of $\Phi^{\tau_n}$. 
For any $E$-valued process $W$ define the random variable 
$\Vert W\Vert_\infty:=\sup_{s\leq T}\Vert W_s\Vert$ and let $\scr V$ 
be the set of adapted c{\`a}dl{\`a}g processes which coincide with $V^n$ until $\tau_n$. 
Observe that $\Phi^{\tau_{n+1}}(\scr V)\subset\scr V$
because it maps adapted c{\`a}dl{\`a}g processes to adapted c{\`a}dl{\`a}g processes 
and $V^n$ is a fixed point of $\Phi^{\tau_n}$. For $W^{1},W^2\in\scr V$ we have
\begin{eqnarray*}
\Vert \Phi^{\tau_{n+1}-}(W^1)_t-\Phi^{\tau_{n+1}-}(W^2)_t \Vert 
&\leq& \int_0^{t} 1_{[\![0,\tau_{n+1}[\![}(s)
\Vert a(s,W^1_{s-}) - a(s,W^2_{s-}) \Vert d X_s\\
 &\leq& \int_{0}^{t} 1_{[\![\tau_n,\tau_{n+1}[\![}(s)\Vert a(s,W^1_{s-}) - a(s,W^2_{s-}) \Vert dX_s\\
 &\leq & 1_{\{\tau_n\neq T\}}C_T \Vert W^1-W^2 \Vert_\infty\vert X_{\tau_{n+1}-}-X_{\tau_n}\vert \\
 &< & \frac{1}{2}\Vert W^1-W^2 \Vert_\infty
\end{eqnarray*}
for any $t\in[0,T]$ and thus we obtain
$$\Vert \Phi^{\tau_{n+1}-}(W^1)-\Phi^{\tau_{n+1}-}(W^2) \Vert_\infty \leq \frac{1}{2}\Vert W^1-W^2 \Vert_\infty.$$
Thus $\Phi^{\tau_{n+1}-}$ is a contraction on the set $\scr V$, which implies that
there is a fixed point $W\in\scr V$ 
of $\Phi^{\tau_{n+1}-}$. In particular, $\Phi^{\tau_{n+1}}$ has an adapted c\`adl\`ag fixed point as well, namely 
$V^{n+1}:=\Phi^{\tau_{n+1}}(W)$. Define the adapted c{\`a}dl{\`a}g process
$U^T_t :=\lim_{n\rightarrow\infty} V^n_t$.
This process $U^T$ is a fixed point of $\Phi^T$. 
Its pointwise limit $Y_t:=\lim_{T\rightarrow\infty}U^T_t$ is an adapted c{\`a}dl{\`a}g process 
and a fixed point of $\Phi$.

Let $Z$ be another pathwise solution, i.e.\ $Z$ is $E$-valued, c{\`a}dl{\`a}g, and
$Z_t = x + \int_0^ta(s,Z_{s-})dX_s$, $t\geq0$. 
Fix $\omega\in\Omega$ and define 
$$f(t):=\sup_{s\in[0,t)}\Vert Y_s(\omega)-Z_s(\omega)\Vert$$ for all $t\in\mathbb R_+$. $f$ is finite
because $Y$ and $Z$ are c{\`a}dl{\`a}g. Moreover, we have
\begin{eqnarray*}
 f(t) &=& \sup_{s\in[0,t)}\Vert \Phi(Y)_s(\omega)-\Phi(Z)_s( \omega)\Vert \\
    &\leq& C_T \left(\int_0^t f(s) dX_s\right)(\omega)
\end{eqnarray*}
for $t\leq T$. \cite[Theorem A.5.1]{ethier.kurtz.86} yields $f(t)=0$ for all $t\geq0$. Thus $Z=Y$.

\end{proof}

If a growth condition holds, the Lipschitz condition can be relaxed as usual to a local version.
\begin{corollary}\label{C:Ex und Eind Satz1}
 Let $X=(X^1,\dots,X^d)$ be an $\mathbb R^d$-valued right-continuous process whose components
are nonnegative and increasing. Moreover, let $(E,\Vert\cdot\Vert)$ be a separable Banach space, 
$x\in E$ and
$a:\mathbb R_+\times E\rightarrow E^d$
measurable such that for any $T\in\rp$ there is $C_T<\infty$ such that 
$$\sup_{s\in[0,T]}\sum_{i=1}^d\Vert a_i(s,y)\Vert\leq C_T(1+\|y\|), \quad y\in E$$
and for any $K,T\in\rp$ some $C_{K,T}<\infty$ such that
$$\sup_{s\in[0,T]}\sum_{i=1}^d\Vert a_i(s,y)-a_i(s,z)\Vert < C_{K,T}\Vert y-z\Vert$$
for any  $y,z\in E$ with $\|y\|\le K$, $\|z\|\le K$.
Then there is a unique $E$-valued c\`adl{\`a}g process $Y$ such that
\begin{eqnarray}\label{E:SDE2}
Y_t = x + \sum_{i=1}^d\int_0^ta_i(s,Y_{s-})dX^i_s, \quad t\ge0.
\end{eqnarray}
This process is adapted to the filtration generated by $X$.
\end{corollary}
\begin{proof}
 Again we consider $d=1$ for the proof.
Fix $K\in\rp$. Choose $a^K:\mathbb R_+\times E\rightarrow E$ such that
$a^K(t,y)=a(t,y)$ for $\|y\|\leq K$ and such that the Lipschitz condition
of Theorem \ref{S:Ex und Eind Satz} holds for $a^K$,
e.g.\
$$a_K(t,y):=a\left(t,y{1+(\|y\|\wedge K)\over1+\|y\|}\right).$$
Denote the corresponding solution to the SDE by $Y^K$.
Gr\"onwall's inequality yields as in the proof of Theorem \ref{S:Ex und Eind Satz}
that for $K,L\in\rp$ the solutions $Y^K$ and $Y^L$ coincide till the norm of either of the two processes
exceeds $K\wedge L$.
Hence, setting
$Y_t=Y^K_t$ 
for $t\in\rp$ with $\sup_{s\in[0,t]}\|Y^K_s\|\leq K$
yields a well-defined adapted c\`adl\`ag process, which may however explode at a finite time.
However, the linear growth condition and Gr\"onwall's inequality yield that
$\sup_{t\in[0,T]}\|Y^K_t\|\leq c_T$ for some finite random variable $c_T$ which does 
not depend on $K$. 
Consequently, $Y$ is defined on $\rp$ and it solves SDE (\ref{E:SDE2}).
Uniqueness follows as in the proof of Theorem \ref{S:Ex und Eind Satz}.

\end{proof}

Corollary \ref{C:Ex und Eind Satz1} can be extended to 
Fr\'echet-valued processes and their respective integrals.
\begin{corollary}\label{C:Ex und Eind Satz}
 Let $X=(X^1,\dots,X^d)$ be an $\mathbb R^d$-valued right-continuous process whose components
are nonnegative and increasing. Moreover, let
$x\in F$, and
$a:\mathbb R_+\times F\rightarrow F^d$
measurable such that 
for any $T\in\rp$ there is $n_0\in\mathbb N$ such that 
for any $n\geq n_0$ there is $C<\infty$ with
$$\sup_{s\in[0,T]}\sum_{i=1}^d\Vert a_i(s,y)\Vert_n\leq C(1+\|y\|_n),\quad y\in F$$
and for any $n\geq n_0$, $K\in\rp$ there is some $C_K<\infty$ with
$$\sup_{s\in[0,T]}\sum_{i=1}^d\Vert a_i(s,y)-a_i(s,z)\Vert_n < C_K\Vert y-z\Vert_n$$
for $y,z\in F$ with $\|y\|_n\le K$, $\|z\|_n\le K$.
Then there is a unique $F$-valued c\`adl{\`a}g process $Y$ such that
\begin{eqnarray*}\label{E:SDE3}
Y_t = x + \sum_{i=1}^d\int_0^ta_i(s,Y_{s-})dX^i_s, \quad t\ge0,
\end{eqnarray*}
where the right-hand side refers to the pathwise Bochner integral introduced above.
This process is adapted to the filtration generated by $X$.
\end{corollary}
\begin{proof}
As before, we consider $d=1$ for ease of notation.

 {\em Existence for a modified equation.} Fix $T\in\rp$ and define
 $$\widetilde a:\mathbb R_+\times F\rightarrow F^d,\quad (t,x)\mapsto 1_{[0,T]}(t)a(t,x).$$
By assumption there is an $n_0\in\mathbb N$ such that for any $n\geq n_0$, 
$K\geq0$ there are $C,C_K<\infty$ such that
\begin{eqnarray}
 \sup_{s\in\rp}\Vert\widetilde a(s,y)\Vert_n &\leq& C(1+\|y\|_n),\quad y\in F, \label{e:growth} \\
 \sup_{s\in\rp}\Vert\widetilde a(s,x)-\widetilde a(s,z)\Vert_n &<& C_K\Vert x-z\Vert_n \label{e:Lipschitz}
\end{eqnarray}
for any $x,z\in F$ with $\|x\|_n\le K$, $\|z\|_n\le K$. Let $n\geq n_0$, $K\in \rp$ 
and choose $C,C_K>0$ such that (\ref{e:growth}, \ref{e:Lipschitz}) hold. 
Consider the factor space $G_n:=F/N_n$ with 
$N_n:=\{y\in F:\Vert y\Vert_n=0\}$ and the corresponding factor norm $\Vert\cdot\Vert_{c,n}$. 
Then $(G_n,\Vert\cdot\Vert_{c,n})$ is a separable Banach space. 
Let $\rho_n:F\rightarrow G_n$ be the factor mapping.
Observe that
 $$\Vert \rho_n(\widetilde a(s,x))-\rho_n(\widetilde a(s,y))\Vert_{c,n} 
 = \Vert\widetilde a(s,x)-\widetilde a(s,y)\Vert_n \leq C_K \Vert x-y\Vert_n =0$$
for any $s\le T$, $K\in\mathbb N$, $x,y\in F$ with $\|x\|_n\le K$, $\|y\|_n\le K$, and
$\Vert x-y\Vert_n=0$. Thus the function
 $$ b^{(n)}:\mathbb R_+\times G_n\rightarrow G_n,\quad (s,[x])\mapsto \rho_n(\widetilde a(s,x))$$
 is well defined. Moreover, $b^{(n)}$ is measurable,
$$\sup_{s\in\rp}\Vert b^{(n)}(s,y)\Vert_{c,n}<C(1+\|y\|_n),\quad y\in G_n$$
and
$$\sup_{s\in\rp}\Vert b^{(n)}(s,y)-b^{(n)}(s,z)\Vert_{c,n} < C_K\Vert y-z\Vert_{c,n}$$
for any $y,z\in G_n$ with $\|x\|_n\le K$, $\|y\|_n\le K$.
Then Corollary \ref{C:Ex und Eind Satz1} yields that there is a unique adapted c\`adl\`ag solution $\hat Y$ to the SDE
 $$\hat Y_t = \rho_n(x) +\int_0^tb^{(n)}(s,\hat Y_{s-})dX_s .$$
 Let $\varphi_n:G_n\rightarrow F$ be a linear choice of representants. 
Then $\varphi_n$ is isometric and hence 
$Y^{(n)}:=\varphi_n(\hat Y)$ is a version of 
$x + \int_0^t\widetilde a(s,Y^{(n)}_{s-})dX_s$, 
where we refer to the $\Vert\cdot\Vert_n$-Bochner integral. 
Lemma \ref{L:Integralkompatibilitaet} yields that $Y^{(n)}_{t}$ 
is a version of $x + \int_0^t\widetilde a(s,Y^{(n)}_{s-})dX_s$, 
where the integral can be interpreted as a $\Vert\cdot\Vert_k$-Bochner integral for $k=1,\dots,n$. 
Thus $Z:=\rho_k(Y^{(n)})$ is the unique solution to the SDE 
$$Z_t = \rho_k(x) + \int_0^tb^{(k)}(s,Z_{s-})dX_s$$
for $k=n_0,\dots,n$, $t\in\rp$. Hence $\Vert Y^{(n)}_t - Y^{(k)}_t \Vert_k =0$ 
for $k=n_0,\dots,n$, $t\in\rp$. 
Thus $(Y^{(n)}_t(\omega))_{n\geq n_0}$ is a 
$\Vert\cdot\Vert_k$-Cauchy sequence for any $k\geq n_0$ and any $\omega\in\Omega$. 
Hence it is a $d$-Cauchy sequence and we define
$$Y_t(\omega) := \lim_{n\rightarrow \infty} Y^{(n)}_t(\omega)$$
relative to the metric $d$ and any $\omega\in\Omega$, $t\in\rp$. Since $\Vert Y^{(k)}_t - Y_t\Vert_k=0$, we have
\begin{equation}\label{e:sde1}
 Y_{t} = x + \int_0^t\widetilde a(s,Y_{s-})dX_s
\end{equation}
for some version of the $\Vert\cdot\Vert_k$-Bochner integral.
 Since this is true for any $k\geq n_0$, we conclude that
the equation holds also relative to the Fr\'echet space-valued Bochner integral. 
The c\`adl\`ag property of $Y$ follows from the c\`adl\`ag properties of $Y^{(n)}$ relative to the semi-norm 
$\Vert\cdot\Vert_n$ for any $n\geq n_0$. 
Random variable $Y_t$ is $\scr F_t$-$\scr B_n$-measurable for any $t\in\rp$,  $n\geq n_0$ because 
this holds for $\hat Y_t=\rho_n(Y_t)$ by Corollary \ref{C:Ex und Eind Satz1}.
Lemma \ref{L:Topologien} yields that $Y$ is adapted.

 {\em Uniqueness for the modified equation.} 
 Let $Z$ be an arbitrary solution to (\ref{e:sde1}) in the sense of $(F,d)$-space valued integrals. 
 Let  $n\geq n_0$ and define the factor space $G_n$, 
the factor mapping $\rho_n$,  and the function $b^{(n)}$ as above. Then we have
 $$\rho_n(Z_t) = \rho_n(x) + \int_0^t b^{(n)}(s,\rho_n(Z_{s-}))dX_s$$
 for any $t\in\rp$. Corollary \ref{C:Ex und Eind Satz1} yields that $\rho_n\circ Z = \rho_n \circ Y$ or, 
equivalently, $\Vert Z_t-Y_t\Vert_n=0$ for any $t\in\rp$. Since this is true for any $n\geq n_0$, 
we have $Z=Y$.

 {\em Consistency of solutions:} Let $0\leq T_1\leq T_2<\infty$ and
\begin{eqnarray*}
 \widetilde a_1:\mathbb R_+\times F\rightarrow F^d,\quad (t,x)\mapsto 1_{[0,T_1]}(t)a(t,x), \\
 \widetilde a_2:\mathbb R_+\times F\rightarrow F^d,\quad (t,x)\mapsto 1_{[0,T_2]}(t)a(t,x).
\end{eqnarray*}
Let $Y^{(i)}$ be the unique $(F,d)$-valued solution to the SDE
 \begin{equation}\label{e:sde2}
    Y^{(i)}_{t} = x + \int_0^t\widetilde a_i(s,Y^{(i)}_{s-})dX_s,\quad t\in\rp
 \end{equation}
for $i=1,2$, which has been constructed above. Then the stopped process
$(\widetilde Y^{(2)})^{T_1}$
solves the $(F,d)$-valued SDE (\ref{e:sde2})
for $i=1$, which implies $\widetilde Y^{(1)}=(\widetilde Y^{(2)})^{T_1}$.

 {\em Existence and uniqueness:} For any $T\in\mathbb N$ let $Y^{(T)}$ be the unique 
 $(F,d)$-valued solution to the SDE
\begin{equation}\label{e:sde3}
 Y^{(T)}_{t} = x + \int_0^t1_{[0,T]}(t) a(s,Y^{(T)}_{s-})dX_s,\quad t\in\rp.
\end{equation}
 Define
 $Y_t:=Y^{(T)}_t$
 for any $T\in \mathbb N$, $ t\leq T$.
 Process $Y$ is well defined due to consistency of solutions. Observe that $Y$ is a solution to the $(F,d)$-valued SDE
 $$ Y_t = x + \int_0^t a(s,Y_{s-})dX_s,\quad t\in\rp.$$
 If $Z$ is any solution to this $(F,d)$-valued SDE, 
 then $Z^T$ solves (\ref{e:sde3}) for any $T\in\mathbb N$.
 This yields $Z_t=Y^{(T)}_t=Y_t$, $t\in[0,T]$.

\end{proof}

\section*{acknowledgements}
We would like to thank the anonymous referees and the associate editor
for various useful comments.

\bibliographystyle{spmpsci}      
\bibliography{HJM}   

\end{document}